\documentclass[11pt,reqno]{amsart}
\usepackage{amssymb}
\usepackage{amsfonts}
\usepackage{amsmath}
\usepackage{stmaryrd}
\usepackage{physics}
\usepackage{braket}
\usepackage{dsfont}
\usepackage{bbold}
\usepackage{graphicx}
\usepackage{relsize}
\usepackage{makecell}
\usepackage[table,xcdraw]{xcolor}
\usepackage{enumerate}
\usepackage[pagebackref, colorlinks = true, linkcolor = blue, urlcolor  = blue, citecolor = red]{hyperref}
\usepackage[margin=1in]{geometry}
\usepackage{enumitem}
\usepackage{mathtools}
\usepackage{graphbox}
\usepackage{comment}
\usepackage{float}
\usepackage[font=scriptsize]{caption}

\usepackage[capitalize]{cleveref}

\renewcommand{\epsilon}{\varepsilon}

\newcommand{\M}[1]{\mathcal{M}_{#1}(\mathbb{C})}

\newtheorem{theorem}{Theorem}[section]
\newtheorem{proposition}[theorem]{Proposition}
\newtheorem{corollary}[theorem]{Corollary}
\newtheorem{lemma}[theorem]{Lemma}
\newtheorem{conjecture}[theorem]{Conjecture}
\newtheorem*{conjecture*}{Conjecture}

\theoremstyle{definition}

\theoremstyle{definition}
\newtheorem{remark}[theorem]{Remark}
\theoremstyle{definition}
\newtheorem{example}[theorem]{Example}
\theoremstyle{definition}
\newtheorem{question}
[theorem]{Question}

\newcommand\vertarrowbox[3][6ex]{%
  \begin{array}[t]{@{}c@{}} #2 \\
  \left\uparrow\vcenter{\hrule height #1}\right.\kern-\nulldelimiterspace\\
  \makebox[0pt]{\scriptsize#3}
  \end{array}%
}

\definecolor{darkgreen}{rgb}{0,0.392,0}

\newcommand{\la}{\langle}
\newcommand{\ra}{\rangle}

\newcommand{\A}{\mathcal{A}}
\newcommand{\B}{\mathcal{B}}
\newcommand{\D}{\mathcal{D}}
\newcommand{\E}{\mathbb{E}}

\newcommand{\Le}{\mathcal{L}}

\newcommand{\T}{\mathcal{T}}

\newcommand{\Real}{\mathbb{R}}
\newcommand{\Comp}{\mathbb{C}}

\newcommand{\id}{\mathrm{id}}
\newcommand{\Om}{\Omega}
\newcommand{\om}{\omega}
\newcommand{\eps}{\varepsilon}

\newcommand{\sch}{\mathcal{SN}}
\newcommand{\SR}{{\rm SR}}
\newcommand{\SN}{{\rm SN}}
\newcommand{\SEP}{\mathcal{SEP}}

\newcommand{\PPT}{\mathcal{PPT}}

\newcommand{\Inv}{{\rm Inv}}
\newcommand{\Cov}{{\rm Cov}}

\newcommand{\Ad}{{\rm Ad}}

\author{Sang-Jun Park}
\email{sjpark@whu.edu.cn}
\address{School of Mathematics and Statistics, Wuhan University, 
430070, Wuhan, Hubei, China}

\title[PPT entanglement under symplectic group symmetries]{$k$-Positivity and high-dimensional bound entanglement under symplectic group symmetries}

\begin{document}
\begin{abstract}
We investigate the structure of $k$-positivity and Schmidt numbers for classes of linear maps and bipartite quantum states exhibiting symplectic group symmetry. Specifically, we consider (1) linear maps on $M_d(\Comp)$ which are covariant under conjugation by unitary symplectic matrices $S$, and  (2) $d\otimes d$ bipartite states which are invariant under $S\otimes S$ or $S\otimes \overline{S}$ actions, each parametrized by two real variables. We provide a complete characterization of all $k$-positivity and decomposability conditions for these maps and explicitly compute the Schmidt numbers for the corresponding bipartite states. In particular, our analysis yields a broad class of PPT states with Schmidt number $d/2$ and the first explicit constructions of (optimal) $k$-positive indecomposable linear maps for arbitrary $k=1,\ldots, d/2-1$, achieving the best-known bounds. Overall, our results offer a natural and analytically tractable framework in which both strong forms of positive indecomposability and high degrees of PPT entanglement can be studied systematically.

We present two further applications of symplectic group symmetry. First, we show that the PPT-squared conjecture holds within the class of PPT linear maps that are either symplectic covariant or conjugate symplectic covariant. Second, we resolve a conjecture of Pál and Vértesi concerning the optimal lower bound of the Sindici–Piani semidefinite program for PPT entanglement.
\end{abstract}

\maketitle

\tableofcontents

\section{Introduction}

The theory of positive linear maps plays a central role in operator algebras and quantum information theory \cite{stormer1963positive,Cho75a,Wor76,paulsen2002completely,TextSto,TextWat18}. Among the various refinements of positivity, the notion of 
\emph{$k$-positivity} has emerged as a fundamental yet notoriously subtle concept. A linear map $\Le:M_{d_A}(\Comp)\to M_{d_B}(\Comp)$ between matrix algebras is said to be 
{$k$-positive} if its ampliation map $\id_k\otimes \Le:M_k(\Comp)\otimes M_{d_A}(\Comp)\to M_k(\Comp)\otimes M_{d_B}(\Comp)$ maps positive semidefinite operators to positive semidefinite operators. While complete positivity ($k=\min(d_A,d_B)$) admits several powerful characterizations, including the Choi–Kraus representation, the structure of $k$-positive maps for intermediate values of $k$
remains poorly understood. In fact, even deciding ($1$-)positivity is computationally hard \cite{Gur03,GHP10}, and several semidefinite programming relaxations for testing $k$-positivity have been proposed recently \cite{CCF25}. Nevertheless, no general structural classification is known beyond very special cases.

In quantum information theory, $k$-positivity acquires an operational meaning through its intimate connection with quantum entanglement \cite{HHHH09}. Via the Jamiołkowski–Choi isomorphism \cite{Jam72,Cho75a,SSZ09}, $k$-positive maps correspond to entanglement witnesses capable of detecting bipartite quantum states whose \textit{Schmidt number} exceeds $k$ \cite{SBL01}. Recall that the Schmidt number \cite{TH00} of a bipartite mixed state is the minimal Schmidt rank required to prepare the state as a convex combination of pure states, and thus provides a natural hierarchy refining the dichotomy between separable and entangled states. From this perspective, the problem of characterizing 
$k$-positive maps is \textit{dual} to that of determining Schmidt number regions of quantum states, and hence is a task which is also known to be highly nontrivial both analytically and computationally.

A recent important open problem in entanglement theory concerns the maximal Schmidt number attainable by bipartite quantum states with \emph{positive partial transpose} (PPT) \cite{HLLMH18}. While the PPT criterion \cite{Per96,HHH96} provides an efficiently checkable necessary condition for separability and is in fact exact in low dimensions \cite{Wor76}, it becomes increasingly coarse in higher-dimensional systems, where PPT entangled states are known to exist \cite{Cho82,Hor97}. All such states are necessarily \emph{bound entangled}, meaning that their entanglement cannot be distilled by local operations and classical communication \cite{HHH98}. In contrast, the Schmidt number offers a finer and operationally meaningful quantifier of entanglement, as it characterizes the minimal local dimension required to prepare a given quantum state \cite{BD11}. Futhermore, there has been several recent works discussing the use of high-dimensional entanglement in various quantum information processing \cite{LHF25,MGM24+}. Demonstrating the existence of PPT states with large Schmidt number shows that bound entanglement can be intrinsically high-dimensional, rather than merely a weak remnant of separability. Moreover, such constructions expose fundamental limitations of PPT-based entanglement detection and provide structured test cases for entanglement manipulation and for long-standing open problems, such as the \emph{PPT squared conjecture} \cite{PPTsq,CMHW19}.

In the case of the $d\otimes d$ system, evidence for the existence of PPT states whose Schmidt number scales linearly in $d$ was first obtained in \cite{SWZ11} using tools from asymptotic geometric analysis, although explicit constructions remained elusive. 
Early explicit examples exhibited PPT states whose Schmidt number grows logarithmically in 
$d$ \cite{CYT17}, while subsequent works improved this bound to linear scaling in $d$ \cite{HLLMH18,PV19,Car20}. At present, the best known general upper bound is 
$\lfloor d/2\rfloor$, except for the dimensions $d=3,5,7$ and $9$ \cite{Cho82,Hor97,KG24,KG25+}, and explicit constructions achieving Schmidt number of arbitrary order 
$d/2$ has been known only in limited settings \cite{PV19,Car20}. The well-known dual formulation of this problem asks for the existence of 
$k$-positive \textit{indecomposable} (i.e., not of the form $\Le'+\top \circ \Le''$ for some completely positive maps $\Le'$ and $\Le''$) linear maps on 
$M_d(\Comp)$ as large as possible, a question originally raised in \cite{Ter01}. This dual viewpoint has also driven much of the development of the theory of $k$-positive maps and its Choi matrices \cite{CK09,YLT16,HLLMH18,BO20,MOM25,HK25a,HK25b,EKD25}.

On the other hand, one of the most effective strategies in quantum information theory is to impose \emph{symmetry} with respect to group representations \cite{Wer89,holevo1993note,VW01,MSD17,LY22}. Quantum states that are invariant under group actions, as well as linear maps that are covariant with respect to such actions, form tractable yet nontrivial subclasses that often admit complete analytic descriptions. This symmetry-based approach has proved extremely fruitful in entanglement theory, where invariance under classical matrix group actions has led to explicit classifications of entanglement, PPT, and separability properties \cite{Wer89,HH99,VW01,EW01,CK06,TG09,ES13,Yu16,SN21,PJPY24,GNS25,GNP25+}.

In this work, we contribute to both of the above questions on $k$-positivity and Schmidt numbers by exploring new classes of quantum objects possessing symmetries with respect to the \emph{symplectic group}. More precisely, we study linear maps that are covariant under conjugation by unitary symplectic matrices $S$, as well as bipartite quantum states that are invariant under $S\otimes S$ or $S\otimes \overline{S}$ actions, each parametrized by two real variables. The representation theory of the symplectic group gives rise to algebraic structures that are formally similar to those arising in symmetry classes associated with the orthogonal group \cite{VW01,PY24}; however, the resulting entanglement properties turn out to be drastically different. As we show, this symmetry setting supports a remarkably rich entanglement landscape, including highly entangled PPT states and highly positive yet indecomposable linear maps.

Our first main result is summarized as follows.

\begin{theorem} \label{thm-main}
Let $d\geq 4$ be an even integer and let $V$ be a skew-symmetric unitary matrix, i.e., $V\in U(d)$ with $V^{\top}=-V$.
\begin{enumerate}
    \item All the $k$--positivity conditions of the linear map
        $$\Le_{p,q}^{(d)}:Z\mapsto \frac{1-p-q}{d}\Tr(Z)I_d+pZ+q V Z^{\top}V^*, \quad Z\in M_d(\Comp),$$
    are completely characterized by linear and quadratic inequalities in the parameters $p$ and $q$. Moreover, there exists a (triangular) set $R=R^{(d)}\subset \Real^2$ with nonempty interior such that $\Le_{p,q}^{(d)}$ is `\emph{$\big(\frac{d}{2}-1\big)$-positive and indecomposable}' if and only if $(p,q)\in R$, whereas $\Le_{p,q}^{(d)}$ becomes decomposable whenever it is $\frac{d}{2}$-positive.
    
    \item The Schmidt number of the bipartite quantum state
        $$\rho_{a,b}^{(d)}=\frac{1-a-b}{d^2}I_d\otimes I_d+ \frac{a}{d} \sum_{i,j=1}^d |ii\ra\la jj|+\frac{b}{d}(I_d\otimes V) \sum_{i,j=1}^d |ij\ra\la ji| (I_d\otimes V)^{*}$$
    in a $d\otimes d$ system, whenever it is well defined, is computed explicitly for all values of $(a,b)$. Furthermore, there exists a (triangular) set $S=S^{(d)}\subset \Real^2$ with nonempty interior such that $\rho_{a,b}^{(d)}$ is `\emph{PPT with the Schmidt number $d/2$}' if and only if $(a,b)\in S$, while $\SN \big(\rho_{a,b}^{(d)}\big)\leq d/2$ whenever $\rho_{a,b}^{(d)}$ is PPT.
\end{enumerate}
\end{theorem}
We refer to \cref{thm-Symp-PosDec,thm-SympkPos,thm-SympSepPPT,thm-SympSch} for the precise regions of positivity and Schmidt numbers, and to Corollaries \ref{cor-Symp-kposIndec} and \ref{cor-Symp-PPTSch} for the exact descriptions of the sets $R$ and $S$. {Both the $k$-positivity and Schmidt number regions for our classes are, in fact, independent of the choice of $V$ (Remark \ref{rmk-V-equiv}), yet remain highly nontrivial even for the minimal dimension $d=4$. For illustration, we depict the corresponding regions in Fig.~\ref{fig-4kPos}.} 

\begin{figure}
    \centering

    \includegraphics[width=0.51\linewidth]{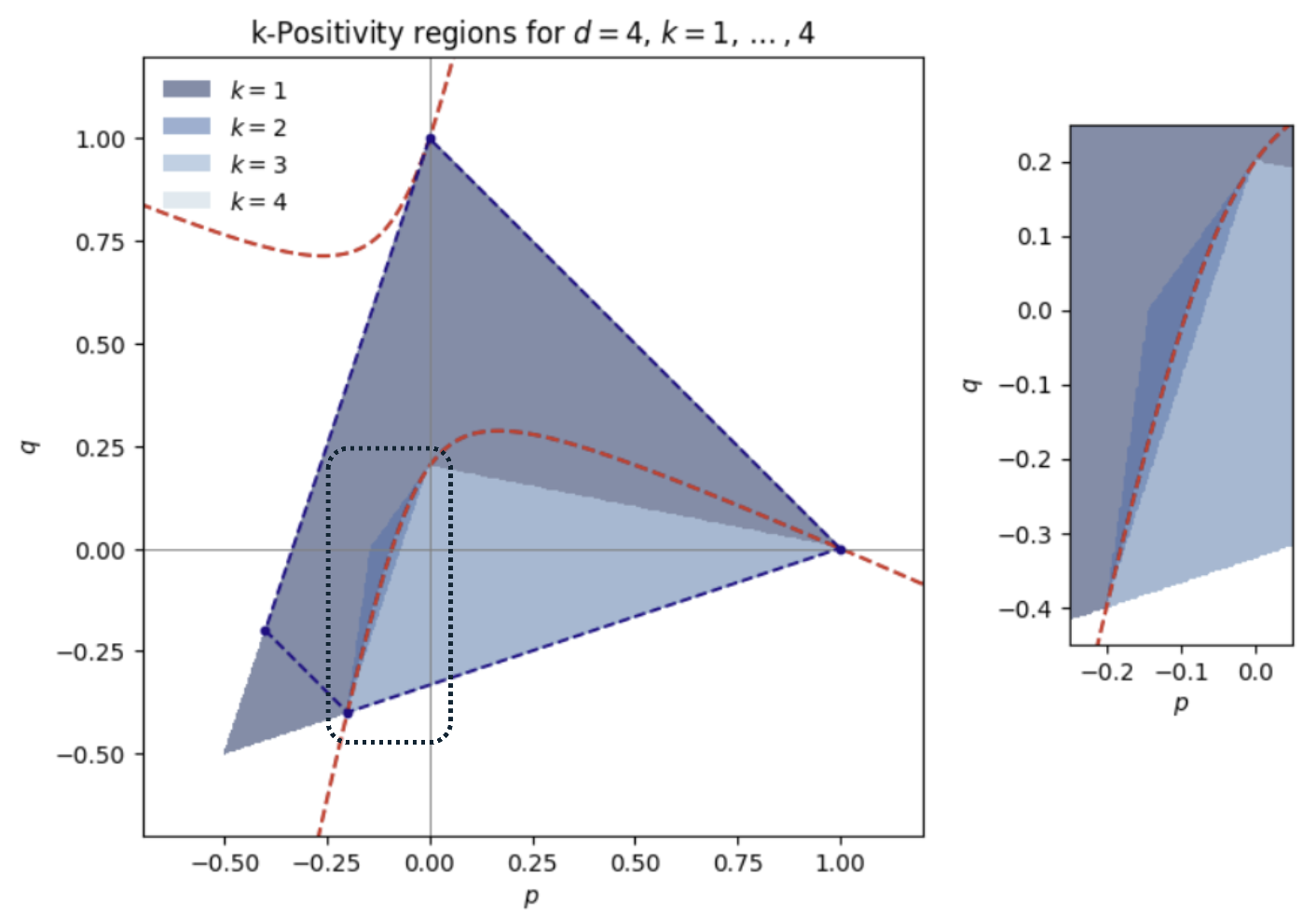} 
    \includegraphics[width=0.35\linewidth]{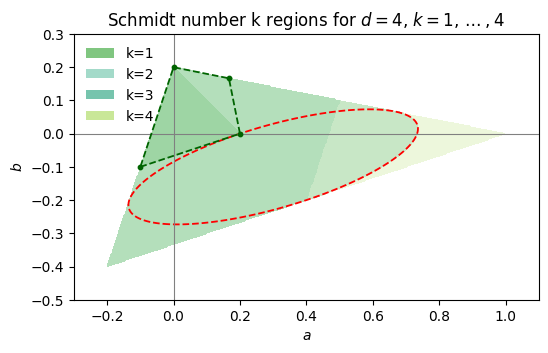}
    
    \caption{$k$-Positivity (left) and Schmidt number (right) regions for $d=4$. The red dashed hyperbola (left) and ellipse (right) partially determine the boundaries for $k=3$. Furthermore, the dark blue quadrilateral (left) and dark green quadrilateral (right) indicate the decomposability and PPT regions, resp.}
    \label{fig-4kPos}
\end{figure}

Here we highlight several new consequences of these results:
\begin{itemize}
    \item We provide a \emph{broad constructive family} of $(d/2-1)$-positive indecomposable linear maps and PPT states with Schmidt number $d/2$ (both within the class of symplectic group symmetry and beyond; see Corollaries \ref{cor-PPTSchOpt1} and \ref{cor-PPTSchOpt2}), thus achieving the best-known bounds. To the best of our knowledge, this is the first explicit construction of indecomposable linear maps attaining the currently best known positivity degree $d/2-1$; while the existence of such maps was previously known indirectly via duality arguments involving Schmidt numbers, explicit realizations had remained elusive.  Moreover, we show that these values of $k$ are \emph{optimal} within these classes of linear maps and quantum states.  

    \item When $p=q=-\frac{1}{d-2}$, the linear map $\Le_{p,q}^{(d)}$ corresponds to the widely studied \textit{Breuer-Hall map} \cite{Bre06,Hal06} $\Le^{\rm BH}$, which has proven effective in detecting PPT entanglement. It follows, however, that $\Le^{\rm BH}$ fails to be \emph{$2$-positive}; consequently, it cannot be used as an entanglement witness tailored solely to states with Schmidt number greater than $2$. In fact, $\Le^{\rm BH}$ is shown to be \emph{atomic}, that is, it is not of the form $\Le^{\rm BH}=\Le'+\top\circ \Le''$ for some $2$-positive linear maps $\Le'$ and $\Le''$.
    
    \item When $a=b=\frac{1}{d+2}$, the state $\rho_{a,b}^{(d)}$ is PPT, and its partial transpose $(\rho_{a,b}^{(d)})^{\Gamma}$ corresponds to the \textit{P\'al--V\'ertesi state} \cite{PV19} $\rho^{\rm PV}$, previously shown to have Schmidt number \emph{at least $d/2$}. Our results further establish that $\rho^{\rm PV}$ has Schmidt number \emph{exactly $d/2$}. 
    Moreover, we verify that $\rho^{\rm PV}$ is local-unitary equivalent to its partial transpose $(\rho^{\rm PV})^{\Gamma} = \rho_{\frac{1}{d+2},\frac{1}{d+2}}^{(d)}$, demonstrating that even bipartite states local-unitary equivalent to their partial transpose can attain a high Schmidt number. In fact, an infinite family of such states can be constructed within our class (Corollary \ref{cor-PPT-SNDiff} (2)).

    \item As another interesting example, we construct a broad class of quantum states $\rho=\rho_{a,b}^{(d)}$ exhibiting a large Schmidt number gap between the state and its partial transpose (Corollary \ref{cor-PPT-SNDiff} (1)):
    \begin{align*}
        &{\SN}(\rho)-{\rm SN}(\rho^{\Gamma})=d/2-2\\
        &= \max\{{\SN}(\rho_{a',b'}):\rho_{a',b'} \text{ is PPT}\} - \min\{\rm SN(\rho'):\rho' \text{ is entangled}\}.
    \end{align*}
    Questions of this type were raised in \cite{CYT17} and first addressed in \cite{HLLMH18}, where the gap was of order $d/4$; our construction improves this bound within a natural and highly symmetric class of states.
\end{itemize}

In particular, we emphasize a new construction of Breuer–Hall–type maps, which are shown to be $k$-positive and indecomposable in stronger sense, and thus remedy a limitation of the original Breuer–Hall map.
\begin{theorem} [Theorem \ref{thm-Symp-opt}, Remark \ref{rmk-SympOptAtomic}]
For $d\geq 4$ an even integer and for $V$ a skew-symmetric unitary matrix, we consider the linear maps
    $$\Le_k^{\rm BH}:Z\mapsto \frac{1}{kd-k-1} (k\Tr(Z) I_d -Z-kVZ^{\top} V^*), \quad Z\in M_d(\Comp).$$
Then for $k=1,\ldots, d/2-1$,
\begin{enumerate}
    \item $\Le_k^{\rm BH}$ is an \emph{optimal $k$-positive indecomposable linear map}: there exists no other $k$-positive witness that is able to detect strictly more PPT states with Schmidt number larger than $k$;

    \item $\Le^{\rm BH}_k$ cannot be decomposed as $\Le'+\top \circ \Le''$ for any $(k+1)$-positive linear maps $\Le'$ and $\Le''$.
\end{enumerate}
\end{theorem}

In this sense, we may refer to $\Le_k^{\rm BH}$ as the \emph{$k$-Breuer--Hall maps}. 
We remark that the construction of $k$-positive linear maps possessing either of the above two properties has remained largely out of reach beyond the case $k=1$; here, however, the group-symmetry method leads to a substantial simplification of the argument. 
Furthermore, from these linear maps we derive general Schmidt number criteria (Proposition \ref{prop-Sch-general}) and obtain additional examples of PPT states with high Schmidt number (Corollaries \ref{cor-PPTSchOpt1} and \ref{cor-PPTSchOpt2}).

As highlighted by our results, there is a sharp contrast between symplectic and \emph{orthogonal group symmetries}. In the orthogonal case, the corresponding $(O,O)$-covariant linear maps and $O\otimes O$-invariant bipartite quantum states are given by:
\begin{align*}
    \Le_{p,q}^{(OO)}(Z) &= \frac{1-p-q}{d} \Tr(Z) I_{d} + pZ + qZ^{\top},\\
    \rho_{a,b}^{(OO)} &= \frac{1-a-b}{d^2} I_{d^2} + \frac{a}{d}\sum_{i,j=1}^d |ii\ra\la jj| + \frac{b}{d} \sum_{i,j=1}^d |ij\ra\la ji|,
\end{align*}
which are similar to those in \cref{thm-main} but \textit{without the twisting unitary $V$}. Within these classes, it is known that positivity implies decomposability, and PPT implies separability \cite{VW01,PJPY24}. This contrast underscores the symplectic group as a natural setting in which both strong forms of positive indecomposability and high degrees of PPT entanglement can be studied systematically.

\medskip

We have two additional applications using our framework. In Section \ref{sec-PPTSq}, we show that the \textit{PPT squared conjecture} \cite{PPTsq,CMHW19} holds within symplectic group symmetry in two perspectives: (1) any composition of two PPT maps is entanglement-breaking, and (2) any composition of a positive map and a PPT map is decomposable. This result is in fact a simple consequence of Theorem \ref{thm-main}. Nevertheless, this is still remarkable in the sense that the verification applies to a broad class of quantum channels whose Choi matrices may exhibit high-dimensionality of entanglement, thereby providing new and nontrivial evidence in support of the conjecture. 

\medskip

In Section \ref{sec-SP-SDP}, we revisit the work of \cite{PV19} on the semidefinite program introduced by Sindici and Piani \cite{SP18} for PPT entanglement, defined as follows. For a $d\otimes d$ antisymmetric state $\rho_{\A}$,
\begin{center}
    $p^{\rm PPT}(\rho_\A):=\max_{\sigma} 
    \Tr(\Pi_\A\sigma)$ \quad subject to \quad $\begin{cases}
        \text{$\sigma$ is a PPT state}, \\ \text{$\Pi_\A\sigma \Pi_\A=\Tr(\Pi_\A\sigma) \rho_\A$,}
    \end{cases}$
\end{center}
where $\Pi_\A$ denotes the projection onto the antisymmetric space $\Comp^d\wedge \Comp^d$. It was shown in \cite{PV19} that, if $p^{\rm PPT}(\rho_A)<\frac{1}{2r+2}$ for some integer $r\geq 2$, then any optimal PPT state $\sigma^*$ for $p^{\rm PPT}(\rho_{\A})$ satisfies $\SN(\sigma^*)\geq r+1$. Furthermore, it was conjectured  that the bound $\frac{1}{2r+2}$ cannot be improved. We resolve this conjecture by showing that the minimum value of $p^{\rm PPT}$ equals $\frac{1}{d+2}$ whenever $d \ge 4$ is even.

%-----------------------

\section{Preliminaries}

\subsection{Schmidt number and positive linear maps} \label{sec-notations}

Throughout this paper, we use Dirac's \textit{bra-ket} notations: column vectors $v \in \mathbb{C}^d$ are written as kets $\ket{v}$ and their conjugate transpose $v^* \in (\Comp^d)^*$ are written as bras $\bra{v}$. We distinguish these from the \textit{entrywise complex-conjugated vectors} $\overline{v}\in \Comp^d$, noting that $v^*=(\overline{v})^{\top}$. The standard \emph{inner product} is denoted by $\langle v | w \rangle:=v^*w\in \Comp$ and the \emph{outer product} (rank-one operator) by $\ketbra{v}{w}:=vw^*\in M_d(\Comp)$. The standard basis in $\Comp^d$ is denoted by $\{\ket{j}:=e_j\}_{j=1}^d$, and we write $\ket{ij} := \ket{i} \otimes \ket{j}$ to denote tensor products of standard basis vectors. Finally, we denote by $\D(\Comp^d)$ the set of \textit{quantum states}, i.e., a positive semidefinite matrix $\rho\in M_d(\Comp)^+$ of unit trace.

Let $H_A=\Comp^{d_A}$ and $H_B=\Comp^{d_B}$ be finite-dimensional Hilbert spaces. Any nonzero bipartite vector $\xi\in H_{AB}:=H_A\otimes H_B$ admits a \textit{Schmidt decomposition} \cite{NiCh} $|\xi\ra=\sum_{i=1}^k{\lambda_i}|v_i\ra\otimes |w_i\ra$ where $\lambda_1\geq \cdots \geq \lambda_k>0$, and $\{v_i\}_{i=1}^k$ and $\{w_i\}_{i=1}^k$ are orthonormal subsets in $H_A$ and $H_B$, respectively. The integer $k$ and the coefficients $\{\lambda_i\}_{i=1}^k$ are uniquely determined, and we call $k$ the \textit{Schmidt rank} of $\xi$ and write ${\rm SR}(|\xi\ra)=k$. For any natural number $k$, let us define
    $$\sch_k(d_A\otimes d_B):={\rm conv}\{|\xi\ra\la \xi|: \xi\in H_{AB}, \;\;\| \xi\|_2=1,\;\; {\rm SR}(|\xi\ra)\leq k\}\subset \D(H_{AB})$$
and write simply $\sch_k$ when no confusion arises. Then the \textit{Schmidt number} of a quantum state $\rho_{AB}\in \D(H_{AB})$ is defined as the smallest natural number $k$ such that $\rho_{AB}\in \sch_k$, and we denote this by ${\rm SN}(\rho_{AB})=k$. By definition, $\sch_k=\{\rho\in \D(H_{AB}):SN(\rho)\leq k\},$
and therefore, we have $\sch_1\subset \sch_2\subset \sch_3\subset\cdots$. Note that $1\leq {\rm SN}(\rho_{AB})\leq \min(d_A,d_B)$ and $\sch_k=\D(H_{AB})$ for all $k\geq \min(d_A,d_B)$. Furthermore, $\sch_1(d_A\otimes d_B)$ coincides with the set $\SEP(d_A\otimes d_B)$ of \textit{separable} states in a $d_A\otimes d_B$ system; every state outside the set $\SEP$ is called \textit{entangled}. Another important class of quantum states is the set $\PPT(d_A\otimes d_B)$ of \textit{positive-partial-trasnpose (PPT)} quantum states, i.e., $\rho_{AB}\in \D(H_{AB})$ satisfying
    $$\rho_{AB}^{\Gamma}:=(\id_A\otimes \top_B)(\rho_{AB})\geq 0.$$
A celebrated Horodecki-Perez criterion \cite{HHH96,Per96} is reformulated as $\SEP\subset \PPT$, providing a simple yet still powerful necessary condition for separability.

A linear map $\Le:M_{d_A}(\Comp)\to M_{d_B}(\Comp)$ is called \emph{positive} if $\Le(M_{d_A}(\Comp)^+)\subseteq M_{d_B}(\Comp)^+$. For an integer $k\geq 1$, $\Le$ is called \emph{$k$-positive} if the ampliation map $\id_k\otimes \Le:M_k(\Comp)\otimes M_{d_A}(\Comp)\to M_k(\Comp)\otimes M_{d_B}(\Comp)$ is positive. Furthermore, $\Le$ is called \emph{completely positive} or shortly \emph{CP} (resp. \emph{completely copositive; co-CP}) if $\Le$ (resp. $\top\circ \Le$) is $k$-positive for all $k\geq 1$, where $\top:Z\mapsto Z^{\top}$ denotes the transposition map. Finally, $\Le$ is said to be \emph{decomposable} if $\Le$ is a sum of a CP and a co-CP linear map. It is well-known that the complete positivity of $\Le$ is equivalent to the $\min(d_A,d_B)$-positivity \cite{Cho75a}. 

The positivity properties are often efficiently analyized via the {Choi–Jamiołkowski isomorphism} \cite{Jam72,Cho75a}. A (normalized) \emph{Choi matrix} of a linear map $\Le:M_{d_A}(\Comp)\to M_{d_B}(\Comp)$ is defined by
    $$C_{\Le}:= (\id_{d_A}\otimes \Le)(|\om_{d_A}\ra\la \om_{d_A}|)=\frac{1}{d_A} \sum_{i,j=1}^{d_A} |i\ra\la j|\otimes \Le(|i\ra\la j|)\in M_{d_A}(\Comp)\otimes M_{d_B}(\Comp),$$
where $|\om_{d_A}\ra:=\frac{1}{\sqrt{d_A}}\sum_{j=1}^{d_A} |jj\ra$ is the \textit{maximally entangled vector} in $\Comp^{d_A}\otimes \Comp^{d_A}$. Arguably one of the most important results is that $\Le$ is completeley positive if and only if $C_{\Le}$ is positive semidefinite \cite{Cho75a}. Concerning the Choi matrices of $k$-positive linear maps, let us recall the following fundamental duality between $k$-positivity and Schmidt number, established in \cite{TH00,SSZ09}

\begin{theorem} \label{thm-dual-schmidt}
For a bipartite quantum state $\rho_{AB}\in \D(\Comp^{d_A}\otimes \Comp^{d_B})$, the following are equivalent:
\begin{enumerate}
    \item $\rho_{AB}\in \sch_k$, i.e., $\SN(\rho_{AB})\leq k$,

    \item $(\id_{d_A}\otimes \Le')(\rho_{AB})$ for all $k$-positive linerar maps $\Le':M_{d_B}(\Comp)\to M_{d_A}(\Comp)$,

    \item $\Tr(\rho_{AB} C_{\Le})\geq 0$ for all $k$-positive linear maps $\Le:M_{d_A}(\Comp)\to M_{d_B}(\Comp)$,
\end{enumerate}
\end{theorem}
In particular, a linear map $\Le$ is $k$-positive if and only if $C_{\Le}$ is \textit{$k$-blockpositive}, that is, 
    $$\la \xi|C_{\Le}|\xi\ra\geq 0\;\;\forall\,\xi\in \Comp^{d_A}\otimes \Comp^{d_B}\text{ with } \SR(|\xi\ra)\leq k,$$
as noted in \cite{Jam72,SSZ09}. Note that the case $k=1$ yields a duality between separability and positivity: a state $\rho_{AB}$ is separable if and only if $(\id_A\otimes \Le')(\rho_{AB})$ is positive semidefinite for all positive linear maps $\Le'$, recovering the result of Horodecki \cite{HHH96}. Furthermore, classical convex cone theory implies that several reverse-type duality relations also hold. For example, if a linear map  $\Le:M_{d_A}(\Comp)\to M_{d_B}(\Comp)$ is $(k-1)$-positive but not $k$-positive, then there exists a quantum state $\rho_{AB}$ with $\SN(\rho_{AB})=k$ such that $\Tr(\rho_{AB} C_{\Le})< 0$. For these reasons, $(k-1)$-positive (and not $k$-positive) linear maps, and their Choi matrices, are often referred to as \textit{$k$-Schmidt number witnesses} \cite{SBL01}, as they are capable of detecting quantum states whose Schmidt number is precisely $k$.

Such duality principles are in fact ubiquitous: they can be formulated within a general framework using the notion of \emph{mapping cones}~\cite{Sto86,Sko11} together with the Choi--Jamiołkowski correspondence. We refer to~\cite{GKS21,Kye23} for a detailed survey in this direction. For our purposes, let us recall another duality relation between PPT states and decomposable linear maps~\cite{Sto82}.

\begin{theorem} \label{thm-dual-ppt}
The following two equivalences hold.
\begin{enumerate}
    \item A bipartite $d_A\otimes d_B$ state $\rho_{AB}$ is PPT if and only if $\Tr(\rho_{AB}C_{\Le})\geq 0$ for all decomposable linear maps $\Le:M_{d_A}(\Comp)\to M_{d_B}(\Comp)$.

    \item A linear map $\Le:M_{d_A}(\Comp)\to M_{d_B}(\Comp)$ is decomposable if and only if $\Tr(\rho_{AB}C_{\Le})\geq 0$ for all $d_A\otimes d_B$ PPT states $\rho_{AB}$.
\end{enumerate}
\end{theorem}

An important implication of \cref{thm-dual-schmidt,thm-dual-ppt} is that the problem of finding PPT states with high Schmidt number can be transferred to a problem concerning positive linear maps, as follows.

\begin{corollary} \label{cor-PPTSch-equiv}
For $d_A, d_B\geq 2$ and $1\leq k< \min(d_A,d_B)$, the following are equivalent:
\begin{enumerate}
    \item There exists a $d_A\otimes d_B$ PPT state $\rho_{AB}$ with $\SN(\rho_{AB})> k$,

    \item There exists a $k$-positive indecomposable linear map $\Le:M_{d_A}(\Comp) \to M_{d_B}(\Comp)$.
\end{enumerate}
\end{corollary}
\begin{proof}
Let us assume $(1)$ and suppose every $k$-positive linear map is decomposable. Then Theorem \ref{thm-dual-ppt} implies that, for every PPT state $\rho_{AB}$ and a $k$-positive linear map $\Le$,
    $$\Tr(\rho_{AB}C_{\Le})\geq 0,$$
as $\Le$ is decomposable. By Theorem \ref{thm-dual-schmidt}, this implies that $\SN(\rho_{AB})\leq k$, a contradiction.

The proof of the direction $(2)\Rightarrow (1)$ is analogous.
\end{proof}

\subsection{Compact group symmetries and twirling operations}
Let us recall several symmetry notions and related results in quantum information theory and representation theroy which have been very effectively applied in the analysis of quantum entanglement and $k$-positivity \cite{VW01,PJPY24,PY24}. For an introduction to representation theory, we refer the reader to the excellent textbooks \cite{Simon95,Fol16}.

For a compact Hausdorff group $G$, consider (finite-dimensional) {\it unitary representations} $\pi:G\to U(d)$, $\pi_A:G\to  U(d_A)$, and $\pi_B:G\to  U(d_B)$ of $G$. Then we call
\begin{enumerate}
    \item a matrix $X\in M_d(\Comp)$ \textit{$\pi$-invariant} if $\pi(x)X\pi(x)^*=X$ for all $x\in G$,

    \item a linear map $\Le:M_{d_A}(\Comp)\to M_{d_B}(\Comp)$ \textit{$(\pi_A,\pi_B)$-covariant} if $\Le\circ \Ad_{\pi_A(x)} = \Ad_{\pi_B(x)} \circ \Le$ for all $x\in G$, where $\Ad_U(Z):=UZU^*$. In other words, for all $x\in G$ and $Z\in M_{d_A}(\Comp)$, we have $\Le(\pi_A(x)Z\pi_A(x)^*)=\pi_B(x)\Le(Z)\pi_B(x)^*$.
\end{enumerate}
Let us denote by ${\rm Inv}(\pi)$ (resp. ${\rm Cov}(\pi_A,\pi_B)$) the set of all $\pi$-invariant matrices (resp. $(\pi_A,\pi_B)$-covariant linear maps). We are often interested in the case that $G$ is a subgroup of a unitary group $U(d)$ and $\pi_A$ and $\pi_B$ are taken as the \textit{fundamental representation} $\iota_G$ of $G$ or its \textit{contragredient representation} $\overline{\iota_G}$, that is,
\begin{align*}
    \iota_G &: U\in G\mapsto U\in U(d),\\
    \overline{\iota_G} &: U\in G\mapsto \overline{U}\in U(d),
\end{align*}
where $\overline{U}$ denotes the entrywise complex conjugate of $U$.

$\pi$ is called \textit{irreducible} if ${\rm Inv}(\pi)=\Comp I_d$ (this coincides with the usual definition of irreducibility in representation theory, thanks to Schur's lemma \cite{Simon95}). If $\pi$ is irreducible, so is its contragredient representation $\overline{\pi}$. For unitary representations $\pi_A$ and $\pi_B$, the \textit{direct sum representation $\pi_A\oplus \pi_B:G\to U(d_A+d_B)$} (resp. {\it tensor representation} $\pi_A\otimes \pi_B:G\rightarrow U(d_Ad_B)$) is given by $(\pi_A\oplus \pi_B)(x)=\pi_A(x)\oplus \pi_B(x)$ (resp. $(\pi_A\otimes \pi_B)(x)=\pi_A(x)\otimes \pi_B(x)$) for $x\in G$. Two unitary
representations of the same dimension $\pi,\pi':G\to U(d)$ are said to be \textit{unitarily equivalent} if there exists a unitary matrix $U\in U(d)$ such that $\pi'(x)=U \pi(x)U^*$ for all $x\in G$; we write $\pi\cong \pi'$.

It is a well-known fact \cite{Simon95} that every finite dimensional unitary representation can be decomposed into irreducible subrepresentations (called \textit{irrep decomposition})
\begin{equation}\label{eq-fusion}
    \pi\cong \bigoplus_{i=1}^l \sigma_i\otimes I_{m_i},
\end{equation}
where $\sigma_i:G\to U(d_i)$ ($i=1,2,\ldots, l$) are mutually inequivalent irreducible sub-representations of $\pi$ with multiplicity $m_i$. The numbers $l$ and $m_1,\ldots, m_l$ are uniquely determined. Furthermore, when $m_i\equiv 1$, we call $\pi$ \textit{multiplicity-free}, in which case we have a rather simple description of the space
\begin{equation} \label{eq-inv-multfree}
    \Inv(\pi)={\rm span}\{\Pi_1,\ldots, \Pi_l\},
\end{equation}
where $\Pi_i$'s are mutually orthogonal projections onto the subrepresentation spaces of $\sigma_i$.

We are additionally interested in two types of symmetrizing operations: for unitary representations $\pi,\pi_A,\pi_B$ of $G$,
\begin{enumerate}
    \item \textit{$\pi$-twirling} $\mathcal{T}_{\pi}X:=\mathbb{E}_{G}[{\rm Ad}_{\pi(\cdot)}(X)]=\mathbb{E}_{G}[\pi(\cdot)X\pi(\cdot)^*]$ for a matrix $X\in M_d(\Comp)$,

    \item \textit{$(\pi_A,\pi_B)$-twirling} $\mathcal{T}_{\pi_A,\pi_B}\Le:=\mathbb{E}_G[{\rm Ad}_{\pi_B(\cdot)^*}\circ \Le\circ {\rm Ad}_{\pi_A(\cdot)}]$ for a linear map $\Le: M_{d_A}(\Comp)\to M_{d_B}(\Comp)$,
\end{enumerate}
where the expectation is defined with respect to the (normalized) Haar measure of the compact group $G$. These two operations are understood as projections (conditional expectations) onto the spaces $\Inv(\pi)$ and $\Cov(\pi_A,\pi_B)$, resp. That is, we have ${\rm Ran}\mathcal{T}_{\pi}=\Inv(\pi)$ and $\mathcal{T}_{\pi}\circ \mathcal{T}_{\pi}=\mathcal{T}_{\pi}$ so that a matrix $X$ is $\pi$-invariant if and only if $\T_{\pi}(X)=X$; the analogous statement holds true for $\mathcal{T}_{\pi_A,\pi_B}$.

We further collect several general properties of group symmetries and twirling operations which were proved and discussed in detail in \cite{VW01,PJPY24,PY24,NP25}.

\begin{proposition} [\cite{VW01,PJPY24}] \label{prop-Symm-general}
Let $\pi_A$ and $\pi_B$ be unitary representations.
\begin{enumerate}
    \item For $X\in M_{d_A}\otimes M_{d_B}$, $X\in \Inv(\pi_A\otimes \pi_B)$ if and only if $X^{\Gamma}=(\id_A\otimes \top)(X)\in \Inv(\pi_A\otimes \overline{\pi_B})$.

    \item For a linear map $\Le:M_{d_A}(\Comp)\to M_{d_B}(\Comp)$, $\Le\in \Cov(\pi_A,\pi_B)$ if and only if $\top\circ \Le\in \Cov(\pi_A,\overline{\pi_B})$.

    \item For a linear map $\Le:M_{d_A}(\Comp)\to M_{d_B}(\Comp)$, we have
        $$C_{\T_{\pi_A,\pi_B}\Le} = \T_{\overline{\pi_A}\otimes \pi_B}C_{\Le}.$$
    In particular,     $\Le\in \Cov(\pi_A,\pi_B)$ if and only if $C_{\Le}\in \Inv(\overline{\pi_A}\otimes \pi_B)$.
\end{enumerate}
\end{proposition}

\begin{proposition} [\cite{PJPY24,PY24}] \label{prop-twirling-preserving}
Let $\pi_A$ ad $\pi_B$ be unitary representations of $G$ with dimensions $d_A$ and $d_B$, resp.
\begin{enumerate}
    \item The $\pi_A\otimes \pi_B$-twirling $\T_{\pi_A\otimes \pi_B}$ preserves following properties of $d_A\otimes d_B$ matrices: separability, PPT property, decomposability, $k$-blockpositivity, and the membership in $\sch_k$.

    \item The $(\pi_A, \pi_B)$-twirling $\T_{\pi_A,\pi_B}$ preserves following properties of linear maps: CP, co-CP, decomposability, and $k$-positivity.
\end{enumerate}
\end{proposition}

\begin{proposition} [\cite{NP25}] \label{prop-twirlformula}
Let $\pi$ be a multiplicity-free unitary representation having the irrep decomposition as in \cref{eq-fusion} with $m_i\equiv 1$. Then we have

\begin{equation}
    \mathcal{T}_{\pi}(X)=\sum_{i=1}^l \frac{{\rm Tr}(\Pi_i X)}{d_i}\Pi_i,
\end{equation}
where $\Pi_1,\ldots, \Pi_l$ are orthogonal projections as in \cref{eq-inv-multfree}.
\end{proposition}

Finally, we emphasize that the dualities between various entanglement properties and positivity properties in \cref{thm-dual-schmidt,thm-dual-ppt} are compatible with group symmetries. Here we partially recall the results that will be of primary interest in the present work; for general results involving mapping cones, we refer to \cite[Theorem~3.3]{PY24}.

\begin{theorem} [\cite{PJPY24,PY24}] \label{thm-DualSymm}
Let $G$ be a compact group and let $\pi_A:G\to U(d_A)$ and $\pi_B:G\to U(d_B)$ be unitary representations of $G$.
\begin{enumerate}
    \item For a $\overline{\pi_A}\otimes \pi_B$ quantum state $\rho_{AB}$, the following are equivalent:
    \begin{enumerate}
        \item $\rho_{AB}\in \sch_k$;

        \item $(\id_A\otimes \Le')(\rho_{AB})\geq 0$ for all $(\pi_B,\pi_A)$-covariant $k$-positive linear maps $\Le'$;

        \item $\Tr(\rho_{AB}C_{\Le})\geq 0$ for all $(\pi_A,\pi_B)$-covariant $k$-positive linear maps $\Le$.
    \end{enumerate}

    \item For a $\overline{\pi_A}\otimes \pi_B$ quantum state $\rho_{AB}$, the following are equivalent:
    \begin{enumerate}
        \item $\rho_{AB}$ is PPT;

        \item $(\id_A\otimes \Le')(\rho_{AB})\geq 0$ for all $(\pi_B,\pi_A)$-covariant decomposable linear maps $\Le'$;

        \item $\Tr(\rho_{AB}C_{\Le})\geq 0$ for all $(\pi_A,\pi_B)$-covariant decomposable linear maps $\Le$.
    \end{enumerate}

    \item A $(\pi_A,\pi_B)$ linear map $\Le$ is decomposable if and only if $\Tr(\rho_{AB}C_{\Le})\geq 0$ for all $\overline{\pi_A}\otimes \pi_B$-invariant PPT states $\rho_{AB}$.
\end{enumerate}
\end{theorem}

As a consequence, the equivalence in Corollary~\ref{cor-PPTSch-equiv} carries over to the setting of group symmetries as follows. The proof is quite the same with that of Corollary~\ref{cor-PPTSch-equiv} which is left to the readers.

\begin{corollary} \label{cor-symm-PPTSch-equiv}
In the setting of Theorem~\ref{thm-DualSymm}, and for $1\leq k < \min(d_A,d_B)$, the following are equivalent:
\begin{enumerate}
    \item There exists a PPT quantum state $\rho_{AB}$ with $SN(\rho)>k$ within the class ${\rm Inv}(\overline{\pi_A}\otimes \pi_B)$;

    \item There exists a $k$-positive indecomposable map within the class ${\rm Cov}(\pi_A,\pi_B)$.
\end{enumerate}
\end{corollary}

In the remainder of this work, we demonstrate the effectiveness of Theorem~\ref{thm-DualSymm} and Corollary~\ref{cor-symm-PPTSch-equiv} in the study of quantum objects exhibiting \emph{symplectic group symmetries}, which are introduced in the following section.

\subsection{Symplectic group symmetries}

Let $d \geq 2$ be an even integer and let $V$ be a skew-symmetric unitary matrix, i.e., $V\in U(d)$ and $V^{\top}=-V$. The \textit{skew-symmetric bilinear form} $\la \cdot, \cdot \ra_V: \Comp^d\times \Comp^d \to \Comp$ associated to $V$ is defined by
    $$\la x,y\ra_V:=x \cdot Vy = x^{\top} Vy,$$
where the usual \textit{dot product} above is defined without complex conjugates. This form satisfies the two elementary properties
\begin{equation} \label{eq-skew-form}
    \la x,y \ra_V = -\la y,x\ra_V, \quad \la x, x\ra_V\equiv 0
\end{equation}
for all $x,y\in \Comp^d$, thanks to the skew-symmetry of $V$. The \textit{compact symplectic group} (or \textit{unitary symplectic group}) $Sp(d; V)\subset U(d)$ associated to $V$ is then defined by
\begin{align*}
    Sp(d;V)&:=\{S\in U(d)\bigm| \la Sx,Sy\ra_{V}=\la x,y\ra_{V} \text{ for all $x,y\in \Comp^d$}\}\\
    &\,\,=\{S\in U(d)\bigm| S^{\top}V S=V\}.
\end{align*}
Each element in $Sp(d;V)$ is called a \textit{unitary symplectic matrix}. In other words, $Sp(d;V)$ is the group of $d\times d$ matrices that preserve both the Hermitian inner product and and the bilinear form $\la \cdot, \cdot\ra_V$. A standard choice of $V$ is
    $$V=\Om_d:=\begin{pmatrix} 0 & I_{d/2} \\ -I_{d/2} & 0\end{pmatrix}; \quad \la x,y\ra_{\Om}:=\sum_{i=1}^{d/2}(x_iy_{i+d/2}-x_{i+d/2}y_i),$$
in which case we simply denote $Sp(d):=Sp(d;\Om_d)$.

We remark on several points concerning these definitions. First, skew-symmetric unitary matrices (and hence symplectic groups) are defined only for \emph{even} $d$. Indeed, if $d$ is odd, then any skew-symmetric matrix $V$ is non-invertible, since
    $$\det V= - \det V^{\top} = -\det V.$$
Furthermore, a classical result in matrix theory (see e.g. \cite[Corollary 4.4.19]{TextHJ12}) states that every skew-symmetric unitary matrix $V \in U(d)$ is in fact \emph{congruent} to $\Omega_d$, i.e., there exists a unitary matrix $U_V \in U(d)$ such that
\begin{equation} \label{eq-Symp-cong}
    U_V^{\top}VU_V = \Om_d.
\end{equation}
In particular, the map $S \mapsto {\rm Ad}_{U_V}(S)=U_V S U_V^*$ defines a group isomorphism between $Sp(d)$ and $Sp(d;V)$. Consequently, all results in this paper would be independent of the particular choice of $V$. Finally, the group $Sp(d)$ is closed under complex conjugation and transposition: if $S\in Sp(d)$, then 
    $$S^{\top} \Om_d S = \Om_d \iff \Om_d= \overline{S} \Om_d S^* \iff \Om_d = S \Om_d S^{\top},$$
where we take matrix adjoint and the fact $\Om_d^*=-\Om_d$ in the second equivalence. In particular, this implies $\Om_dS = \overline{S}\Om_d$ and $\Om_d \overline{S} =S\Om_d$ for all $S\in Sp(d)$. The same argument gives in general,
\begin{equation} \label{eq-Symp-Comm}
    S\in Sp(d; V) \implies VS = \overline{S} V \quad \text{and} \quad \overline{V}\overline{S} = S \overline{V}.
\end{equation}

We are mainly interested in classes of linear maps and bipartite matrices that are symmetric with respect to the fundamental representation
    $$\iota_{Sp(d;V)}:S\in Sp(d;V)\mapsto S\in U(d).$$
For simplicity of notation, let us set
\begin{align*}
\Inv(S\otimes S) = \Inv(\iota_{Sp(d)} \otimes \iota_{Sp(d)}),
&\quad \Cov(S, S) = \Cov(\iota_{Sp(d)}, \iota_{Sp(d)}),\\
\Inv(S\otimes \overline{S}) = \Inv(\iota_{Sp(d)} \otimes \overline{\iota_{Sp(d)}}),
&\quad \Cov(S, \overline{S}) = \Cov(\iota_{Sp(d)}, \overline{\iota_{Sp(d)}}),
\end{align*}
in the case $V=\Om_d$, and refer to elements of these spaces as $S\otimes S$-invariant, $S\otimes \overline{S}$-invariant, $(S,S)$-covariant, and $(S,\overline{S})$-covariant, resp. Note that all these classes are closed under complex conjugation; for instance, $\Inv(S\otimes \overline{S})=\Inv(\overline{S}\otimes S)$ and $\Cov(S,\overline{S}) =  \Cov(\overline{S},S)$, since $Sp(d)$ is closed under the complex conjugation. Furthermore, the \emph{Schur–Weyl duality} for the compact symplectic group (see, e.g., \cite{TextWeyl,TextGoodman,Wen88}) states that, for $d\geq 4$, the algebra $\Inv(S\otimes S)$ is $3$-dimensional and is described by the action of the Brauer algebra \cite{Bra37}. More precisely, one has
    $$\Inv(S\otimes S)={\rm span}\{I_d\otimes I_d, F_d, (I_d\otimes \Om_d)|\om_d\ra\la \om_d|(I_d\otimes \Om_d)^*\},$$
where $F_d:=\sum_{i,j=1}^d |ij\ra\la ji|$ denotes the flip operator. The remaining three spaces can be described explicitly via the relations
    $$\Inv(S\otimes \overline{S})=(\id_d\otimes \top)\big(\Inv(S\otimes S)\big), \quad \Cov(S,S)= C^{-1}\big(\Inv(S\otimes \overline{S})\big), \quad \Cov(S,\overline{S})= C^{-1}\big(\Inv(S\otimes S)\big)$$
as a consequence of Proposition \ref{prop-Symm-general}. Here, $C^{-1}$ denotes the inverse of the Choi isomorphism. These characterizations are summarized in the following proposition.

\begin{proposition} \label{prop-Symp-object}
Let $d\geq 4$ be an even integer.
\begin{enumerate}
    \item A matrix $\rho\in M_d(\Comp)\otimes M_d(\Comp)$ is $S\otimes {S}$-invariant (resp. $S\otimes \overline{S}$-invariant) if and only if
    \begin{center}
        $\rho\in {\rm span}\{I_{d^2},F_d, \Ad_{I_d\otimes \Om_d}(|\om\ra\la \om|)\}$ (resp. $\rho\in {\rm span}\{I_{d^2},|\om_d\ra\la \om_d|, \Ad_{I_d\otimes \Om_d}(F_d)\}$). 
    \end{center}
    
    \item A linear map $\Le:M_d(\Comp)\to M_d(\Comp)$ is $(S,S)$-covariant (resp. $(S,\overline{S})$-covariant) if and only if 
    \begin{center}
        $\Le\in {\rm span}\{ \Delta, \id, \Ad_{\Om}\circ \top\}$ (resp. $\Le\in {\rm span}\{ \Delta, \top, \Ad_{\Om}\}$),
    \end{center}
    where $\Delta(Z):=\frac{\Tr(Z)}{d}I_d$ is the completely depolarizing map on $M_d(\Comp)$.
\end{enumerate}
\end{proposition}

\begin{corollary} \label{cor-Symp-proj}
The set $\Inv(S\otimes \overline{S})$ is a commutative $*$-algebra spanned by three (mutually orthogonal) minimal projections
\begin{equation} \label{eq-SSbar-proj}
    \Pi_1^{S\overline{S}} = |\om_d\ra\la \om_d|, \quad \Pi_2^{S\overline{S}}=\frac{1}{2}(I_{d^2}+ F_d^{\Om}), \quad \Pi_3^{S\overline{S}}=\frac{1}{2}(I_{d^2}-F_d^{\Om}) -|\om_d\ra\la \om_d|,
\end{equation}
where $F_d^{\Om}:=\Ad_{I_d\otimes \Om_d}(F_d)$. Similarly, the $*$-algebra $\Inv(S\otimes S)$ is spanned by three projections
\begin{equation} \label{eq-SS-proj}
    \Pi_1^{SS} = |\om_d^{\Om}\ra\la \om_d^{\Om}|, \quad \Pi_2^{SS}=\frac{1}{2}(I_{d^2}+F_d), \quad \Pi_3^{SS}=\frac{1}{2}(I_{d^2}-F_d) -|\om_d^{\Om}\ra\la \om_d^{\Om}|,
\end{equation}
where $|\om_d^{\Om}\ra:=(I_d\otimes \Om_d)|\om_d\ra$.
\end{corollary}
\begin{proof}
We know that $\Inv(S\otimes S)=\{S\otimes S:S\in Sp(d)\}'$ is a three-dimensional $*$-algebra spanned by $I_{d^2}, |\om_d^{\Om}\ra\la \om_d^{\Om}|$, and $F_d$. Furthermore, the operators $\Pi_{\mathcal{S}}=\frac{1}{2}(I_{d^2}+F_d)$ and $\Pi_{\A}=\frac{1}{2}(I_{d^2}-F_d)$ are orthogonal projections onto the symmetric / antisymmetric subspaces $\Comp^d\vee \Comp^d$ and $\Comp^d\wedge \Comp^d$, respectively. On the other hand, one has
    $$|\om_d^{\Om}\ra=\frac{1}{\sqrt{d}}\sum_{i=1}^{d/2}(|i+d/2,i\ra -|i.i+d/2\ra) \in \Comp^d \wedge \Comp^d = {\rm Ran}(\Pi_{\A}),$$
which implies that $\Pi_1^{SS}$, $\Pi_2^{SS}=\Pi_{\mathcal{S}}$, and $\Pi_3^{SS}=\Pi_{\A}-\Pi_1^{SS}$ are orthogonal projections satisfying $\Pi_1^{SS}+\Pi_2^{SS}+\Pi_3^{SS}=I_{d^2}$. Therefore, these three projections are mutually orthogonal, showing the second assertion. The first assertion follows similarly.

\end{proof}

\begin{remark} \label{rmk-V-equiv}
Similar classes of linear maps and bipartite matrices with symmetries under the general symplectic group $Sp(d;V)$ can be characterized. For example, a linear map $\Le$ is $(\iota_{Sp(d;V)},\iota_{Sp(d;V)})$-covariant if and only if
\begin{center}
    $\Le = p\Delta +q\,\id_d + r\Ad_{\overline{V}}\circ \top$.
\end{center}
for some $p,q,r\in \Comp$. Here, we take the complex-conjugated matrix $\overline{V}$; the covariance property of $\Ad_{\overline{V}}\circ \top$ follows from \cref{eq-Symp-Comm}. However, by recalling the congruence relation \cref{eq-Symp-cong}, this map is unitarily equivalent to the $(S,S)$-covariant map
    $$\Le'=p\Delta+q\, \id_d+r\Ad_{\Om_d}\circ \top = \Ad_{U_V^*} \circ \Le \circ \Ad_{U_V}.$$
In particular, $\Le$ and $\Le'$ share all $k$-positivity and entanglement-related properties. An analogous conclusion holds for $(\iota_{Sp(d;V)},\iota_{Sp(d;V)})$-covariant maps, as well as for $\iota_{Sp(d;V)}\otimes \iota_{Sp(d;V)}$- and $\iota_{Sp(d;V)}\otimes \overline{\iota_{Sp(d;V)}}$-invariant matrices. Consequently, all results concerning $Sp(d)=Sp(d;\Om_d)$ symmetries established in this paper extend \emph{verbatim} to the general $Sp(d;V)$ setting upon replacing $\Om_d$ with $V$.
\end{remark}

We collect several extensively studied examples of symplectic covariant linear maps and symplectic invariant states.
\begin{example}

\begin{enumerate}
    \item The depolarizing maps and the transpose-depolarizing maps
    \begin{align*}
        \Cov(S,S) \ni \Le^{\text{dep}}_p&: Z\mapsto \frac{1-p}{d}\Tr(Z) I_d + pZ,\\
        \Cov(S,\overline{S}) \ni \Le_q^{ \top\text{-dep}}&:Z\mapsto \frac{1-q}{d} \Tr(Z) I_d + qZ^{\top},
    \end{align*}
    are characterized by $(U,U)$- and $(U,\overline{U})$-covariance ($U\in U(d)$), respectively. The general maps in $\Cov(S,S)$ and $\Cov(S,\overline{S})$ are threrfore mixtures of depolarizing and transpose-depolarizing maps, up to a “twist” by $\Ad_{\Omega}$.
    
    \item The \textit{isotropic states} \cite{HH99} and the \textit{Werner states } \cite{Wer89}
    \begin{align*}
        \Inv(S\otimes \overline{S})\ni\rho^{\rm iso}_{a} &= \frac{1-a}{d^2}I_{d^2}+ a|\om_d\ra\la \om_d|, \quad -\frac{1}{d^2-1}\leq a \leq 1, \\
        \Inv(S\otimes S) \ni \rho^{\rm Wer}_b &= \frac{1-b}{d^2}I_{d^2} + \frac{b}{d}F_d, \qquad -\frac{1}{d-1}\leq b\leq \frac{1}{d+1}, 
    \end{align*}
    are characterized by $U\otimes \overline{U}$- and $U\otimes U$-invariance ($U\in U(d)$), respectively.
\end{enumerate}
\end{example}

\begin{example} \label{ex-BH}
For a skew-symmetric unitary matrix $V\in U(d)$ and the complex conjugation map $J:v\in \Comp^d \mapsto \overline{v}$, the antiunitary operator $VJ$ is often referred to as a \emph{time-reversal operator}. Furthermore, the associated linear map
    $$(\Ad_{V}\circ \top)(Z)=V Z^{\top} V^* = (VJ)Z^* (VJ)^{-1}$$
is called the \emph{time-reversal transformation} \cite{Bre06}, and it turns out to be $(\iota_{Sp(d;\overline{V})},\iota_{Sp(d;\overline{V})})$-covariant. The class $\Cov(\iota_{Sp(d;\overline{V})},\iota_{Sp(d;\overline{V})})$ contains another important example, namely the \emph{Breuer-Hall map} \cite{Bre06,Hal06}
    $$\Le^{\rm BH}(Z):= \frac{1}{d-2}(d\Delta-\id_d- \Ad_V\circ \top)(Z)=\frac{1}{d-2}(\Tr( Z)I_d - Z - V Z^{\top} V^*),$$
where the normalization is chosen so that $\Le^{\rm BH}$ is TP. The map $\Le^{\rm BH}$ is known to be positive and indecomposable, and hence can be used to detect PPT entangled states. In addition, the Choi matrix of $\Le^{\rm BH}$ is an \emph{optimal} entanglement witness, meaning that there exists no other witness that can detect strictly more entangled states \cite{Bre06}. 
In the following, we will show that $\Le^{\rm BH}$ is \emph{not} $2$-positive; hence, although it can be used to detect PPT entanglement, it cannot be used to detect states whose Schmidt number is strictly larger than $2$.
\end{example}

\begin{example} \label{ex-PV19}
An important class of entangled states with $S \otimes \overline{S}$ symmetry is given by the \textit{P\'al--V\'ertesi states} \cite{PV19},
    $$\rho^{\rm PV}:= \frac{1}{d+2} (I_d\otimes \Om_d)|\om_d\ra\la \om_d| (I_d\otimes \Om_d)^* + \frac{1}{d(d+2)}(I_{d^2}+F_d).$$
It was shown that for all even integers $d \ge 4$, the state $\rho^{\rm PV}$ is PPT while satisfying $\SN(\rho^{\rm PV})\geq d/2$, which is, to date, the largest known Schmidt number of PPT states for general $d$.

In the following section, we will show that the inequality $\SN(\rho^{\rm PV})\geq d/2$ is in fact tight, i.e., $\SN(\rho^{\rm PV})=d/2$. Moreover, this bound cannot be improved within either of the symmetry classes $\Inv(S \otimes S)$ and $\Inv(S \otimes \overline{S})$, by determining the exact region of PPT states having Schmidt number $d/2$.

\end{example}

From now on, we shall consider the following two-parameter families of linear maps and bipartite matrices:
\begin{align}
    &\Cov(S,S)\ni\Le_{a,b}=\Le^{(d)}_{a,b}: Z\mapsto (1-a-b)\frac{\Tr(Z)}{d}I_d+aZ+b\,\Om_d Z^{\top}\Om_d^*, \label{eq-SympCov}\\ 
    &\Inv(S\otimes \overline{S})\ni \rho_{a,b}=\rho_{a,b}^{(d)}:=\frac{1-a-b}{d^2}I_d\otimes I_d+a|\om_d\ra\la \om_d|+\frac{b}{d}(I_d\otimes \Om_d){F_d} (I_d\otimes \Om_d)^{*}. \label{eq-SympInv}
\end{align}
By Proposition \ref{prop-Symp-object}, the families above characterize the \((S,S)\)-covariant trace-preserving linear maps and the bipartite \(S \otimes \overline{S}\)-invariant matrices with unit trace, resp. Furthermore, the maps $\top\circ \Le_{a,b}$ and the the matrices $\rho_{a,b}^{\Gamma}$ characterize the $(S, \overline{S})$-covariant TP maps and the trace-one $S\otimes S$-invariant matrices, resp. For example, up to equivalence, the Breuer-Hall map and the P\'al--V\'ertesi state correspond to
    $$\Le^{\rm BH}=\Le_{-\frac{1}{d-2}, -\frac{1}{d-2}}, \quad \rho^{\rm PV}=(\rho_{\frac{1}{d+2}, \frac{1}{d+2}})^{\Gamma}$$
Note that $C_{\Le_{a,b}}=\rho_{a,b}$ for every $a,b\in \Comp$, and $\rho_{a,b}$ is Hermitian (equivalently, $\Le_{a,b}$ is Hermitian preserving) if and only if $a,b \in \Real$. 

The aim of the following sections is to characterize $k$-positivity and decomposability of $\Le_{a,b}$ and $\top\circ \Le_{a,b}$, and to compute the Schmidt numbers of $\rho_{a,b}$ and $\rho_{a,b}^{\Gamma}$ (whenever they define quantum states). In fact, the following simple but useful observation allows us to restrict attention to only half of the parameter space.

\begin{proposition} \label{prop-Symp-equivalence}
For any $a,b\in \Real$, we have
\begin{align*}
    \top\circ \Le_{a,b} = \Ad_{\Om_d}\circ \Le_{b,a}, \quad 
    \rho_{a,b}^{\Gamma} = (I_d\otimes \Om_d)\rho_{b,a} (I_d\otimes \Om_d)^*.
\end{align*}
In particular, $\top\circ \Le_{a,b}$ is unitarily equivalent to $\Le_{b,a}$, and $\rho_{a,b}^{\Gamma}$ is local-unitary equivalent to $\rho_{b,a}$.
\end{proposition}

\section{Positivity of symplectic covariant linear maps} \label{sec-SympkPos}

In this section, we investigate various positivity properties of the Hermitian-preserving $(S,S)$-covariant maps $\{\Le_{p,q}\}_{p,q\in \Real}$ defined in \cref{eq-SympCov}.

\subsection{Positivity, complete positivity, and decomposability}

Let us begin with a characterization of complete (co-)positivity, positivity, and decomposability of $\Le_{p,q}=\Le_{p,q}^{(d)}$ in terms of the parameters  $p$ and $q$, for arbitrary even dimensions $d\geq 4$. For convenience, we briefly illustrate the \textit{geometric} description of these four regions in Fig. \ref{fig-PosDec}, while their \textit{algebraic} characterizations are given below.

\begin{figure}
    \centering

    \includegraphics[width=0.4\linewidth]{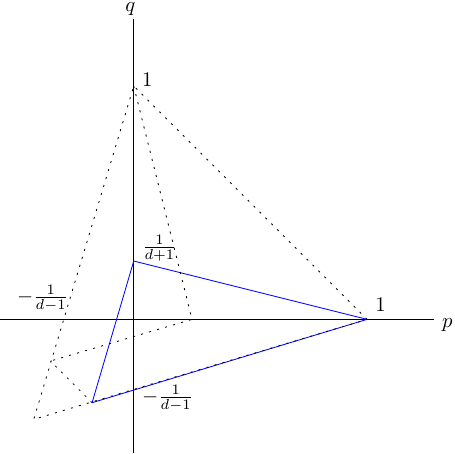} \;
    \includegraphics[width=0.4\linewidth]{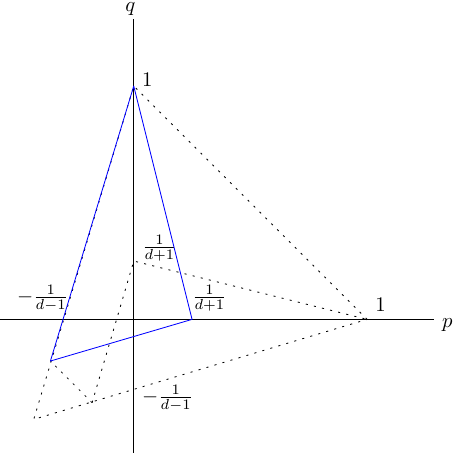} \\

    \vspace{2mm}
    \includegraphics[width=0.4\linewidth]{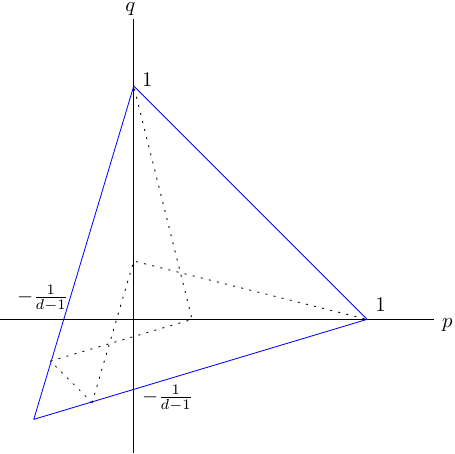}\;\;
    \includegraphics[width=0.4\linewidth]{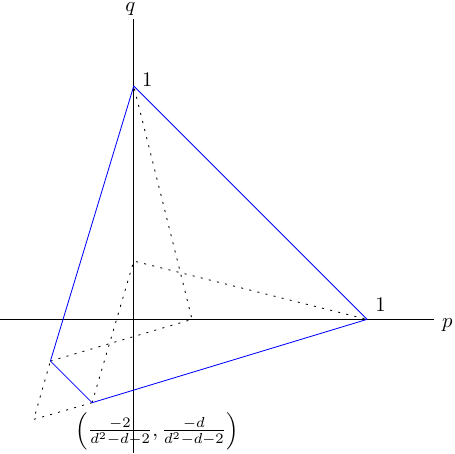}
    
    \caption{The complete positivity (top left), complete copositivity (top right), positivity (bottom left), and decomposability (bottom right) regions of $\Le_{p,q}$}
    \label{fig-PosDec}
\end{figure}

\begin{proposition} \label{prop-Symp-CP}
The linear map $\Le_{p,q}$ is completely positive if and only if $(p,q)\in \mathbb{P}_d$, where
\begin{equation} \label{eq-Symp-CP}
    \mathbb{P}_d=\mathbb{P}_d^{(d)}:= \left\{\ (p,q)\in \Real^2:\;\begin{array}{lll} p+(1-d)q\leq 1, \\ p+(1+d)q\leq 1, \\ (1-d)p+q\leq\frac{1}{d+1}\end{array} \right\}.
\end{equation}
Furthermore, $\Le_{a,b}$ is completely copositive if and only if $(b,a)\in \mathbb{P}_d$.

\end{proposition}
\begin{proof}
The complete positivity of $\Le_{a,b}$ is equivalent to the positivity of $C_{\Le_{p,q}}=\rho_{p,q}$. Furthermore, it is straightforward to check that
    $$\rho_{p,q} = \Big(\frac{1-p-q}{d^2}+p-\frac{q}{d}\Big)\Pi_1^{S\overline{S}}+\Big(\frac{1-p-q}{d^2}+\frac{q}{d}\Big)\Pi_2^{S\overline{S}} + \Big(\frac{1-p-q}{d^2} - \frac{q}{d}\Big)\Pi_3^{S\overline{S}},$$
where $\Pi_{1,2,3}^{S\overline{S}}$ are mutually orthogonal projections defined in Corollary \ref{cor-Symp-proj}.
Therefore, $\rho_{p,q}\geq 0$ if and only if the three coefficients above are nonnegative, which is equivalent to $(p,q)\in \mathbb{P}_d$. The second statement follows from Proposition \ref{prop-Symp-equivalence}: $\Le_{p,q}$ is co-CP if and only if $\top\circ \Le_{p,q}\cong \Le_{q,p}$ is CP if and only if $(q,p)\in \mathbb{P}_d$.
\end{proof}

\begin{corollary} \label{cor-Symp-psd}
An $S\otimes \overline{S}$-invariant matrix $\rho_{a,b}=C_{\Le_{a,b}}$ as in \cref{eq-SympInv} defines a quantum state if and only if $(a,b)\in \mathbb{P}_d$. 
\end{corollary}

Interestingly, the positivity and decomposability of $\Le_{p,q}$ are also simple to characterize and turn out to be distinct (Fig. \ref{fig-PosDec}) although both to properties are in general highly nontrivial to characterize.

\begin{theorem} \label{thm-Symp-PosDec}
The positivity and decomposability regions of $\Le_{p,q}$ is completely described by the sets $\mathbb{P}_1=\mathbb{P}_1^{(d)}$ and $\mathbb{D}=\mathbb{D}^{(d)}$, respectivelty, where
\begin{align} 
    \mathbb{P}_1&:=\left\{(p,q)\in \Real^2: \;\begin{array}{lll} p+q\leq 1, \\ p+(1-d)q\leq 1, \\ (1-d)p+q\leq 1\end{array}\right\}, \label{eq-Symp-pos} \\ \mathbb{D}&:=\left\{(p,q)\in \Real^2: \;\begin{array}{lll} -\frac{2+d}{d^2-d-2} \leq p+q\leq 1, \\ p+(1-d)q\leq 1, \\ (1-d)p+q\leq 1\end{array}\right\}. \label{eq-Symp-dec}
\end{align}
In particular, $\Le_{p,q}$ is positive and indecomposable in the parameter region $\begin{cases} p+q<-\frac{2+d}{d^2-d-2} , \\ p+(1-d)q\leq 1, \\ (1-d)p+q\leq 1.\end{cases}$
\end{theorem}
\begin{proof}
By convexity, the positivity of $\Le_{p,q}$ is equivalent to the condition 
    $$\Le_{p,q}(|v\ra \la v|) = \frac{1-p-q}{d} I_{d^2} + p|v\ra\la v| + q|\Om_d\overline{v}\ra\la \Om_d\overline{v}| \geq 0$$
for all unit vectors $v\in \Comp^d$. However, we have $v\perp \Om_d\overline{v}$ for all $v$ by Eq. \eqref{eq-skew-form}:
    $$\la v|\Om_d\overline{v}\ra = \la \overline{v},\overline{v}\ra_{\Om}=0.$$
Consequently, $\Le_{p,q}(|v\ra\la v|)\geq 0$ if and only if the three inequalities
    $$\frac{1-p-q}{d}\geq 0, \quad \frac{1-p-q}{d}+p\geq 0, \quad \frac{1-p-q}{d}+q\geq 0,$$
which is further equivalent to $(p,q)\in \mathbb{P}_1$. Since the condition above is independent of the choice of $v$, this precisely characterizes the positivity condition for $\Le_{p,q}$.

Now suppose $\Le_{p,q}$ is decomposable, and write $\Le_{p,q} = \Le^{(1)}+\Le^{(2)}$ where $\Le^{'}$ is CP and $\Le^{''}$ is co-CP. By applying the $(S,S)$-twirling operation $\mathcal{T}_{S,S}:=\mathcal{T}_{\iota_{Sp(d)},\iota_{Sp(d)}}$, we have $\Le_{p,q}=\T_{S,S}\Le^{'}+\mathcal{T}_{S,S}{\Le''}$ where $\T_{S,S}\Le^{'}$ and $\T_{S,S}\Le^{''}$ are $(S,S)$-covariant CP and is co-CP maps, resp, by Proposition \ref{prop-twirling-preserving}. Therefore, the trace-preserving property of $\Le_{p,q}$, combined with Proposition \ref{prop-Symp-CP}, implies that
    $$\Le_{p,q}=\lambda \Le_{p_1,q_1}+ (1-\lambda)\Le_{p_2,q_2} = \Le_{\lambda p_1+(1-\lambda)p_2, \lambda q_1+(1-\lambda)q_2}$$
for some $(p_1,q_1)\in \mathbb{P}_d$, $(q_2,p_2)\in \mathbb{P}_d$, and $\lambda\in [0,1]$. Consequently, $\Le_{p,q}$ is decomposable if and only if $(p,q)$ is in the convex hull of $\mathbb{P}_d$ and its reflection $\mathbb{P}_d^{\top}=\{(q,p):(p,q)\in \mathbb{P}_d\}$, that is,
    $$(p,q)\in {\rm conv}(\mathbb{P}_d\cup \mathbb{P}_d^{\top})=\mathbb{D}.$$
\end{proof}

\subsection{$k$-Positivity} \label{sec-application}

We now establish a complete characterization of $k$-positivity of $\Le_{p,q}$, which is more involved than positivity or decomposability in the previous section, and the geometry of the regions highly depends on the values $1\leq k\leq d$. Let us first set the notation
\begin{equation} \label{eq-kposset}
    \mathbb{P}_k=\mathbb{P}_k^{(d)}:=\set{(p,q)\in \Real^2:\Le_{p,q}^{(d)} \text{ is $k$-positive}},\;1\leq k\leq d.
\end{equation}
In the previous section, we already obtained the positivity region $\mathbb{P}_1$ and the complete positivity region $\mathbb{P}_d$, so we may consider only the cases $1<k<d$.

The main result of this section is as follows (we also refer to Fig. \ref{fig-kPos} for brief illustrations).

\begin{figure}
    \centering

    \includegraphics[width=0.4\linewidth]{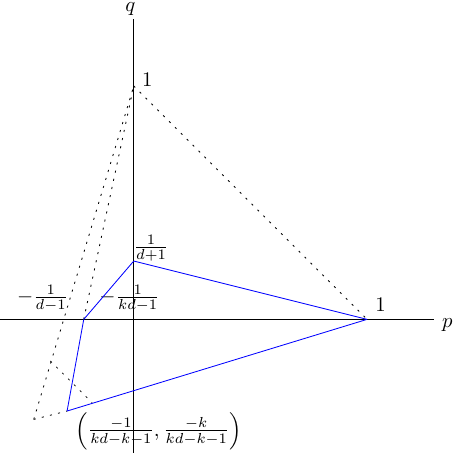} \;
    \includegraphics[width=0.4\linewidth]{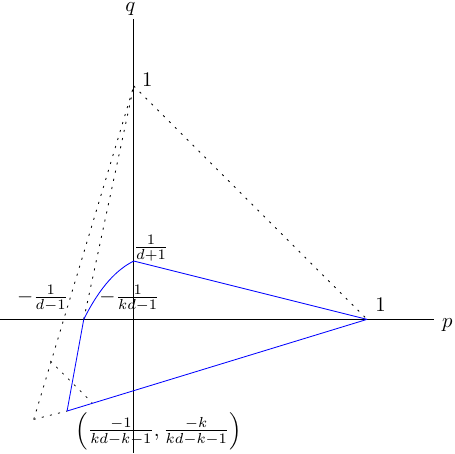} \\

    \vspace{2mm}
    \includegraphics[width=0.4\linewidth]{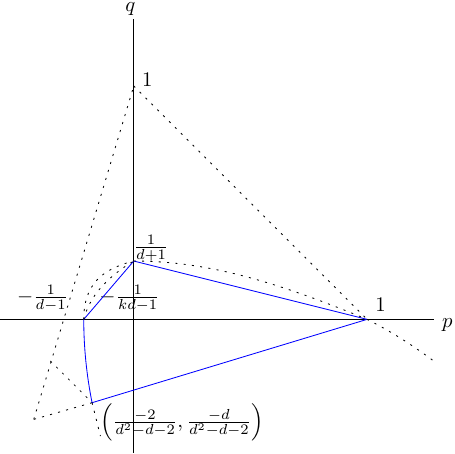}\;\;
    \includegraphics[width=0.4\linewidth]{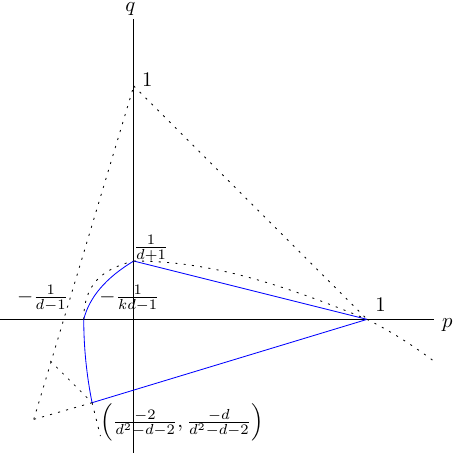}
    
    \caption{The $k$-Positivity regions $\mathbb{P}_k$ of $\Le_{p,q}$ in the case: (1) top left, (2) top right, (3) bottom left, and (4) bottom right.}
    \label{fig-kPos}
\end{figure}

\begin{theorem} \label{thm-SympkPos}
For a real number $u\geq 0$, define a quadratic polynomial 
\begin{equation} \label{eq-conic}
    f_u^{(d,k)}(x,y):= \big(1-x-(1+d)y \big) \big(1-(1-kd)x-y \big)+d^2uxy.
\end{equation}
Then for $d\geq 4$ even and $1<k<d$, the $k$-positivity region $\mathbb{P}_k\subset \Real^2$ consists of all points $(p,q)$ satisfying the following inequalities, depending on the value and parity of $k$:
\begin{enumerate}[leftmargin=*]
    \item The case $1<k\leq \frac{d}{2}$ and $k$ even:
    \begin{equation} \label{eq-Symp-kpos1e}
        p+(1+d)q\leq 1,  \quad  (1-kd)p+(1+d)q\leq 1, \quad (1-kd)p+q\leq 1, \quad  p+(1-d)q\leq 1.
    \end{equation}

    \item The case $1<k\leq \frac{d}{2}$ and $k$ odd:
    \begin{equation} \label{eq-Symp-kpos1o}
        p+(1+d)q\leq 1,  \quad  f_{k-1}^{(d,k)}(p,q)\geq 0, \quad (1-kd)p+q\leq 1, \quad  p+(1-d)q\leq 1.
    \end{equation}

    \item The case $\frac{d}{2}< k <d$ and $k$ even:
    \begin{equation} \label{eq-Symp-kpos2e}
        p+(1+d)q\leq 1, \quad  (1-kd)p+(1+d)q\leq 1, \quad f^{(d,k)}_{2k-d}(p,q)\geq 0, \quad  p+(1-d)q\leq 1.
    \end{equation}

    \item The case $\frac{d}{2}< k <d$ and $k$ odd:
    \begin{equation} \label{eq-Symp-kpos2o}
        p+(1+d)q\leq 1, \quad f_{k-1}^{(d,k)}(p,q)\geq 0, \quad  f^{(d,k)}_{2k-d}(p,q)\geq 0, \quad p+(1-d)q\leq 1.
    \end{equation}
\end{enumerate}
\end{theorem}

\begin{remark} \label{rmk-kPos}
Together with Fig. \ref{fig-kPos}, we provide additional details on Theorem \ref{thm-SympkPos}, supplementing the geometric description of the regions $\mathbb{P}_k$ for $k=2,\ldots, d-1$.
\begin{itemize}
    \item For cases (1) and (2), $\mathbb{P}_k$ has four extreme points
        $$(1,0), \quad \Big(0,\frac{1}{d+1}\Big), \quad \Big(\frac{-1}{kd-1},0\Big), \quad \Big(\frac{-1}{kd-k-1},\frac{-k}{kd-k-1}\Big)$$
    In case (1), the vertices are simply connected by straight line segments. In case (2), the points $\Big(0,\frac{1}{d+1}\Big)$ and $\Big(\frac{-1}{kd-1},0\Big)$ are connected along the curve $f_{k-1}^{(d,k)}(x,y)=0$, which is always a hyperbola (Lemma \ref{lem-Conic}).

    \item For cases (3) and (4), $\mathbb{P}_k$ has four extreme points
        $$(1,0), \quad \Big(0,\frac{1}{d+1}\Big), \quad \Big(\frac{-1}{kd-1},0\Big), \quad \Big(\frac{-2}{d^2-d-2},\frac{-d}{d^2-d-2}\Big).$$
    In case (3), the points $\Big(\frac{-1}{kd-1},0\Big)$ and $\Big(\frac{-2}{d^2-d-2},\frac{-d}{d^2-d-2}\Big)$ are connected along the hyperbola $f_{2k-d}^{(d,k)}(x,y)=0$ (Lemma \ref{lem-Conic}). In case (4), in addition to this, the points $\Big(0,\frac{1}{d+1}\Big)$ and $\Big(\frac{-1}{kd-1},0\Big)$ are connected along the hyperbola $f_{k-1}^{(d,k)}(x,y)=0$.

    \item The two hyperbolas $f_{k-1}^{(d,k)}(p,q)=0$ and $f_{2k-d}^{(d,k)}(p,q)=0$ coincide precisely when $k=d-1$. This explains the regions shown in Fig.~\ref{fig-4kPos}, where the boundary is described by a single hyperbola in the case $k=3$.
\end{itemize}
\end{remark}

The proof of Theorem \ref{thm-SympkPos} requires several additional steps, and we postpone it to \cref{app-kPos}. For the remainder, we discuss several byproducts of the above results. Most importantly, we can directly compare the $k$-positivity region $\mathbb{P}_k$ with the decomposability region $\mathbb{D}$, which reveals an explicit construction of a broad family of $(\frac{d}{2}-1)$-positive indecomposable linear maps. Note that this is the largest known value of $k$ for which $k$-positive indecomposable maps are proven to exist. Furthermore, to the best of our knowledge, this is the \textit{first explicit construction} attaining this maximal $k$, whereas previously such existence results were obtained only indirectly via duality with quantum entanglement.

\begin{corollary} \label{cor-Symp-kposIndec}
For every even integer $d\geq 4$, the set $\mathbb{P}_k^{(d)}\setminus \mathbb{D}^{(d)}$ is nonempty for all $k=1,\ldots, \frac{d}{2}-1$, whereas $\mathbb{P}_{d/2}^{(d)}\subset \mathbb{D}^{(d)}$. In particular, the family of linear maps
    $$\{\Le_{p,q}^{(d)}:(p,q)\in \mathbb{P}_{d/2-1}^{(d)}\setminus \ \mathbb{D}^{(d)}\} = \left\{\Le_{p,q}^{(d)}:\;\begin{array}{lll} p+q < -\frac{2+d}{d^2-d-2},\\ p+(1-d)q\leq 1, \\ (1-(d/2-1)d)p+q\leq 1 \end{array} \right\}$$
provides the explicit examples of $(\frac{d}{2}-1)$-positive indecomposable linear maps on $M_d(\Comp)$ while every $\frac{d}{2}$-positive $(S,S)$-covariant linear maps are decomposable.
\end{corollary}
\begin{proof}
It is straightforward to verify that the point $P=(-\frac{2}{d^2-3d}, -\frac{d-2}{d^2-3d})$ satisfies the inequalities above, and hence $P\in \mathbb{P}_{d/2-1}\setminus \mathbb{D}$. On the other hand, the boundary line segment of $\mathbb{P}_{d/2}$ connecting the vertices $(-\frac{2}{d^2-d-2},-\frac{d}{d^2-d-2})$ and $(-\frac{2}{d^2-2},0)$ lies entirely within the half-plane $x+y\geq -\frac{2+d}{d^2-d-2}$. This implies $\mathbb{P}_{d/2}\subset \mathbb{D}$.
\end{proof}

As a concrete example, the linear map
\begin{equation} \label{eq-(d/2-1)-BH}
    \Le_P:Z\in M_d(\Comp)\mapsto \frac{1}{d^2-3d}\big((d-2)\Tr( Z)I_d -2Z-(d-2)\Om_d Z^{\top}\Om_d^*\big)
\end{equation}
considered in the proof of Corollary \ref{cor-Symp-kposIndec} is $(\frac{d}{2}-1)$-positive and indecomposable. 
Note that the point $P$ corresponds to the unique extreme point of $\mathbb{P}_{d/2-1}$ that is not contained in $\mathbb{D}$. For each $1\leq k < d/2$, the similar property holds for the point $(\frac{-1}{kd-k-1},\frac{-k}{kd-k-1})\in \mathbb{P}_k$, We refer to the associated linear map as the \emph{$k$-Breuer-Hall map}, defined by:
\begin{equation} \label{eq-k-BH}
   \Le^{\rm BH}_k:=\Le_{\frac{-1}{kd-k-1},\frac{-k}{kd-k-1}}: Z\mapsto \frac{1}{kd-k-1} \big(k \Tr(Z)I_d - Z - k\Om_d Z^{\top}\Om_d^* \big).
\end{equation}
Note that $\Le^{\rm BH}_1 = \Le^{\rm BH}$ as introduced in Example \ref{ex-BH}. These maps can be used to detect quantum states that are PPT but have Schmidt number larger than $k$. In the following, we establish a much stronger property, namely the optimality of the (Choi matrix of the)
$k$-Breuer-Hall map in the sense of \cite{LKCH00,SBL01,Bre06}. Recall that a $k$-positive (and not $(k+1)$-positive) indecomposable map $\Le$ is called \textit{optimal}, as an \textit{indecomposable $(k+1)$-Schmidt number witness}, if there is no other $k$-positive indecomposable linear map $\Le'$ that detects strictly more PPT states with Schmidt number larger than $k$:
    $$\{\rho\in \PPT: \Tr(\rho C_{\Le})<0\} \subseteq \{\rho\in \PPT: \Tr(\rho C_{\Le'})<0\} \implies \Le'=c\Le \text{ for some $c>0$}.$$
It is further shown (we refer to \cite[Theorem 1(b)]{LKCH00} and \cite{BCH+02}) that $\Le$ is optimal if and only if there is no nonzero decomposable map $\Le''$ such that the map $\Le'=\Le-\Le''$ is $k$-positive. The construction of such linear maps has remained largely beyond reach, except in the case $k=1$, but here the group symmetry method again leads to a substantial simplification.

\begin{theorem} \label{thm-Symp-opt}
For any even $d\geq 4$ and for any $k=1,\ldots, \frac{d}{2}-1$, the $k$-Breuer-Hall map $\Le_k^{\rm BH}$ is an optimal $k$-positive indecomposable linear map.
\end{theorem}
\begin{proof}
Suppose $\Le_k^{\rm BH}=\Le'+\Le''$ for some $k$-positive map $\Le'$ and a decomposable map $\Le''$. As in the proof of Theorem \ref{thm-Symp-PosDec}, we take the $(S,S)$-twirling $\T_{S,S}$ to have 
\begin{align*}
    \Le_k^{\rm BH} &= \T_{S,S}\Le' + \T_{S,S}\Le''\\
    &=\lambda \Le_{p_1,q_1}+(1-\lambda)\Le_{p_2,q_2} = \Le_{\lambda p_1+(1-\lambda)p_2, \lambda q_1+(1-\lambda)q_2}
\end{align*}
for some $(p_1,q_1)\in \mathbb{P}_k$, $(p_2,q_2)\in \mathbb{D}$, and $\lambda\in [0,1]$. By recalling the geometry of $\mathbb{P}_k$ and $\mathbb{D}$ as in Figs. \ref{fig-PosDec} and \ref{fig-kPos}, the point $\big(\frac{-1}{kd-k-1},\frac{-k}{kd-k-1} \big)$ is a convex combination of $(p_1,q_1)$ and $(p_2,q_2)$ only if $\lambda=1$ and $(p_1,q_1) = \big(\frac{-1}{kd-k-1},\frac{-k}{kd-k-1} \big)$. In other words, we have $\T_{S,S}\Le'=\Le_k^{\rm BH}$ and $\T_{S,S}\Le''=0$. On the other hand, the trace-preserving property of $\mathcal{T}_{\overline{S}\otimes S}$ and Proposition \ref{prop-Symm-general} (3) imply that
    $$\Tr(C_{\Le''})=\Tr(\T_{\overline{S}\otimes S}C_{\Le''})=\Tr(C_{\T_{S,S}\Le''}) = 0.$$
Since $C_{\Le''}$ is block-positive, we conclude that $\Le''=0$ and therefore $\Le'=\Le_k^{\rm BH}$ (see \cite[Comment [30]{LKCH00}).
\end{proof}

\begin{remark} \label{rmk-SympOptAtomic}
We recover the result of~\cite{Bre06} in the case $k=1$ in a concise manner by exploiting symplectic group symmetry. In fact, the same method used in the proof of Theorem~\ref{thm-Symp-opt} enables us to obtain further characterizations of linear maps of interest, as summarized below.
\begin{enumerate}
    \item Without imposing the \emph{indecomposability} condition, the optimal $k$-positive linear maps within $\Cov(S,S)$ corresponds precisely to, up to positive multiplicative constant,
    \begin{align*}
        \textstyle\{\Le_{p,q}: (1-d)p+q=1, \;\; \frac{-1}{d-2}\leq p\leq 0\} & \quad \text {for $k=1$,}\\
        \textstyle\{\Le_{p,q}: (1-kd)p+q=1, \;\; \frac{-1}{kd-k-1}\leq p\leq \frac{-1}{kd-1}\} & \quad \text {for $1<k<d$ and $k$ even,}\\
        \{\Le_k^{\rm BH}\} \cup\{\Le_{p,q}: (p,q)\in \partial\, \mathbb{P}_k\setminus \mathbb{P}_d,\;\; p+(1-d)q\neq 1\} & \quad \text {for $1<k<d$ and $k$ odd,}
    \end{align*}
    where $\partial\, \mathbb{P}_k$ denotes the boundary of $\mathbb{P}_k$. For example, the \textit{$k$-reduction map}
        $$\Le_k^{\rm red}:= \Le_{\frac{-1}{kd-1},0}:Z\mapsto \frac{1}{kd-1} (k \Tr(Z) I_d - Z)$$
    is \emph{optimal} in the sense that there exists no nonzero CP map $\Phi$ such that $\Le'=\Le_k^{\rm red}-\Phi$ remains $k$-positive, as also established in \cite{SBL01}.

    \item For $k=1,\ldots, d/2-1$, the map $\Le_{p,q}$ is $k$-positive but cannot be written in the form $\Le'+\top\circ \Le''$ for any $(k+1)$-positive linear maps $\Le'$ and $\Le''$ if and only if 
        $$p+q<-\frac{k+2}{(k+1)d-k-2}, \quad p+(1-d)q\leq 1, \quad (1-kd)p+q\leq 1,$$
    or equivalently $(p,q)\in \mathbb{P}_k\setminus {\rm conv}(\mathbb{P}_{k+1}\cup \mathbb{P}_{k+1}^{\top})$, where $\mathbb{P}_{k+1}^{\top}=\{(p,q):(q,p)\in \mathbb{P}_{k+1}\}$ denotes the reflection of $\mathbb{P}_{k+1}$ across the line $y=x$. Note that the $k$-Breuer-Hall map $\Le_k^{\rm BH}$ satisfies this property, and the set $\mathbb{P}_k\setminus {\rm conv}(\mathbb{P}_{k+1}\cup \mathbb{P}_{k+1}^{\top})$ has nonempty interior. When $k=1$, this property is known as \emph{atomicity}. In particular, the Breuer–Hall map is atomic, recovering the result of \cite{CK08}.
\end{enumerate}
\end{remark}

\begin{remark}
Motivated by properties of finite-dimensional operator systems, a recent work \cite{aubrun2024completely} introduced the quantities
\begin{align*}
    r_k(M_d)&:=\sup\{\|\phi\|_{cb}:\phi:\mathcal{S}\to M_d(\Comp) \text{ unital $k$-positive}\},\\
    d_k(M_d)&:=\sup\{\|\phi\|_{cb}:\phi: M_d(\Comp)\to \mathcal{S} \text{ unital $k$-positive}\}
\end{align*}
for $1\leq k\leq d$, where the supremum ranges over all operator systems $\mathcal{S}$ and $\|\phi\|_{cb}=\sup_{n\geq 1} \|\id_n\otimes \phi\|$ denotes the \emph{completely bounded norm}. The value of $r_k(M_d)$ is computed exactly as $\frac{2d-k}{k}$, while $d_k(M_d)$ admits the bounds
    $$\max\left( 1+\frac{2(d-k)}{d(kd-1)} , c\sqrt{\frac{d}{k}} \right) \leq d_k(M_d) \leq \frac{2d-k}{k}$$
for a universal constant $c>0$. Moreover, the explicit lower bound is attained by the $k$-reduction map $\Le_k^{\rm red}=\Le_{\frac{-1}{kd-1},0}$. It is therefore natural to ask whether our family $\{\Le_{a,b}\}_{(a,b)\in \mathbb{P}_k}$ yields an improved lower bound. Although we do not include the details here, one finds that the supremum $\sup_{(a,b)\in \mathbb{P}_k} \|\Le_{a,b}\|_{cb}$ is still attained by $\Le_k^{\rm red}$. In particular, it is computed that
    $$\|\Le_k^{\rm BH}\|_{cb} = 1+\frac{2(d-2k)}{d(kd-k-1)} < 1+ \frac{2(d-k)}{d(kd-1)} = \|\Le_k^{\rm red}\|_{cb}$$
for $1\leq k\leq d/2$. This shows that extremality in $k$-positive indecomposability in general does not lead to an improvement of the completely bounded norm.
\end{remark}

Finally, thanks to the unitary equivalence between $(S,S)$- and $(S,\overline{S})$-covariant linear maps (Proposition \ref{prop-Symp-equivalence}), all our results carry over to the class of $(S,\overline{S})$-covariant linear maps. We summarize the corresponding results as follows.

\begin{theorem}
Let $d\geq 4$ be an even integer, and for $(p,q)\in \Real^2$, consider the linear map
    $$\Le_{p,q}^{\top}:=\top\circ \Le_{p,q}:Z\in M_d(\Comp)\mapsto (1-p-q)\frac{ \Tr(Z)}{d}I_d +pZ^{\top}+ q\Om_d Z \Om_d^*.$$
Then
\begin{enumerate}
    \item $\Le_{p,q}^{\top}$ is $k$-positive ($1\leq k\leq d$) if and only if $(q,p)\in \mathbb{P}_k^{(d)}$;

    \item $\Le_{p,q}^{\top}$ is decomposable if and only if $(q,p)\in \mathbb{D}^{(d)}$ (if and only if $(p,q)\in \mathbb{D}^{(d)}$).
\end{enumerate}
In particular, $(\frac{d}{2}-1)$-positive indecomposable linear maps can be constructed within the class $\Cov(S,\overline{S})$ while every $\frac{d}{2}$-positive $(S,\overline{S})$-covariant linear map is decomposable.
\end{theorem}

Consequently, by Corollary~\ref{cor-symm-PPTSch-equiv}, we conclude that there exist PPT states with Schmidt number $d/2$ within both classes $\Inv(S\otimes \overline{S})$ and $\Inv(S\otimes S)$ while every PPT state in these classes has Schmidt number \emph{at most} $d/2$. In the following section, we provide the complete \textit{analytic description} the Schmidt number regions of these classes.

% Schmidt number------

\section{PPT and Schmidt number of quantum states} \label{sec-SympSch}

\subsection{Symplectic invariant states}

Now we are ready to proceed with the computation of the Schmidt number of all $S\otimes \overline{S}$-invariant quantum states $\{\rho_{a,b}\}_{a,b\in \mathbb{P}_d}$ as introduced and discussed in \cref{eq-SympInv} and Corollary \ref{cor-Symp-psd}). As before, let us denote by 
    $$\mathbb{S}_k=\mathbb{S}_k^{(d)}:=\{(a,b)\in \Real^2:\rho_{a,b}\in \sch_k\},$$
and we aim to describe the geometric and algebraic structures of $\mathbb{S}_1\subset \mathbb{S}_2\subset \cdots \subset \mathbb{S}_d=\mathbb{P}_d$. Thanks to \cref{thm-DualSymm} and the convexity, it suffices to utilize the $k$-positive $(S,S)$-covariant linear maps $\{\Le_{p,q}\}_{(p,q)\in {\rm ext}(\mathbb{P}_k)}$ as Schmidt number witnesses for $\rho_{a,b}$, where ${\rm ext}(\mathbb{P}_k)$ denotes the set of extreme points of $\mathbb{P}_k$. Specifically, we have
\begin{align}
    &{\rm SN}(\rho_{a,b})\leq k \nonumber \\
    & \iff \Tr(\rho_{a,b}C_{\Le_{p,q}}) = \Tr(\rho_{a,b}\,\rho_{p,q}) \geq 0 \quad\quad\;\;\;  \text{ for all }(p,q)\in {\rm ext}(\mathbb{P}_k) \nonumber\\
    & \iff \begin{pmatrix} p & q \end{pmatrix}\begin{pmatrix}d-1 & -1 \\ -1 & d-1\end{pmatrix}\binom{a}{b}\geq -\frac{1}{d+1} \text{ for all }(p,q)\in {\rm ext}(\mathbb{P}_k) \label{eq-Symp-paring}. 
\end{align}
Let us begin with the comparison between separability and PPT property, which are rather simple to characterize (We refer to Fig. \ref{fig-SepPPT} for their illustration).

\begin{theorem} \label{thm-SympSepPPT}
The separability region $\mathbb{S}_1=\mathbb{S}_1^{(d)}$ and the PPT region $\mathbb{T} = \mathbb{T}^{(d)}$ of $\rho_{a,b}$ is characterized as
\begin{align} 
    \mathbb{S}_1&=\left\{(a,b)\in \Real^2: \;\begin{array}{lll} (1-d)a+b\leq \frac{1}{d+1}, \\ a+ (1-d)b\leq \frac{1}{d+1},\end{array} \quad a+b\leq \frac{1}{d+1}\right\}, \label{eq-Symp-sep} \\ \mathbb{T}&=\left\{(a,b)\in \Real^2: \;\begin{array}{lll} (1-d)a+b\leq \frac{1}{d+1}, \quad a+(1+d)b\leq 1,\\ a+ (1-d)b\leq \frac{1}{d+1}, \quad (1+d)a+b\leq 1\end{array}\right\}, \label{eq-Symp-PPT}
\end{align}
respectively. In particular, $\rho_{a,b}$ is \emph{PPT entangled} in the parameter region  $\begin{cases}
    a+b > \frac{1}{d+1},\\
    a+(1+d)b\leq 1,\\
    (1+d)a+b\leq 1.
\end{cases}$
\end{theorem}
\begin{proof}
By Theorem \ref{thm-Symp-PosDec}, the set $\mathbb{P}_1$ has three extreme points $(1,0)$, $(0,1)$, and $(-\frac{1}{d-2},-\frac{1}{d-2})$. Therefore, we have Eq. \eqref{eq-Symp-sep} from Eq. \eqref{eq-Symp-paring} with these three points. On the other hand, since $\rho_{a,b}^{\Gamma}=C_{\top\circ \Le_{a,b}}$, $\rho_{a,b}$ is PPT if and only if $(a,b)\in \mathbb{P}_d$ and $(b,a)\in \mathbb{P}_d$ which is equivalent to Eq. \eqref{eq-Symp-PPT}.
\end{proof}

\begin{figure}
    \centering

    \includegraphics[width=0.31\linewidth]{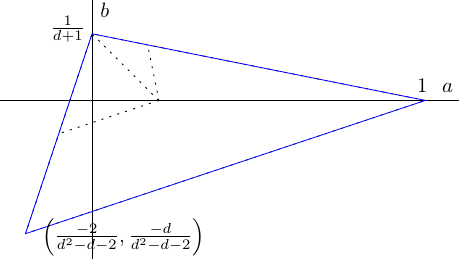} \;
    \includegraphics[width=0.31\linewidth]{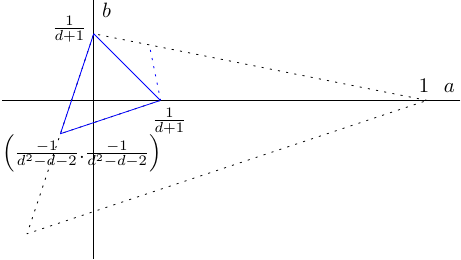} \;
    \includegraphics[width=0.31\linewidth]{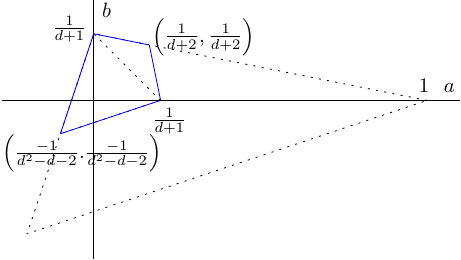}\;\;

    \caption{The regions where $\rho_{a,b}$ is a quantum state (left), separable (middle), and PPT (right).}
    \label{fig-SepPPT}
\end{figure}

Although the remaining sets $\mathbb{P}_k$ ($1<k<d$) are also described explicitly in \cref{thm-SympkPos}, a difficulty remains in solving the condition \eqref{eq-Symp-paring}: except for the cases $k=1$ and $k$ even with $1<k\leq \frac{d}{2}$, the set
$\mathbb{P}_k$ has infinitely many extreme points. Nevertheless, we overcome this obstacle by exploiting the notion of \textit{polarity} from projective geometry \cite{Cox03,TextBK}. The detailed method along these lines is deferred to Appendix~\ref{app-SN-hard}. 

We now state the main result of this section and discuss the consequences relevant to our work. We also provide a brief geometric depiction in Fig. \ref{fig-Sch}.

\begin{figure}
    \centering

    \includegraphics[width=0.4\linewidth]{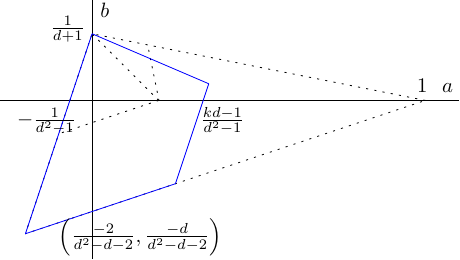} \;
    \includegraphics[width=0.4\linewidth]{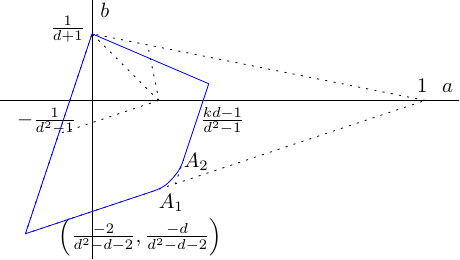} \\

    \vspace{2mm}
    \includegraphics[width=0.4\linewidth]{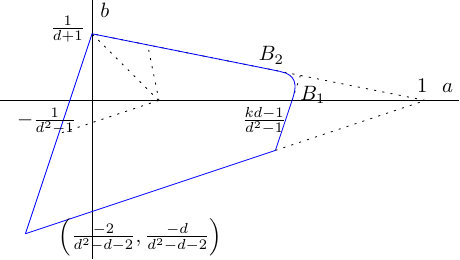}\;\;
    \includegraphics[width=0.4\linewidth]{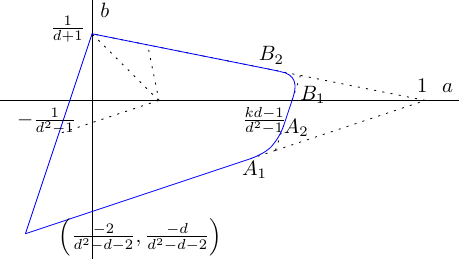}
    
    \caption{The Schimidt number regions $\mathbb{S}_k$ of $\rho_{a,b}$ in the case: (1) top left, (2) top right, (3) bottom left, and (4) bottom right.}
    \label{fig-Sch}
\end{figure}

\begin{theorem} \label{thm-SympSch}
Let $d\geq 4$ be even and $1<k<d$. Then there exist two linear polynomials $l_j(x,y)=l_j^{(d,k)}(x,y)$ ($j=1,2$) and two quadratic polynomials $g_j(x,y)=g_j^{(d,k)}(x,y)$ ($j=1,2$) such that the following are satisfied.
\begin{itemize}
    \item $l_1$ and $l_2$ are explicitly given by
    \begin{align}
        l_1(x,y) & = (k+d-1)x - (kd-k+1)y  -\frac{2kd+k-d-1}{d+1}, \label{eq-Symp-line1}\\
        l_2(x,y) &= (3d-k-3)x+(kd-k-3)y - \frac{d^2+kd-k-3}
        {d+1}. \label{eq-Symp-line2} 
    \end{align}

    \item The regions $g_1(x,y)\leq 0$ and $g_2(x,y)\leq 0$ are both filled ellipses inscribed in the same parellelogram determined by four lines
    {\small
        $$x+(1-d)y=\frac{1}{d+1}, \quad x+(1-d)y=1, \quad (1-d)x+y=\frac{1}{d+1}, \quad (1-d)x+y = -\frac{kd-1}{d+1}.$$
    Furthermore, $g_1=g_2$ when $k=d-1$.
    }

    \item The line $l_1(x,y)=0$ intersects the ellipse $g_1(x,y)=0$ in the fourth quadrant at two points
    \begin{center}
        $A_1=\left(\frac{kd-d-k-1}{d^2-d-2}, \frac{-d+k-1}{d^2-d-2}\right), \quad A_2=\left(\frac{k^2d-2k-k^2+1}{k(d^2-d-2)}, \frac{(k-1)(k-d+1)}{k(d^2-d-2)}\right),$
    \end{center}
    and $l_2(x,y)=0$ and $g_2(x,y)=0$ intersect in the first quadrant at two points
    \begin{center}
        $B_1=\left(\frac{k^2d-3k-k^2+d}{k(d^2-d-2)}, \frac{(k-d)(k-d+1)}{k(d^2-d-2)}\right), \quad B_2=\left(\frac{d}{3d-2k}, \frac{2d-2k}{(d+1)(3d-2k)}\right).$
    \end{center}

    \item The Schmidt number region $\mathbb{S}_k\subset \Real^2$ consists of all points $(a,b)$ satisfying the following inequalities, depending on the value and parity of $k$:
\end{itemize}
\begin{enumerate}[leftmargin=*]
    \item The case $1<k\leq \frac{d}{2}$ and $k$ even:
    \begin{equation} \label{eq-Symp-SN1e}
        -\frac{kd-1}{d+1}\leq (1-d)a+b\leq \frac{1}{d+1}, \quad  a+(1-d)b\le 1, \quad  \frac{d-k-1}{kd-k-1}a+b\leq \frac{1}{d+1}.
    \end{equation}

    \item The case $1<k\leq \frac{d}{2}$ and $k$ odd:
    \begin{equation} \label{eq-Symp-SN1o}
        \text{All the inequalities in Eq. \eqref{eq-Symp-SN1e} and } \left[ l_1(x,y)\leq 0 \text{ or } \begin{cases}
            l_1(x,y)\geq 0, \\ g_1(x,y)\leq 0.
        \end{cases}\right]
    \end{equation}

    \item The case $\frac{d}{2}< k <d$ and $k$ even:
    \begin{equation} \label{eq-Symp-SN2e}
        -\frac{kd-1}{d+1}\leq (1-d)a+b \leq \frac{1}{d+1}, \quad a + (1 \pm d)b \leq 1, \quad \left[ l_2(x,y)\leq 0 \text{ or } \begin{cases}
            l_2(x,y)\geq 0, \\ g_2(x,y)\leq 0.
        \end{cases}\right]
    \end{equation}

    \item The case $\frac{d}{2}< k <d$ and $k$ odd:
    \begin{equation} \label{eq-Symp-SN2o}
        \text{All the inequalities in Eq. \eqref{eq-Symp-SN2e} and } \left[ l_1(x,y)\leq 0 \text{ or } \begin{cases}
            l_1(x,y)\geq 0, \\ g_1(x,y)\leq 0.
        \end{cases}\right]
    \end{equation}
\end{enumerate}
\end{theorem}
\begin{proof}[Proof of (1)]
The proof in this case is analogous to the proof of Theorem \ref{thm-SympSepPPT} as the set $\mathbb{P}_k$ in this case has precisely four extreme points
    $$(1,0), \quad \Big(0,\frac{1}{d+1}\Big), \quad \Big( -\frac{1}{kd-1},0\Big), \quad \Big(-\frac{1}{kd-k-1},-\frac{k}{kd-k-1}\Big),$$
as shown in \cref{thm-SympkPos} (1) and Remark \ref{rmk-kPos}. The remainder of the proof is considerably more involved, and we therefore defer it to Appendix~\ref{app-SN-hard}.
\end{proof}

One of the most important consequences of Theorem \ref{thm-SympSch} is that, combined with Theorem \ref{thm-SympSepPPT}, PPT states with high entanglement-dimensionality, that is, ones with high Schmidt numbers are broadly found and even characterized within the class of $S\otimes \overline{S}$-invariant quantum states. Furthermore, the same results holds within the class with $S\otimes S$-symmetry.

\begin{corollary} \label{cor-Symp-PPTSch}
For every even $d\geq 4$, the set $\mathbb{T}^{(d)}\setminus \mathbb{S}_k$ is nonempty for all $k=1,\ldots, \frac{d}{2}-1$ while $\mathbb{T}^{(d)}\subset \mathbb{S}_{d/2}^{(d)}$. Consequently, the family of quantum states
    $$\{\rho_{a,b}^{(d)}:(a,b)\in \mathbb{T}^{(d)}\setminus \mathbb{S}_{d/2-1}^{(d)}\} = \left\{\rho_{a,b}^{(d)}:\;\begin{array}{lll} \frac{1}{d-3} a+b >\frac{1}{d+1} ,\\ a+(1+d)b\leq 1, \\ (1+d)a+b\leq 1 \end{array} \right\}$$
provides the exact regions of $S\otimes \overline{S}$-invariant PPT states with Schmidt number \textit{precisely} $d/2$. Furthermore, every PPT $S\otimes \overline{S}$-invariant state has Schmidt number at most $d/2$.
\end{corollary}

\begin{theorem} \label{thm-SSInv-Sch}
For $d\geq 4$ even, consider the quantum states
    $$\rho_{a,b}^{\Gamma} = \frac{1-a-b}{d^2}I_{d^2}+\frac{a}{d}F_d+ b(I_d\otimes \Om_d)|\om_d\ra\la \om_d|(I_d\otimes \Om_d)^* \in \Inv(S\otimes S), \quad(b,a)\in \mathbb{P}_d.$$
(see Proposition \ref{prop-Symp-CP}). Then one has $\SN(\rho_{a,b}^{\Gamma}) = \SN(\rho_{b,a})$, and therefore,
\begin{enumerate}
    \item $\rho_{a,b}^{\Gamma}$ is PPT if and only if $(b,a)\in \mathbb{T}$ (if and only if $(a,b)\in \mathbb{T}$);
    
    \item $\SN(\rho_{a,b}^{\Gamma})\leq k$ if and only if $(b,a)\in \mathbb{S}_k$.
\end{enumerate}
In particular, PPT quantum states with Schmidt number $d/2$ can be constructed within the class $\Inv(S\otimes S)$ while every $S\otimes S$-invariant PPT state has Schmidt number at most $d/2$.
\end{theorem}
\begin{proof}
This follows from Theorem \ref{thm-SympSch} and the local-unitary equivalence property (Proposition \ref{prop-Symp-equivalence}).
\end{proof}

When $(a,b)=(\frac{1}{d+2},\frac{1}{d+2})\in \mathbb{T}\setminus \mathbb{S}_{d/2}$, we recover the result of \cite{PV19} and, moreover, obtain that $\SN(\rho^{\rm PV})=d/2$, as advertised in Example \ref{ex-PV19}. Furthermore, our results allow us to identify additional interesting classes of quantum states by comparing the Schmidt numbers of a PPT state $\rho$ and its partial transpose $\rho^{\Gamma}$. We discuss two such classes below.

\begin{corollary} \label{cor-PPT-SNDiff}
Let $d\geq 4$ be an even integer. Then
\begin{enumerate}
    \item There exists a $d\otimes d$ quantum state $\rho$ (both within $\Inv(S\otimes \overline{S})$ and $\Inv(S\otimes S)$) which is PPT and satisfies
        $$\SN(\rho)-\SN(\rho^{\Gamma}) = d/2-2.$$

    \item There exists a $d\otimes d$ quantum state $\rho$ (both within $\Inv(S\otimes \overline{S})$ and $\Inv(S\otimes S)$) whose partial transpose $\rho^{\Gamma}$ is local-unitary equivalent to $\rho$ while $\SN(\rho) = d/2$.
\end{enumerate}
\end{corollary}
\begin{proof}
Let us first consider an $S\otimes \overline{S}$ invariant quantum state
    $$\rho = \rho_{\frac{1}{2d-2},\frac{2d-3}{2d^2-2}} = \frac{1}{2d(d^2-1)} \big((2d-3)I_{d^2} + d(d+1)|\om_d\ra\la \om_d| + (2d-3)(I_d\otimes \Om_d)F_d (I_d\otimes \Om_d) \big).$$
Then it is simple to check that $\rho$ is PPT, $\SN(\rho)=d/2$, and $\SN(\rho^{\Gamma})=2$ from \cref{thm-SympSch,thm-SSInv-Sch}, which establishes the first assertion. In fact, the parameter region of $(a,b)$ for which $\rho_{a,b}$ has these properties has a nonempty interior.

For the second assertion, we may consider a family of $S\otimes \overline{S}$-invariant states
    $$\rho=\rho_{a,a} = \frac{1-2a}{d^2}I_{d^2}+a|\om_d\ra\la \om_d|+\frac{a}{d}(I_d\otimes \Om_d)F_d (I_d\otimes \Om_d)^*.$$
Then Proposition \ref{prop-Symp-equivalence} implies that $\rho_{a,a}^{\Gamma}$ is local-unitary equivalent to $\rho_{a,a}$. Furthermore, by Corollary \ref{cor-Symp-PPTSch}, we have $\SN(\rho_{a,a})=d/2$ in the range $\frac{d-3}{d^2-d-2}<a\leq \frac{1}{d+2}$.
\end{proof}

\subsection{General Schmidt number criteria}

The detection of quantum states with high Schmidt number is not limited to states possessing symplectic group symmetry. First, our $k$-positivity classifications in Theorems \ref{thm-Symp-PosDec} and \ref{thm-SympkPos} imply the general necessary conditions for quantum states belonging to the set $\sch_k$.

\begin{proposition} \label{prop-Sch-general}
Let $d\geq 4$ be an even integer and let $\rho=\rho_{AB}$ be a $d\otimes d$ bipartite quantum state. 
\begin{enumerate}

    \item If $1\leq k\leq \frac{d}{2}$ and $\SN(\rho)\leq k$, then for every skew-symmetric unitary matrix $V\in U(d)$,
    \begin{equation} \label{eq-kBre}
        \rho_{AB}+k(I_d\otimes V)\rho_{AB}^{\Gamma}(I_d\otimes V)^*\leq k (\rho_A\otimes I_d),
    \end{equation}
    where $\rho_A:=(\id_d\otimes \Tr)(\rho_{AB})$.

    \item For a skew-symmetric unitary matrix $V$, define two real parameters
    \begin{align}
        a=a(V) &= \frac{1}{d^2-d-2}(-1 +  \Tr(F_d^{V}\rho) + (d^2-d) \la \om_d|\rho|\om_d\ra), \label{eq-para1}\\
        b = b(V) &=  \frac{1}{d^2-d-2}(-1 +(d-1) \Tr(F_d^{V}\rho) + d \la \om_d|\rho|\om_d\ra), \label{eq-para2}
    \end{align}
    where $F_d^V=\Ad_{I_d\otimes V}(F_d)$. If $1\leq k\leq d$ and  $\SN(\rho)\leq k$, then $(a,b)\in \mathbb{S}_k$.
\end{enumerate}
\end{proposition}
\begin{proof}
By Remark \ref{rmk-V-equiv}, it suffices to check only for $V=\Om_d$. First,
(1) is equivalent to $(\id_d\otimes \Le_k^{\rm BH})(\rho)\geq 0$ where $\Le_k^{\rm BH}$ is the $k$-Breuer-Hall map defined in Eq. \eqref{eq-k-BH}. For the assertion (2), it suffices to show that $\T_{S\otimes \overline{S}}(\rho) = \E_{S\in Sp(d)}\big[ (S\otimes S)\rho (S\otimes S)^* \big] = \rho_{a,b}$ since the $S\otimes \overline{S}$-twirling operation preserves the membership in $\sch_k$ (Proposition \ref{prop-twirling-preserving}). Indeed, let us recall the three mutually orthogonal projections $\Pi_j^{S\overline{S}}$ ($j=1,2,3$) defined in Eq. \eqref{eq-SSbar-proj}. By denoting $d_j:={\rm rank}(\Pi_j^{SS})$ and $\la X\ra_{\rho}:=\Tr(X\rho)$ for a Hermitian matrix $X$, Proposition \ref{prop-twirlformula} implies that
{\small
\begin{align*}
    &\T_{S\otimes \overline{S}}(\rho) = \sum_{j=1}^3 \frac{1}{d_j} \la \Pi_j^{S\overline{S}}\ra_{\rho}\, \Pi_j^{S\overline{S}}\\
    &= \frac{d^2}{2}\left(\frac{\la \Pi_2^{S\overline{S}}\ra_{\rho}}{d_2}+\frac{\la \Pi_3^{S\overline{S}}\ra_{\rho}}{d_3} \right) (I_{d^2}/d^2) + \left(\frac{\la \Pi_1^{S\overline{S}}\ra_{\rho}}{d_1}-\frac{\la \Pi_3^{S\overline{S}}\ra_{\rho}}{d_3} \right) |\om_d\ra\la \om_d| + \frac{d}{2} \left(\frac{\la \Pi_2^{S\overline{S}}\ra_{\rho}}{d_2}-\frac{\la \Pi_3^{S\overline{S}}\ra_{\rho}}{d_3} \right) (F_d^{\Om}/d).
\end{align*} \normalsize}
From $d_1=1$, $d_2=\frac{d^2+d}{2}$, and $d_3=\frac{d^2-d-2}{2}$ it is straightforward to check that $a=\frac{\la \Pi_1^{S\overline{S}}\ra_{\rho}}{d_1}-\frac{\la \Pi_3^{S\overline{S}}\ra_{\rho}}{d_3}$ and $b=\frac{d}{2} \Big(\frac{\la \Pi_2^{S\overline{S}}\ra_{\rho}}{d_2}-\frac{\la \Pi_3^{S\overline{S}}\ra_{\rho}}{d_3} \Big)$ coincide with  Eqs. \eqref{eq-para1} and \eqref{eq-para2}.
\end{proof}

\begin{remark}
When $\rho$ is a PPT state, then the condition Eq. \eqref{eq-kBre} is stronger than the \textit{$k$-reduction criterion} \cite{HH99,TH00}:
\begin{equation} \label{eq-kRed}
    \SN(\rho_{AB})\leq k \implies \rho_{AB}\leq k(\rho_A\otimes I_d),
\end{equation}
which follows by applying $k$-reduction map $\Le_k^{\rm red}=\Le_{\frac{-1}{kd-1},0}:Z\mapsto \frac{1}{kd-1}(k\Tr(Z) I_d-Z).$ Furthermore, when $k=1$, the condition reduces to the entanglement criterion from \cite{Bre06}:
    $$\rho_{AB}\in \SEP \implies \rho_{AB}+(I_d\otimes V)\rho_{AB}^{\Gamma}(I_d\otimes V)^*\leq (\rho_A\otimes I_d).$$
\end{remark}

Most importantly, Proposition \ref{prop-Sch-general} allows us to construct much wider class of PPT states with high Schmidt number. 

\begin{corollary} \label{cor-PPTSchOpt1}
Suppose $(a,b)\in \mathbb{T}\setminus \mathbb{S}_{d/2-1}$ as in Corollary \ref{cor-Symp-PPTSch}. Then for any unit vectors $v_1,\ldots, v_m\in \Comp^d$ and probability mass $(p_0,p_1,\ldots, p_m)$ (i.e., $p_j\geq 0$ and $\sum_{j=0}^m p_j=1$) with $p_0>0$, the $d\otimes d$ quantum state
    $$\rho_{AB}=p_0 \rho_{a,b}+\sum_{j=1}^m p_j |v_j\ra\la v_j|\otimes |\Om_d v_j\ra\la \Om_d v_j|$$
is PPT and $\SN(\rho_{AB})=d/2$.
\end{corollary}
\begin{proof}
The facts that $\rho_{AB}$ is PPT and $\SN(\rho_{AB})\leq d/2$ are clear. We claim that $\rho$ does not satisfies the condition Eq. \eqref{eq-kBre} with $k=d/2-1$. Indeed, for each $k$, one has
\begin{align*}
    &\la \om_d| k (\rho_A\otimes I_d) -\rho_{AB}-k(I_d\otimes \Om_d)\rho_{AB}^{\Gamma}(I_d\otimes \Om_d)^* |\om_d\ra \\
    &= (kd-k-1) \Tr\big( |\om_d\ra\la \om_d|\, (\id_d\otimes \Le_{k}^{\rm BH})(\rho_{AB})\big)\\
    &= (kd-k-1) \Tr\big( (\id_d\otimes \Le_{k}^{\rm BH})(|\om_d\ra\la \om_d|)\,\rho_{AB}\big)\\
    &= (kd-k-1) \Tr\big( C_{\Le_{k}^{\rm BH}}\,\rho_{AB}\big)
\end{align*}
where the second equality follows from the observation that $(\Le_{a,b})^*=\Le_{a,b}$ for all $a,b\in \Real$, where $\Le^*$ denotes the dual map of $\Le$. Note that the condition $\frac{1}{d-3}a+b>\frac{1}{d+1}$ implies that
    $$\Tr(C_{\Le_{d/2-1}^{\rm BH}} \rho_{a,b}) = \Tr(\rho_{\frac{-2}{d^2-3d},\frac{-d+2}{d^2-3d}}\;\rho_{a,b})<0,$$
by comparison with Eq. \eqref{eq-Symp-paring}. Furthermore, for all unit vector $v\in \Comp^d$ and for any $k$, we have
    $$\Tr(C_{\Le_k^{\rm BH}}\, |v\otimes \Om_dv\ra\la v\otimes \Om_d v|) = \frac{1}{d(kd-k-1)}(k- \big|\la \overline{v}|\Om_dv\ra \big|^2 - k|\la v|v\ra|^2) = 0.$$
Putting everything together, we conclude that $\Tr\big( C_{\Le_{k}^{\rm BH}}\,\rho_{AB}\big)<0$ and therefore $\SN(\rho_{AB})\geq d/2$.
\end{proof}

We conclude this section by suggesting further examples of PPT states that may improve the current best-known bounds on the Schmidt number.

\begin{corollary} \label{cor-PPTSchOpt2}
Suppose $(a,b)$ belongs to the interior of the set $\mathbb{T}\setminus  \mathbb{S}_{d/2-1}$. Then for sufficiently small $\eps>0$, the state
    $$\rho_{AB} = (1-\eps) \rho_{a,b} + \eps |\om_d^{\Om}\ra\la \om_d^{\Om}| = (1-\eps) \rho_{a,b} + \eps (I_d\otimes \Om_d)|\om_d\ra\la \om_d| (I_d\otimes \Om_d)^*$$
is PPT with $\SN(\rho_{AB})\geq d/2$.
\end{corollary}
\begin{proof}
We first note that $\la \om_d|\om_d^{\Om}\ra = 0$ and $\la \om_d^{\Om}|F_d^{\Om}|\om_d^{\Om}\ra = \la \om_d|F_d|\om_d\ra = 1$, and hence
    $$\Tr(C_{\Le_k^{\rm BH}} |\om_d^{\Om}\ra\la \om_d^{\Om}|) = \frac{1}{d(kd-k-1)}(k-0-k)=0, \quad k=1,\ldots, d/2-1.$$
Therefore, the same arguement with the proof of Corollary \ref{cor-PPTSchOpt1} shows that $\SN(\rho_{AB})\geq d/2$. On the other hand, as in the proof of Proposition \ref{prop-Symp-CP}, we have
    $$\rho_{a,b}^{\Gamma} = \Big(\frac{1-a-b}{d^2}+b-\frac{a}{d}\Big)\Pi_1^{SS}+\Big(\frac{1-a-b}{d^2}+\frac{a}{d}\Big)\Pi_2^{SS} + \Big(\frac{1-a-b}{d^2} - \frac{a}{d}\Big)\Pi_3^{SS},$$
which is strictly positive whenever $(a,b)\in {\rm int}\,\mathbb{T}$. Therefore, $\rho_{AB}^{\Gamma}\geq 0$ for sufficiently small $\eps>0$.
\end{proof}

%----PPTsquared------------

\section{The PPT squared conjecture under symplectic group symmetry} \label{sec-PPTSq}

Initially inspired by the theory of quantum repeaters \cite{BCHW15, CF17}, the \emph{PPT squared conjecture} has since emerged as one of the central open problems in quantum information theory, with recent connections to operator theory and operator algebras. We briefly recall its formulation below.

\begin{conjecture} [\cite{PPTsq, CMHW19}] \label{conj-PPT1}The composition of any two \emph{PPT} linear maps is \emph{entanglement breaking}. In other words, if $\Le_1$ and $\Le_2$ are CP linear maps whose Choi matrices are PPT, then the Choi matrix of $\Le_1\circ \Le_2$ is separable.
\end{conjecture}

The conjecture has resolved only in low dimensions $d=2$ and $d=3$ \cite{CMHW19, Chen2019}. A plethora of works have established the conjecture in various restricted settings or proved weaker variants, e.g., showing that repeated self-composition of a PPT map eventually yields an entanglement-breaking channel \cite{Lami2015entanglebreak,Kennedy2017,Rahaman2018,CYZ18,hanson2020eventually}. Notably, the conjecture is known to be equivalent to the following formulation in terms of positive maps \cite{CMHW19,GKS21}:

\begin{conjecture} \label{conj-PPT2}
For any positive linear map $\Le_1$ and any PPT linear map $\Le_2$, the composition $\Le_1\circ \Le_2$ is decomposable.
\end{conjecture}

On the other hand, several case studies have examined the conjecture within specific families of quantum channels exhibiting group symmetries \cite{CMHW19,singh2022ppt,NP25}, all of which provide further supporting evidence for the PPT squared conjecture. In this section, we contribute additional evidence by verifying both Conjectures \ref{conj-PPT1} and \ref{conj-PPT2} for the class of linear maps that are covariant with respect to the symplectic group. Our result is particularly noteworthy, as it applies to settings where the existence of highly positive indecomposable maps and PPT states with a large degree of entanglement are guaranteed by Corollaries \ref{cor-Symp-PPTSch} and \ref{cor-Symp-PPTSch}.

\begin{theorem} \label{thm-PPTsquared}
Let $d\geq 4$ be an even integer and let  $\Le_1,\Le_2:M_d(\Comp)\to M_d(\Comp)$ be two Hermitian-preserving linear maps.
\begin{enumerate}
    \item If both $\Le_1$ and $\Le_2$ are PPT maps that are either $(S,S)$- or $(S,\overline{S})$-covariant, then the composition $\Le_1\circ \Le_2$ is entanglement breaking.

    \item If $\Le_1$ is a positive map that is either $(S,S)$- or $(S,\overline{S})$-covariant, and $\Le_2$ is a PPT map that is either $(S,S)$- or $(S,\overline{S})$-covariant, then both compositions $\Le_1\circ \Le_2$ and $\Le_1\circ \Le_2$ are decomposable.
\end{enumerate}
\end{theorem}
\begin{proof}
For both two statements, we may assume that $\Le_1$ and $\Le_2$ are trace-preserving. Furthermore, for any $(S,S)$-covariant linear map $\Le_{a,b}$ defined in Eq. \eqref{eq-SympCov}, it is simple to check $\top \circ \Le_{a,b} = \Le_{a,b}\circ \top\in \Cov(S,\overline{S})$. Therefore, any composition between TP linear maps which are either $(S,S)$-covariant or $(S,\overline{S})$-covariant is of the form
    $$\Le_{a,b}\circ \Le_{p,q} \text{ \quad or \quad} \top\circ \Le_{a,b}\circ \Le_{p,q}$$
for some $(a,b),(p,q)\in \Real^2$. Since the positivity, PPT property, entanglement breaking property, and decomposability of lienar maps are all invariant under the composition of $\top$, we may consider only the case $\Le_1=\Le_{a,b}$ and $\Le_2=\Le_{p,q}$. In this case, it is also simple to further verify
    $$\Le_{a,b}\circ \Le_{p,q}=\Le_{ap+bq, aq+bp}$$
from the reltations $\Delta\circ \Delta = \Delta \circ (\Ad_{\Om}\circ \top) = \Delta$ and $(\Ad_{\Om}\circ \top)\circ (\Ad_{\Om} \circ \top)=\id$.

Now suppose both $\Le_1$ and $\Le_2$ are PPT, that is, $(a,b),(p,q)\in \mathbb{T}$, and let us show $(ap+bq,aq+bp)\in \mathbb{S}_1$. Since the map $\Le_1\circ \Le_2$ is already guaranteed to be PPT, it suffices to show that
    $$(ap+bq)+(aq+bp)=(a+b)(p+q)\leq \frac{1}{d+1},$$
by comparing the conditions in Theorem \ref{thm-SympSepPPT}. However, since the condition $(a,b),(p,q)\in \mathbb{T}$ implies $|a+b|, |p+q|\leq \frac{1}{d+2}+\frac{1}{d+2}=\frac{2}{d+2}$, we actually have $(a+b)(p+q)\leq \frac{4}{(d+2)^2}<\frac{1}{d+1}$. This proves the first assertion.

For the second assertion, suppose $(a,b)\in \mathbb{P}_1$ and $(p,q)\in \mathbb{T}$. Similarly, it suffices to check $(ap+bq)+(aq+bp)=(a+b)(p+q)\geq -\frac{2+d}{d^2-d-2}$
from Theorem \ref{thm-Symp-PosDec}, under the additional assumptions that $(a,b)\in \mathbb{P}_1\setminus \mathbb{D}$ and $(p,q)\in \mathbb{T}\setminus \mathbb{S}_1$. In this case, the conditions imply that 
\begin{align*}
    -\frac{d+2}{d^2-d-2}>\, &a+b \geq -\frac{1}{d-2}-\frac{1}{d-2}=-\frac{2}{d-2};\\
    \quad \frac{1}{d+1}<\, &p+q\leq \frac{1}{d+2}+\frac{1}{d+2} =\frac{2}{d+2},
\end{align*}
and hence we actually have $(a+b)(p+q)\geq -\frac{2}{d-2} \cdot \frac{2}{d+2}= -\frac{4}{d^2-4}> -\frac{d+2}{d^2-d-2}$, which completes the proof.
\end{proof}

It is natural to expect that Conjectures \ref{conj-PPT1} and \ref{conj-PPT2} remain equivalent when restricted to classes of maps defined by general group symmetries. However, we were unable to find a complete argument establishing this equivalence in full generality. 

\begin{question}
For any two unitary representations $\pi_A$ and $\pi_B$ of a compact group $G$ and for two covariant maps $\Le_1\in \Cov(\pi_A,\pi_B)$ and $\Le_2\in \Cov(\pi_B,\pi_A)$, are the following equivalent?
\begin{enumerate}
    \item $\Le_1\circ\Le_2$ is entanglement-breaking whenever $\Le_1$ and $\Le_2$ are PPT;

    \item $\Le_1\circ \Le_2$ is decomposable whenever $\Le_1$ is positive and $\Le_2$ is PPT.
\end{enumerate}
\end{question}

% -Application------------

\section{Optimal lower bounds for the Sindici–Piani semidefinite program} \label{sec-SP-SDP}

In order to generate PPT entanglement, Sindici and Piani \cite{SP18} proposed the following semidefinite program (SDP). Let $\Pi_{\A}=\frac{1}{2}(I_{d^2}+F_d)$ denote the projection onto the antisymmetric subspace of $\Comp^d\otimes \Comp^d$. Given an \textit{antisymmetric state} $\rho_{\A}$, i.e., a state satisfying $\Pi_{\A}\rho_{\A}\Pi_{\A}=\rho_{\A}$, define
\begin{center}
    $p^{\rm PPT}(\rho_\A):=\max_{\sigma} 
    \Tr(\Pi_\A\sigma)$ \quad subject to \quad $\begin{cases}
        \text{$\sigma$ is a PPT state}, \\ \text{$\Pi_\A\sigma \Pi_\A=\Tr(\Pi_\A\sigma) \rho_\A$.}
    \end{cases}$
\end{center}
It was shown that $p^{\rm PPT}$ admits a universal upper bound $p^{\rm PPT}(\rho_\A)\leq \frac{1}{2}$, and that equality is attained whenever an optimal solution $\sigma^*$ can be chosen to be separable. In particular, if $p^{\rm PPT}(\rho_{\A})<\frac{1}{2}$, then any optimal state $\sigma^*$ is necessarily PPT entangled.

The quantity $p^{\rm PPT}$ was further studied by Pál and Vértesi \cite{PV19} in their construction of PPT states with high Schmidt number. To this end, for $d\geq 2$, we consider the optimal lower bound
    $$p_{\min}(d):=\min_{\rho_\A} p^{\rm PPT}(\rho_{\A})$$
where the minimization is taken over all $d\otimes d$ antisymmetric states $\rho_{\A}$. By definition, the sequence satisfies $p_{\min}(d)\leq p_{\min}(d-1)\leq \cdots \leq p_{\min}(2)\leq 1/2$. For all antisymmetric state $\rho_{\A}$, it is shown that $p^{\rm PPT}(\rho_{\A})\geq p_{\min}(\SN(\rho_{\A}))$ and the Schmidt number $\SN(\rho_{\A})$ is always an even integer. In other words, if $p^{\rm PPT}(\rho_{\A})<p_{\min}(2r)$ for some $r\geq 1$, then it necessarily implies that $\SN(\rho_{\A})\geq 2r+2$. Therefore, using the Schmidt number inequality $$\SN(\Pi_{\A} \sigma\Pi_{\A})\leq 2\SN(\sigma),$$ obtained from \cite{PV19,Car20}, we conclude that $\SN(\sigma)\geq \frac{1}{2}\SN(\rho_{\A})\geq r+1$, yielding a lower bound on the Schmidt number of any optimal PPT state $\sigma$ achieving $p^{\rm PPT}(\rho_{\A})$.

In \cite{PV19} it was further shown that $p_{\min}(2)=p_{\min}(3)=1/2$ and $p_{\min}(2r) = p_{\min}(2r+1)$ for all integers $r\geq 1$. Moreover, the minimum $p_{\min}(d)$ is attained by the antisymmetric pure state
    $$\rho_{\A}= |\om_d^{\Om}\ra\la \om_d^{\Om}|; \quad |\om_d^{\Om}\ra = (I_d\otimes \Om_d)|\om_d\ra = \frac{1}{\sqrt{d}} \sum_{i=1}^{d/2} (|i+d/2,i\ra - |i,i+d/2\ra),$$
from which it was estimated that $p_{\min}(2r)\ge \tfrac{1}{2r+2}$ for all $r\ge 2$. In the following, we exploit symplectic group symmetry to show that equality indeed holds, thereby resolving the conjecture posed in \cite{PV19}.

\begin{theorem}
For every even integer $d\geq 4$, we have
    $$p_{\min}(d)=p^{\rm PPT}(|\om_d^{\Om}\ra\la \om_d^{\Om}|) = \frac{1}{d+2}.$$
\end{theorem}
\begin{proof}
Let $\sigma$ be a PPT state such that $\Pi_{\A} \sigma \Pi_{\A} = \Tr(\Pi_{\A} \sigma) |\om_d^{\Om} \ra\la \om_d^{\Om}|$. We first show that its $S\otimes S$-twirling $\tilde{\sigma}=\T_{S\otimes S}(\sigma)=\E_{S\in Sp(d)}\big[ (S\otimes S)\sigma (S\otimes S)^* \big]$ satisfies the same relation
\begin{equation} \label{eq-SStwirling-Opt}
    \Pi_{\A} \tilde{\sigma} \Pi_{\A} = \Tr(\Pi_{\A}\tilde{\sigma})|\om_d^{\Om}\ra\la \om_d^{\Om}| = \Tr(\Pi_{\A}{\sigma})|\om_d^{\Om}\ra\la \om_d^{\Om}|.
\end{equation}
Indeed, by recalling that $\Pi_1^{SS} = |\om_d^{\Om}\ra\la \om_d^{\Om}|$, $\Pi_2^{SS}=\Pi_{\mathcal{S}}$, and $\Pi_3^{SS}= \Pi_{\A}-|\om_d^{\Om}\ra\la \om_d^{\Om}|$ are three mutually orthogonal projections in $\Inv(S\otimes S)$ (Corollary \ref{cor-Symp-proj}), Proposition \ref{prop-twirlformula} implies that
\begin{align*}
    \tilde{\sigma}&= \sum_{j=1}^3 \frac{1}{{\rm rank}(\Pi_j^{SS})} \Tr( \Pi_j^{SS} \sigma) \Pi_j^{SS}\\
    &=\la \om_d^{\Om}|\sigma|\om_d^{\Om}\ra \, |\om_d^{\Om}\ra\la \om_d^{\Om}| + \frac{2}{d^2+d}\Tr(\Pi_{\mathcal{S}} \sigma) \Pi_{\mathcal{S}}+\frac{2}{d^2-d-2}\Tr((\Pi_{\mathcal{A}}-|\om_d^{\Om}\ra\la \om_d^{\Om}|)\sigma)\big(\Pi_{\mathcal{A}}-|\om_d^{\Om}\ra\la \om_d^{\Om}|\big) \\
    &= \Tr(\Pi_{\A}\sigma) |\om_d^{\Om}\ra\la \om_d^{\Om}| + \frac{2}{d^2+d}(1-\Tr(\Pi_{\A}\sigma))\Pi_{\mathcal{S}},
\end{align*}
where the second and third equalities follow from $|\om_d^{\Om}\ra\in {\rm Ran}(\Pi_{\A})$ and
    $$\la \om_d^{\Om}|\sigma|\om_d^{\Om}\ra = \la \om_d^{\Om}|\Pi_{\A}\sigma \Pi_{\A}|\om_d^{\Om}\ra = \Tr(\Pi_{\A} \sigma),$$
by the assumption of $\sigma$. This shows the relation \eqref{eq-SStwirling-Opt}. In particular, this implies that the optimal state $\sigma^*$ for $p^{\rm PPT}(|\om_d^{\Om}\ra\la \om_d^{\Om}|)$ can be chosen to be of the form
    $$\sigma^*=p|\om_d^{\Om}\ra\la \om_d^{\Om}|+\frac{2(1-p)}{d^2+d}\Pi_{\mathcal{S}} = \rho_{\frac{1-p}{d+1},p}^{\Gamma},$$
where in this case $p=\Tr(\Pi_{\A}\sigma)$. By Theorem \ref{thm-SympSepPPT}, $\sigma^*$ above is PPT if and only if $0\leq p\leq \frac{1}{d+2}$.
\end{proof}

\section{Conclusion}

By exploiting symplectic group symmetries, we construct a broad class of new $k$-positive indecomposable linear maps and PPT states exhibiting a high degree of entanglement, attaining the best known general bounds from \cite{PV19,Car20}. In particular, we provide a systematic construction of nontrivial quantum objects, including $k$-positive maps with strong indecomposability properties and PPT entangled states displaying an extremal Schmidt number gap between a state and its partial transpose, a phenomenon that had remained largely beyond reach prior to this work. Furthermore, we demonstrate several applications of our framework, including results related to the PPT-squared conjecture and the optimal bound in the SDP associated with PPT entanglement.

Our results demonstrate the effectiveness of the symmetry method as a systematic tool for analyzing nontrivial phenomena related to Schmidt numbers and positivity. By exploiting symplectic group symmetries and tools from projective geometry, we obtain analytically tractable families in which representation-theoretic structures translate directly into geometric descriptions of entanglement, revealing an unexpected interplay between \textit{representation theory, projective geometry, and quantum information.} Remarkably, even a simple two-parameter family already yields a broad class of highly entangled PPT states and highly positive yet indecomposable maps, suggesting that symplectic symmetry captures essential features of extremal entanglement behavior. This observation naturally motivates further exploration beyond the unitary symplectic group, for instance toward weaker symmetry assumptions such as \textit{orthogonal symplectic symmetry} or more general covariance classes. We expect that such extensions will uncover new entanglement structures and provide a unified framework for constructing and classifying quantum states, potentially improving the current best bounds and further leading to counterexamples to several important open problems, such as the PPT-squared conjecture \cite{PPTsq,CMHW19} and the NPT distillation problem \cite{HRZ22}. We leave these directions for future work.

\bigskip

\emph{Acknowledgements}: The author thanks Aabhas Gulati, Seung-Hyeok Kye, Robin Krebs, and Mariami Gachechiladze for the helpful discussions and comments. We would also like to thank the organizers of the first meeting of \href{http://www.fields.utoronto.ca/activities/24-25/operator-algebras-quantum-information}{Operator Algebras and Quantum Information} for the kind invitation to present a preliminary version of this work. S.-J.P. was supported by the ANR project \href{https://esquisses.math.cnrs.fr/}{ESQuisses}, grant number ANR-20-CE47-0014-01.

%-----------------------------------

\appendix

\section{Proof of \cref{thm-SympkPos}} \label{app-kPos}

This section is devoted to provide the complete proof of Theorem \ref{thm-SympkPos}. We first recall a $k$-positivity criterion presented in \cite[Lemma 1]{Tom85}.

\begin{proposition} \label{prop-kpos}
Let $1\leq k\leq d$. Then a linear map $\Le:M_d(\Comp)\to M_d(\Comp)$ is $k$-positive if and only if the bipartite matrix
   \begin{equation}\label{eq40}
       C_k^v(\Le):=\sum_{i,j=1}^k |i\ra\la j|\otimes \Le(|v_i\ra\la v_j|)\in M_k(\Comp)\otimes M_d(\Comp)
   \end{equation}
is positive semidefinite for any choice of an orthonormal subset $\{v_1,\ldots, v_k\}$ of $\Comp^d$.
\end{proposition}

Note that $C_k^v(\Le)$ reduces to the (unnormalized) Choi matrix of $\Le$ when $k=d$ and we choose the standard basis $|v_i\ra = |i\ra$ for $i=1,\ldots, d$. However, this criterion is still difficult to apply in general, since one must verify the positivity of $C_k^v(\Le)$ for \emph{arbitrary} choices of $\{v_1,\ldots, v_k\}$. Nevertheless, we show that it can be applied surprisingly effectively in our setting to obtain a precise description of the 
$k$-positivity regions $\mathbb{P}_k^{(d)}$ of $\Le_{p,q}$.

The following optimization results play a crucial role in applying Proposition \ref{prop-kpos} for the proof of \cref{thm-SympkPos}.

\begin{lemma} \label{lem-Symp-optimization}
For $1\leq k\leq d$, we have
\begin{align}
    \max\sum_{j,j'=1}^k|\la {v_j}|\Om_d|\overline{v_{j'}}\ra|^2 &= 2\lfloor k/2\rfloor, \label{eq-MaxOptimize}\\
    \min \sum_{j,j'=1}^k|\la {v_j}|\Om_d|\overline{v_{j'}}\ra|^2 &= \max(2k-d,0), \label{eq-MinOptimize}
\end{align}
where the maximum and minimum are taken over all orthonormal subsets $\{v_1,\ldots, v_k\}\subset \Comp^d$.
\end{lemma}
\begin{proof}
Let us first take an ordered set $\B:=\{v_1,\ldots, v_k, \Om_d\overline{v}_1,\ldots, \Om_d\overline{v}_k\}$ consisting of $2k$ vectors and consider the \textit{Gram matrix} $G$ associated to $\B$, that is,
    $$G=\big(\la a|b\ra\big)_{a,b\in \B} = \begin{pmatrix}
        I_k & B \\ B^* & I_k
    \end{pmatrix} \in M_{2k}(\Comp)^+,$$
where $B:=(\la v_i|\Om_d\overline{v}_j\ra)_{1\leq i,j\leq k}\in M_k(\Comp)$. Note that the diagonal blocks of $G$ are identity since both the sets $\{v_i\}_{i=1}^k$ and $\{\Om\overline{v}_i\}_{i=1}^k$ are orthonormal. We have two additional observations:
\begin{enumerate}
    \item $\|B\|_{\infty}\leq 1$, since $G\geq 0$ (see e.g. \cite[Lemma 3.18]{TextWat18}).
    
    \item $B^{\top}=-B$, i.e., $B$ is skew-symmetric, from the fact $B_{ij}=\la \overline{v_i},\overline{v_j}\ra_{\Om}$ and by Eq. \eqref{eq-skew-form}.
\end{enumerate}
Therefore, the classical result of skew-symmetric matrix (see \cite[Corollary 4.4.19]{TextHJ12} implies that there exists a unitary matrix $U\in U(k)$ and nonnegative numbers $s_1,\ldots, s_r\geq 0$ ($2r\leq k$) such that
    $$U^{\top} B U = 0_{k-2r}\oplus \begin{pmatrix} 0 & s_1 \\ -s_1 & 0    \end{pmatrix} \oplus \cdots \oplus \begin{pmatrix} 0 & s_r \\ -s_r & 0    \end{pmatrix}.$$
In fact, we have $s_1,\ldots, s_r\leq 1$ since $\|B\|_{\infty}\leq 1$. Therefore, one has
    $$\sum_{j,j'=1}^k|\la {v_j}|\Om_d|\overline{v_{j'}}\ra|^2 = \|B\|_2^2 =2(s_1^2+\cdots +s_r^2)\leq 2r\leq  2\lfloor k/2\rfloor.$$
Now if $k$ is even, then the choice of orthonormal vectors $|v_{2j-1}\ra=|j\ra$ and $|v_{2j}\ra =|j+d/2\ra$ ($j=1,\ldots, \frac{k}{2}$) attains this maximum $\|B\|_2^2=k$. If $k$ is odd, a similar construction with $|v_{2j-1}\ra=|j\ra$ ($j=1,\ldots, \frac{k+1}{2}$) and $|v_{2j}\ra =|j+d/2\ra$ ($j=1,\ldots, \frac{k-1}{2}$) yields the maximum $\|B\|_2^2 = k-1$. This completes the proof of the maximization in Eq. \eqref{eq-MaxOptimize}.

We now turn to the minimization in \eqref{eq-MinOptimize}. If $2k\leq d$, then we may simply take $|v_j\ra=|j\ra$ ($j=1,\ldots, k$) and check $\la {v_j}|\Om_d|\overline{v_{j'}}\ra=0$ for every $j,j'$, so the equality \eqref{eq-MinOptimize} holds. From now, let us assume $2k>d$. If we choose $|v_j\ra= |j\ra$ ($j=1,\ldots, k$) as before, then we have
\begin{align*}
    \sum_{j,j'}|\la v_j|\Om_d|\overline{v_{j'}}\ra|^2=2(k-d/2)=2k-d,
\end{align*}
from the observations that $\Om_d|\overline{v_{j+\frac{d}{2}}}\ra=|v_j\ra$ for $j=1,\ldots, k-\frac{d}{2}$ and that $|v_j\ra$ is orthogonal to all $\Om_d|\overline{v_{j'}}\ra$ unless $j'=j\pm\frac{d}{2}$. On the other hand, for any orthonormal vectors $v_1,\ldots, v_k$, we associate two orthogonal projections $\Pi_v=\sum_{j=1}^k |v_j\ra\la v_j|$ and $\Pi_{\overline{v}}=\sum_{j=1}^k |\overline{v_j}\ra \la \overline{v_j}|$. Then we can observe that
\begin{align*}
    \sum_{j,j'} |\la {v_j}|\Om_d|\overline{v_{j'}}\ra|^2=\text{Tr}(\Pi_v \Om_d \Pi_{\overline{v}}\Om_d^{*})&\geq \text{dim}\big(\text{Ran}(\Pi_{v})\cap \text{Ran}(\Om_d \Pi_{\overline{v}}\Om_d^{*})\big),
    \end{align*}
and the right hand side is further equal to
    \begin{align*}
    \text{dim}\big(\text{Ran}(\Pi_{v})\big)+\text{dim}\big(\text{Ran}(\Om_d \Pi_{\overline{v}} \Om_d^{*})\big)-\text{dim}\big(\text{Ran}(\Pi_{v})+\text{Ran}(\Om_d \Pi_{\overline{v}}\Om_d^{*})\big).
\end{align*}
Since both two projections $\Pi_v$ and $\Om_d\Pi_{\overline{v}}\Om_d^*$ have rank $k$, one has $\sum_{j,j'} |\la {v_j}|\Om_d|\overline{v_{j'}}\ra|^2 \geq 2k-d$, and hence establish the advertised minimum \cref{eq-MinOptimize}.
\end{proof}

Next, we collect several properties of the polynomial $f_{u}^{(d)}(x,y)$ defined in \cref{eq-conic}.

\begin{lemma} \label{lem-Conic}
Let $f_u^{(d,k)}(x,y):= \big(1-x-(1+d)y \big) \big(1-(1-kd)x-y \big)+d^2uxy$ for $t\geq 0$.
\begin{enumerate}
    \item For every $u\geq 0$, the conic section $f_{u}^{(d,k)}(x,y)=0$ represents a (possibly degenerate) hyperbola passing through the four points
        $$(1,0), \quad \Big (-\frac{1}{kd-1}, 0 \Big),\quad (0,1), \quad \Big(0,\frac{1}{d+1} \Big).$$
    Furthermore the curve degenerates (into two lines) when $t=0$ or $u=k$.

    \item $f_{k-1}^{(d,k)}(x,y)=0$ is tangent to the line $(1-(k-1)d)x+(1+d)y=1$ at $(0,\frac{1}{d+1})$.

    \item $f_{2k-d}^{(d)}(x,y)=0$ passes through an additional point $(-\frac{2}{d^2-d-2}, -\frac{d}{d^2-d-2})$ and is tangent to the line $(1-d(d-k))x+y=1$ at $(0,1)$.
\end{enumerate}

\end{lemma}
\begin{proof} 
(1) It is straightforward to verify that the four points above lie on the curve $f_u^{(d,k)}(x,y)=0$. Since the point $(0,\frac{1}{d+1})$ lies in the interior of the triangle determined by the other three points $(1,0),(-\frac{1}{kd-1}, 0)$, and $(0,1)$, it follows that the conic $f_u^{(d,k)}(x,y)=0$ is a hyperbola.

(2) Assume that a point $(p,q)\in \Real^2$ is on the line $(1-(k-1)d)x+(1+d)y=1$. Then $1-p-(1+d)q=-d(k-1)p$ and $1-(1-kd)p-q=d(p+q)$, and therefore,
    $$f_{k-1}^{(d,k)}(p,q)=-d^2(k-1)p(p+q)+d^2(k-1)pq=-d^2(k-1)p^2\leq 0,$$
and moreover, $f_{k-1}^{(d,k)}(p,q)=0$ if and only if $(p,q)=(0,\frac{1}{d+1})$.

The proof of (3) is analogous.
\end{proof}

\begin{proof}[\textbf{Proof of Theorem \ref{thm-SympkPos}}]
For an orthonormal subset $\{v_1,\ldots, v_k\}$ in $\Comp^d$, the associated bipartite matrix in \eqref{eq40} is given by
    $$C_k^v(\Le_{p,q})=\frac{1-p-q}{d}I_k\otimes I_d+ kp|\om_k^v\ra\la \om_k^v|+qF_k^{v,\Om},$$
where $\begin{cases} |\om_k^v\ra=\frac{1}{\sqrt{k}}\sum_{j=1}^k|j\ra\otimes |v_j\ra\in \Comp^k\otimes \Comp^d,\\
F_k^{v,\Om}=\sum_{i,j=1}^k|i\ra\la j|\otimes |\Om_d\overline{v_j}\ra\la \Om_d \overline{v_i}|\in \M{k}\otimes \M{d}.\end{cases}$ 
Moreover, we can write $F_k^{v,\Om}=\Pi_{\mathcal{S}}^{v,\Om}-\Pi_{\mathcal{A}}^{v,\Om}$, where 
$\Pi_{\mathcal{S}}^{v,\Om}$ and $\Pi_{\mathcal{A}}^{v,\Om}$  are orthogonal projections onto the (symmetric-like and antisymmetric-like) spaces 
\begin{center}
    ${\rm span}\{\frac{|i\ra| \Om_d 
\overline{v_j}\ra+|j\ra|\Om_d \overline{v_i}\ra}{\sqrt{2}}: 1\leq i\leq j\leq k\}$ 
and ${\rm span}\{\frac{|i\ra|\Om_d \overline{v_j}\ra-|j\ra|\Om_d \overline{v_i}\ra}{\sqrt{2}}: 1\leq i< j\leq k\}$,
\end{center}
respectively. Note that ${\rm Ran}(\Pi_{\mathcal{S}}^v)\perp {\rm Ran}(\Pi_{\mathcal{A}}^v)$, and $|\om_k^v\ra\perp {\rm Ran}(\Pi_{\mathcal{S}}^{v,\Om})$ since
    $$\left \la \om_k^v\Bigg|\frac{e_i\otimes \Om_d \overline{v_j} + e_j\otimes \Om_d \overline{v_i}}{\sqrt{2}}\right \ra=\frac{1}{\sqrt{2k}}(\la v_i|\Om_d| \overline{v_j}\ra + \la v_j| \Om_d |\overline{v_i}\ra)=0$$
for all $1\leq i\leq j\leq k$ (observation (2) in the proof of Lemma \ref{lem-Symp-optimization}). Therefore, after rewriting $C_k^v(\Le_{p,q})$ into the orthogonal decomposition
    $$C_k^v(\Le_{p,q})=\left(A(I_{kd}-\Pi_{\mathcal{S}}^{v,\Om})+ kp|\om_k^v\ra\la \om_k^v|-q\Pi_{\mathcal{A}}^{v,\Om}\right)\oplus \left(A+q\right)\Pi_{\mathcal{S}}^v$$
with $A:=\frac{1-p-q}{d}$, the desired condition $C_k^v(\Le_{p,q})\geq 0$ is equivalent to $A+q\geq 0$ and 
\begin{equation}\label{eq-pos}
    A(I_{kd}-\Pi_{\mathcal{S}}^{v,\Om})+ kp |\om_k^v\ra\la \om_k^v| - q\Pi_{\mathcal{A}}^{v,\Om}\geq 0.
\end{equation}
A difficulty with the condition \eqref{eq-pos} is that $|\om_k^v\ra\la \om_k^v|$ and $\Pi_{\mathcal{A}}^{v,\Om}$ do not commute in general and are therefore not simultaneously diagonalizable, unlike the case $|\om_d\ra\la \om_d|$ and $\Pi_{\A}^{\Om}=\Ad_{I_d\otimes \Om_d}(\Pi_\A)$ considered in Corollary \ref{cor-Symp-proj}. To overcome this, let us define $|\xi_1\ra=\Pi_{\mathcal{A}}^{v,\Om}|\om_k^v\ra\in {\rm Ran}(\Pi_{\mathcal{A}}^{v,\Om})$ and  $|\xi_2\ra=|\om_k^v\ra - |\xi_1\ra \perp |\xi_1\ra$. Then 
\[\xi_2\perp ({\rm Ran}(\Pi_{\mathcal{S}}^{v,\Om})\cup {\rm Ran}(\Pi_{\mathcal{A}}^{v,\Om})).\]
Therefore, after identifying $|\om_k^v\ra\la \om_k^v|=\begin{pmatrix}
    |\xi_1\ra \\ |\xi_2\ra
\end{pmatrix} \big(\la \xi_1|\; \la \xi_2| \big) = \begin{pmatrix}
    |\xi_1\ra\la \xi_1| & |\xi_1\ra\la \xi_2| \\ |\xi_2\ra\la \xi_1| & |\xi_2\ra\la \xi_2|
\end{pmatrix}$, one has the following block matrix decomposition
\begin{align} \label{eq-OOblock}
    &A(I_{kd}-\Pi_{\mathcal{S}}^{v,\Om})+ kp |\om_k^v\ra\la \om_k^v| - q\Pi_{\mathcal{A}}^{v,\Om} \nonumber \\
    &\cong\begin{pmatrix} (A-q)\Pi_{\mathcal{A}}^{v,\Om} + kp|\xi_1\ra\la \xi_1| & kp|\xi_1\ra\la \xi_2| \\ kp|\xi_2\ra\la \xi_1| & A(I-\Pi_{\mathcal{S}}^{v,\Om}-\Pi_{\mathcal{A}}^{v,\Om})+ kp|\xi_2\ra\la \xi_2|\end{pmatrix}.
\end{align}
We remark here that  ${\rm rank}(\Pi_{\mathcal{A}}^{v,\Om})=\frac{k(k-1)}{2}\geq 2$ unless $k\leq 2$ and ${\rm rank}(I-\Pi_{\mathcal{S}}^v-\Pi_{\mathcal{A}}^v)=d^2-k^2\geq 2$ whenever $k<d$. Therefore, the block matrix in \eqref{eq-OOblock} is positive semidefinite if and only if
\begin{equation} \label{eq-OOblock2}
\begin{cases}
    (a)& A-q\geq 0 \text{ (if $k>2$)},\\
    (b)& A\geq 0 \text{ (if $k<d$)},\\
    (c)& A-q+kp\|\xi_1\|^2\geq 0,\\
    (d)& A+kp\|\xi_2\|^2\geq 0,\\
    (e)& \left(A-q+kp\|\xi_1\|^2\right)\left(A+kp\|\xi_2\|^2\right)\geq (kp)^2\|\xi_1\|^2\|\xi_2\|^2.
\end{cases}
\end{equation}
Since $\|\xi_1\|^2+\|\xi_2\|^2=\|\om_k^v\|^2=1$, the conditions $(d)$ and $(e)$ are respectively equivalent to
\begin{equation}
    \begin{cases}
        (d')& A+kp-kp\|\xi_1\|^2\geq 0,\\
        (e')& (A-q)(A+kp)+ kpq\|\xi_1\|^2=\frac{1}{d^2}f_{k\|\xi_1\|^2}^{(d,k)}(p,q)\geq 0,
    \end{cases}
\end{equation}
where the equality in $(e')$ is a simple comparison with Eq. \eqref{eq-conic}. Note that the first two conditions (a) and (b) are independent of the choices of an orthonormal subset $\{v_i\}_{i=1}^k\subset \Comp^d$, and that  the other inequalities in $(c)$, $(d')$ and $(e')$ are linear in $\|\xi_1\|^2$. Since the inequalities in $(c)$, $(d')$ and $(e')$ should hold for all possible choices of $\{v_i\}_{i=1}^k\subseteq \Comp^d$, it suffices to consider the conditions for the maximum and minimum values of $\|\xi_1\|^2$. We can further compute that
\begin{align*}
    \|\xi_1\|^2&=\la \om_k^v|\Pi_{\mathcal{A}}^{v,\Om}|\om_k^v\ra\\
    &=-\la \om_k^v|F_k^{v,\Om}|\om_k^v\ra \quad (\,\because |\om_k^v\ra\perp {\rm Ran}(\Pi_{\mathcal{S}}^{v,\Om}))\\
    &=\frac{1}{k}\sum_{j,j'}\la v_j|\Om_d |\overline{v}_{j'}\ra \cdot \la \overline{v}_{j'}| \Om_d^{\top} |v_j\ra=\frac{1}{k}\sum_{j,j'}|\la v_j| \Om_d| \overline{v}_{j'}\ra|^2,
\end{align*}
and therefore, we can apply  Lemma \ref{lem-Symp-optimization} to conclude that
    $$\max \|\xi_1\|^2 = \frac{2}{k}\left\lfloor \frac{k}{2}\right\rfloor, \quad \min \|\xi_1\|^2 = \frac{1}{k}\max(2k-d,0).$$
To summarize, $\Le_{p,q}$ is $k$-positive ($1< k < d$) if and only if the six inequalities $(a)$, $(b)$, $(c)$, $(d')$, $(e')$, and $A+q\geq 0$ hold for $\|\xi_1\|^2\in \set{\frac{2}{k}\lfloor \frac{k}{2}\rfloor , \frac{\max(2k-d,0)}{k}}$ and $A=\frac{1-p-q}{d}$ (Here we note that the condition $(a)$ is the same with $(c)$ when $\|\xi_1\|^2=0$, so we may include $(a)$ even if $k=2$).

Suppose now that $2\leq k\leq \frac{d}{2}$. Then the conditions for $k\|\xi_1\|^2\in \{2\lfloor \frac{k}{2}\rfloor,0\}$ reduce to
\begin{equation} \label{eq-kPos1}
    A\pm q\geq 0, \quad A-q+ 2\Big\lfloor \frac{k}{2} \Big\rfloor p\geq 0, \quad A+ kp\geq 0, \quad (A-q)(A+ kp) + 2\Big\lfloor \frac{k}{2} \Big\rfloor pq \geq 0,
\end{equation}
where we omitted the conditions $A \geq 0$ (implied by $A \pm q \geq 0$) and $A + (k - 2\lfloor \frac{k}{2} \rfloor)p \geq 0$ (follows from $A \geq 0$ and $A + kp \geq 0$). Let us now further examine these conditions separately according to the parity of $k$:
\begin{itemize}
    \item If $k$ is even, then $2\lfloor \frac{k}{2}\rfloor = k$ gives 
        $$A\pm q\geq 0, \quad A-q+kp\geq 0, \quad A+kp \geq 0, \quad A(A-q+kp)\geq 0,$$
    where the last inequality may again be omitted (implied by the second one and $A\geq 0$). Therefore, we obtain the following minimal conditions:
        $$p+(1\pm d)q\leq 1, \quad (1-kd)p+(1+d)q \leq 1, \quad (1-kd)p+q\leq 1.$$

    \item If $k$ is odd, then $2\lfloor \frac{k}{2}\rfloor = k-1$ similarly gives 
    \begin{align*}
        p+(1\pm d)q\leq 1, \quad (1-(k-1)d)p+(1+d)q \leq 1, \quad (1-kd)p+q\leq 1, \quad f_{k-1}^{(d,k)}(p,q)\geq 0.
    \end{align*}
    However, the second inequality above is implied by the others, by Lemma \ref{lem-Conic} (2). Therefore, the above conditions are further reduced to
    \begin{align*}
        p+(1\pm d)q\leq 1, \quad (1-kd)p+q\leq 1, \quad f_{k-1}^{(d,k)}(p,q)\geq 0.
    \end{align*}
\end{itemize}
This completes the proof of (1) and (2) of Theorem~\ref{thm-SympkPos}. The remaining assertions (3) and (4) (i.e., the case $\frac{d}{2}<k<d$) can be proved analogously and are left to the reader.
\end{proof}

\section{Proof of \cref{thm-SympSch} (2), (3), and (4)} \label{app-SN-hard}

This section is devoted to provide the complete proof of Theorem \ref{thm-SympSch}. We first recall several notions and methods developed in \cite{PY24}. Let us first denote by 
\begin{equation}\label{eq43}
H_{p,q}=\left\{(x,y)\in \mathbb{R}^2: px+qy\leq 1 \right\}
\end{equation}
for all $(p,q)\in \mathbb{R}^2\setminus\left\{(0,0)\right\}$, and by $\alpha:\mathbb{R}^2\rightarrow \mathbb{R}^2$ a linear isomorphism given by 
\begin{equation}\label{eq-alpha}
    \alpha: \begin{pmatrix}
        x \\ y
    \end{pmatrix} \mapsto -(d+1) \begin{pmatrix}
        d-1 & -1 \\ -1 & d-1
    \end{pmatrix} \begin{pmatrix}
        x \\ y
    \end{pmatrix}.
\end{equation}
Then by Eq. \eqref{eq-Symp-paring}, we have the following identity:
\begin{equation} \label{eq-SNPolar}
    \mathbb{S}_k=\bigcap_{(p,q)\in {\rm ext}(P_k)}H_{\alpha(p,q)}=\alpha^{-1}\bigg(\bigcap_{(p,q)\in {\rm ext}(P_k)}H_{p,q}\bigg),
\end{equation} 
where the last equality is from the observation $H_{\alpha(p,q)} = \alpha^{-1}(H_{p,q})$. Recall that for $k$ in cases (2),(3), and (4) of Theorem \ref{thm-SympSch}, the set ${\rm ext}(\mathbb{P}_k)$ contains one or two segments of conic sections (which are always non-degenerate hyperbolas according to Lemma \ref{lem-Conic}). In this situation, it is often useful to exploit the notions of \emph{pole–polar duality} and \emph{dual curves} from projective geometry \cite{TextBK,Cox03} to describe the regions in \eqref{eq-SNPolar}.

The pole–polar duality provides a bijective correspondence between the space $\mathbb{R}^2\setminus \left\{(0,0)\right\}$ and the set of all lines in $\Real^2$ that are not passing through the origin $(0,0)$. To a point $P=(p,q)\in \mathbb{R}^2\setminus \left\{(0,0)\right\}$ we associate the line
$$l=\left\{(x,y)\in \mathbb{R}^2: px+qy=1\right\}.$$
We call $l$ the \emph{polar} of $P$ and $P$ the \emph{pole} of $l$ (with respect to the unit circle $x^2+y^2=1$), and denote by $l=:{\rm Polar}(P)$ and $P=:{\rm Pole}(l)$ respectively. 
Now for a smooth curve $\gamma:I\to \Real^2\setminus \{(0,0)\}$, we associate its \textit{dual curve} $\tilde{\gamma}:I\to \Real^2\setminus \{(0,0)\}$ defined by
    $$\tilde{\gamma}(t):={\rm Pole}(l_t), \quad t\in I,$$ 
where $l_t$ denotes the line tangent to $\gamma$ at $\gamma(t)$. More explicitly, if $\gamma(t)=(p(t),q(t))$ is a smooth parametrization, the  dual curve $\tilde{\gamma}$ is given by 
\begin{equation}\label{eq-dual-curve}
    \tilde{\gamma}(t)=\left(\frac{q'(t)}{p(t)q'(t)-q(t)p'(t)}, \frac{-p'(t)}{p(t)q'(t)-q(t)p'(t)}\right).
\end{equation}
The following lemma provides the connection between $\tilde{\gamma}$ and the intersection $\bigcap_{t\in I}H_{\gamma(t)}$.

\begin{lemma} [{\cite[Lemma 4.5]{PY24}}] \label{lem-curve}
Let $I$ be an open interval and $\gamma: I\to \Real^2\setminus \{(0,0)\}$ be a smooth, regular, and strictly convex curve. Suppose that for every $t\in I$, the tangent lines $l_t$ at $\gamma(t)$ do not pass through the origin $(0,0)$, and that the origin lies in the same (closed) half-plane as $\gamma$ with respect to $l_t$. Then the dual curve $\tilde{\gamma}$ from Eq. \eqref{eq-dual-curve} satisfies the following properties.
\begin{enumerate}
    \item ${\rm Polar}(\gamma(t))$ is tangent to $\tilde{\gamma}$ at $\tilde{\gamma}(t)$ for each $t\in I$. 

    \item For any closed interval $[t_0,t_1]\subset I$, the intersection
        $$\bigcap_{t\in [t_0,t_1]} H_{\gamma(t)}$$
    is the largest convex region containing $(0,0)$, which is bounded by two lines ${\rm Polar}(\gamma(t_0))$ and ${\rm Polar}(\gamma(t_1))$ as well as the image of the dual curve $\tilde{\gamma}|_{[t_0,t_1]}$ (as in Figure \ref{fig-SNPolar}).

    \item If the image of $\gamma$ is a segment of a conic, then so is $\tilde{\gamma}$.
\end{enumerate}
\end{lemma}
\begin{figure}[htb!]
    \centering
    \includegraphics[scale=0.45]{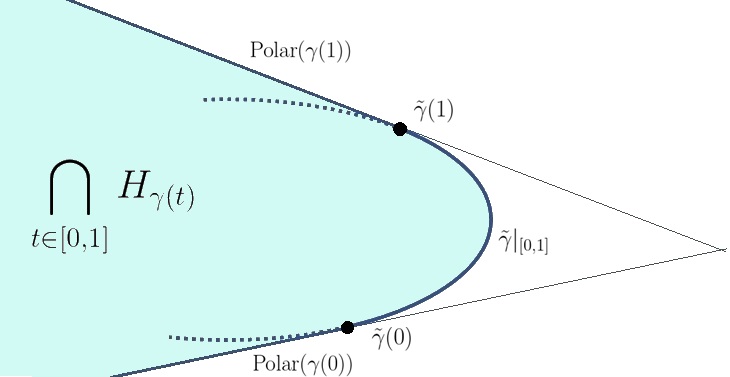}
    \caption{Geometric description of the intersection $\displaystyle \bigcap_{t\in [0,1]} H_{\gamma(t)}$}
    \label{fig-SNPolar}
\end{figure}

For the remainder of this section, and for each $0<u<k$, we consider a smooth curve $\gamma:(-\eps,1+\eps)\to \Real^2$ parametrizing one branch of the hyperbola $f_{u}^{(d,k)}(x,y)=0$ (that is, $f_{u}^{(d,k)}(\gamma(t))\equiv 0$) whose convex hull contains the origin. Then Lemma \ref{lem-curve} implies that the dual curve $\tilde{\gamma}$ represents a conic section and so is the curve $\alpha^{-1}\circ \tilde{\gamma}$ (recall that $\alpha$ defined in Eq. \eqref{eq-alpha} is a linear isomorphism). In other words, there exists a unique quadratic polynomial $\tilde{f}_u^{(d,k)}(x,y)$ satisfying $\tilde{f}_u^{(d,k)}(\alpha^{-1}(\tilde{\gamma}(t)))\equiv 0$. Consequently, a segment of the conic $\tilde{f}_u^{(d,k)}(x,y)=0$ determines a portion of the boundary of the set $\alpha^{-1}\big(\bigcap_{t\in [0,1]} H_{\gamma(t)}\big)$.

For completeness, we provide in detail how to determine the polynomial $\tilde{f}_u^{(d,k)}(x,y)$. First, set
    $$\begin{cases}A=1-kd, \quad B=-kd^2-kd+d+2+d^2t, \quad C=d+1, \\ D=kd-2, \quad E=-d-2, \quad  F=1,\end{cases}$$
for simplicity. For each $(p,q)$ on the conic $f_u^{(d,k)}(x,y)=Ax^2+Bxy+Cy^2+Dx+Ey+F=0$, the tangent line $l_{p,q}$ at $(p,q)$ is given by
\begin{equation} \label{eq-tangent}
    (2Ap+Bq+D)x+(2Cq+Bp+E)y+(Dp+Eq+2F)=0.
\end{equation}
Then by Eq. \eqref{eq-dual-curve}, the ``dual point '' ${\rm Pole}(l_{p,q})$ corresponding to $(p,q)$ is 
    $$(\tilde{a},\tilde{b})=-\left(\frac{2Ap+Bq+D}{Dp+Eq+2F}, \frac{2Cq+Bp+E}{Dp+Eq+2F}\right),$$
and $(a,b)=\alpha^{-1}(\tilde{a},\tilde{b})$ becomes a point on the conic $\tilde{f}_u^{(d,k)}(x,y)=0$. 

On the other hand, the conic $f_u^{(d,k)}(x,y)=0$ passes through four points
    $$(0,1),\quad  (1,0), \quad \left(0,\frac{1}{d+1}\right), \quad \left(\frac{-1}{kd-1},0\right),$$
independent of $u$, and $f_{2k-d}^{(d,k)}(x,y)=0$ further passes through $\left(\frac{-2}{d^2-d-2},\frac{-d}{d^2-d-2}\right)$ by Lemma \ref{lem-Conic}.
Therefore, by applying the above procedure to these choices of $(p,q)$, we can find (at least) four explicit points $(a,b)$. Furthermore, Lemma~\ref{lem-curve} (1) implies that the tangent lines to the curve $\alpha^{-1}\circ \tilde{\gamma}$ at these points $(a,b)$ coincide with the boundary lines of the half-planes $H_{\alpha(p,q)}$, that is, 
    $$\{(x,y): (p(1-d)+q)x+((1-d)q+p)y=1/(d+1)\}.$$
Consequently, we obtain at least \emph{four points together with their corresponding tangent lines} on
$\alpha^{-1} \circ \tilde{\gamma}$, which altogether uniquely determine the conic $\tilde{f}_u^{(d,k)}(x,y)=0$ for each $d$ and $k$. We summarize all the computations in three \cref{tab-ellipse,tab-ellipse1,tab-ellipse2}: one for general $u\in (0,k)$ and two for the specific values $u=k-1$ and $u=2k-d$, respectively.

\begin{remark} \label{rmk-ellipse}
The four tangent lines listed in \cref{tab-ellipse} encloses a parallelogram with vertices
\begin{center}
    $\big(\frac{-2}{d^2-d-2},\frac{-d}{d^2-d-2}\big), \quad \big(\frac{-1}{d^2-d-2},\frac{-1}{d^2-d-2}\big), \quad \big(\frac{kd-k-1}{d^2-d-2},\frac{k-1}{d^2-d-2}\big), \quad \big(\frac{kd-k-2}{d^2-d-2},\frac{-d+k}{d^2-d-2}\big).$
\end{center}
Furthermore, if $0<u<k$, one verifies that the four points $(a,b)$ listed in \cref{tab-ellipse} lie on the four distinct sides of this parallelogram. Therefore, the conic $\tilde{f}_{u}^{(d,k)}(x,y)=0$ defines an \textit{ellipse} inscribed in the above parallelogram whenever $0<u<k$. 
\end{remark}

\begin{table}[h!]
  \begin{center}
    \caption{Four canonical points $(a,b)$ on $\tilde{f}_{u}^{(d,k)}(x,y)=0$ and the corresponding tangent lines}
    \label{tab-ellipse}
    \begin{tabular}{c|c|c} % <-- Alignments: 1st column left, 2nd middle and 3rd right, with vertical lines in between
       %$\backslash$ 
\hline
    $(p,q)$ & $(a,b)$ & Tangent at $(a,b)$\\
\hline\hline
    $\left(1,0\right)$ &  $\left(\frac{-2k+u}{k(d^2-d-2)}, \frac{-kd+du-u}{k(d^2-d-2)}\right)$ & $(1-d)x+y=\frac{1}{d+1}$\\
\hline
    $\left(0,1\right)$ &  $\left(\frac{kd-1-du-k+u}{d^2-d-2}, \frac{k-u-1}{d^2-d-2}\right)$ & $x+(1-d)y=\frac{1}{d+1}$\\
\hline
    $\left(0,\frac{1}{d+1}\right)$ &  $\left(\frac{-2+du-u}{d^2-d-2}, \frac{-d+u}{d^2-d-2}\right)$ & $x+(1-d)y=1$\\
\hline
    $\left(-\frac{1}{kd-1},0\right)$ &  $\left(\frac{k^2d-k-k^2-u}{k(d^2-d-2)}, \frac{k^2-k-du+u}{k(d^2-d-2)}\right)$ & $(1-d)x+y=-\frac{kd-1}{d+1}$\\
\hline
    \end{tabular}
  \end{center}
\end{table}

\begin{table}[h!]
  \begin{center}
    \caption{Four canonical points $(a,b)$ on $\tilde{f}_{k-1}^{(d,k)}(x,y)=0$ and the corresponding tangent lines}
    \label{tab-ellipse1}
    \begin{tabular}{c|c|c} % <-- Alignments: 1st column left, 2nd middle and 3rd right, with vertical lines in between
       %$\backslash$ 
\hline
    $(p,q)$ & $(a,b)$ & Tangent at $(a,b)$\\
\hline \hline
    $\left(1,0\right)$ &  $\left(\frac{-k-1}{k(d^2-d-2)}, \frac{-d-k+1}{k(d^2-d-2)}\right)$ & $(1-d)x+y=\frac{1}{d+1}$\\
\hline
    $\left(0,1\right)$ &  $\left(\frac{1}{d+1}, 0 \right)$ & $x+(1-d)y=\frac{1}{d+1}$\\
\hline
    $\left(0,\frac{1}{d+1}\right)$ &  $\left(\frac{kd-d-k-1}{d^2-d-2}, \frac{-d+k-1}{d^2-d-2}\right)$ & $x+(1-d)y=1$\\
\hline
    $\left(-\frac{1}{kd-1},0\right)$ &  $\left(\frac{k^2d-2k-k^2+1}{k(d^2-d-2)}, \frac{(k-1)(k-d+1)}{k(d^2-d-2)}\right)$ & $(1-d)x+y=-\frac{kd-1}{d+1}$\\
\hline
    \end{tabular}
  \end{center}
\end{table}

\begin{table}[h!]
  \begin{center}
    \caption{Five canonical points $(a,b)$ on $\tilde{f}_{2k-d}^{(d,k)}(x,y)=0$ and the corresponding tangent lines}
    \label{tab-ellipse2}
    \begin{tabular}{c|c|c} % <-- Alignments: 1st column left, 2nd middle and 3rd right, with vertical lines in between
       %$\backslash$ 
\hline
    $(p,q)$ & $(a,b)$ & Tangent at $(a,b)$\\
\hline \hline
    $\left(1,0\right)$ &  $\left(\frac{-d}{k(d^2-d-2)}, \frac{kd-d^2-2k+d}{k(d^2-d-2)}\right)$ & $(1-d)x+y=\frac{1}{d+1}$\\
\hline
    $\left(0,1\right)$ &  $\left(\frac{-kd+d^2+k-d-1}{d^2-d-2}, \frac{d-k-1}{d^2-d-2}\right)$ & $x+(1-d)y=\frac{1}{d+1}$\\
\hline
    $\left(0,\frac{1}{d+1}\right)$ &  $\left(\frac{-2+2kd-d^2-2k+d}{d^2-d-2}, \frac{2k-2d}{d^2-d-2}\right)$ & $x+(1-d)y=1$\\
\hline
    $\left(-\frac{1}{kd-1},0\right)$ &  $\left(\frac{k^2d-3k-k^2+d}{k(d^2-d-2)}, \frac{(k-d)(k-d+1)}{k(d^2-d-2)}\right)$ & $(1-d)x+y=-\frac{kd-1}{d+1}$\\
\hline
    $\left(\frac{-2}{d^2-d-2},\frac{-d}{d^2-d-2}\right)$ &  $\left(\frac{d}{3d-2k}, \frac{2d-2k}{(d+1)(3d-2k)}\right)$ & $x+(1+d)y=1$\\
\hline
    \end{tabular}
  \end{center}
\end{table}

\begin{lemma} \label{lem-SN-line}
The line $l_1(x,y)=0$ introduced in Eq. \eqref{eq-Symp-line1} passes through two points
    $$\left(\frac{kd-d-k-1}{d^2-d-2}, \frac{-d+k-1}{d^2-d-2}\right), \quad \left(\frac{k^2d-2k-k^2+1}{k(d^2-d-2)}, \frac{(k-1)(k-d+1)}{k(d^2-d-2)}\right)$$
which appear as the points $(a,b)$ in \cref{tab-ellipse1}. Similarly, the line $l_2(x,y)=0$ defined in Eq. \eqref{eq-Symp-line2} passes through two points
    $$\left(\frac{k^2d-3k-k^2+d}{k(d^2-d-2)}, \frac{(k-d)(k-d+1)}{k(d^2-d-2)}\right), \quad \left(\frac{d}{3d-2k}, \frac{2d-2k}{(d+1)(3d-2k)}\right)$$
as listed in \cref{tab-ellipse2}.
\end{lemma}
\begin{proof}
This is straightforward.
\end{proof}

We are now ready to prove Theorem \ref{thm-SympSch}

\begin{proof}[\textbf{Proof of Theorem \ref{thm-SympSch}}]
In the case $1<k\leq \frac{d}{2}$ with $k$ odd, we have
    $${\rm ext}(\mathbb{P}_k)=\left\{(1,0), \left(0,\frac{1}{d+1}\right), \left(-\frac{1}{kd-1},0\right),\left(\frac{-1}{kd-k-1},\frac{-k}{kd-k-1}\right)\right\}\,\cup\, \{\gamma_1(t):t\in [0,1]\}$$
by Theorem \ref{thm-SympkPos} (2), where $\gamma_1:(-\eps,1+\eps) \to \Real^2$ is a smooth curve satisfying 
\begin{itemize}
    \item $f_{k-1}^{(d,k)}(\gamma_1(t))\equiv 0$ and the image of $\gamma_1\big|_{[0,1]}$ lies entirely in the second quadrant;
    
    \item $\gamma_1(0)=(0,\frac{1}{d+1})$ and $\gamma_1(1)=(-\frac{1}{kd-1},0)$.
\end{itemize}
Eq. \eqref{eq-SNPolar} then implies that the set $\alpha(\mathbb{S}_k)$ is given by
    $$H_{1,0}\cap  H_{0,\frac{1}{d+1}} \cap H_{\frac{-1}{kd-1},0}\cap H_{\frac{-1}{kd-k-1}, \frac{-k}{kd-k-1}}\cap \bigcap_{t\in [0,1]}H_{\gamma_1(t)},$$
and by the previous discussion, there exists a quadratic polynomial $g_1(x,y):=\tilde{f}_{k-1}^{(d,k)}(x,y)$ such that 
    $$g_1\left ((\alpha^{-1} \circ \tilde{\gamma}_1)(t)\right )\equiv 0.$$ 
and $g_1(x,y)\leq 0$ becomes a filled ellipse, by Remark \ref{rmk-ellipse} and by modifying the sign if necessary. Thanks to Lemma \ref{lem-curve} (2), the set $\alpha^{-1}\big(\bigcap_{t\in [0,1]} H_{\gamma_1(t)}\big)$ is the largest convex region containing the origin and bounded by two lines 
    $$x+(1-d)y=1 \text{\quad and \quad} (1-d)x+y=-\frac{kd-1}{d+1},$$
which are tangent to $\alpha^{-1}\circ \tilde{\gamma}_1$ at 
\begin{center}
    $\alpha^{-1}(\tilde{\gamma}_1(0))=\left(\frac{kd-d-k-1}{d^2-d-2}, \frac{-d+k-1}{d^2-d-2}\right)$ and $\alpha^{-1}(\tilde{\gamma}_1(1))=\left(\frac{k^2d-2k-k^2+1}{k(d^2-d-2)}, \frac{(k-1)(k-d+1)}{k(d^2-d-2)}\right)$,
\end{center}
respectively (\cref{tab-ellipse1}), and further bounded by the ellipse segment $\alpha^{-1}\circ \tilde{\gamma}_1\big|_{[0,1]}$. Note that this set is a union of a region described by three inequalities
    $$x+(1-d)y\le 1, \quad (1-d)x+y\geq -\frac{kd-1}{d+1}, \quad l_1(x,y)\leq 0,$$
and another region from $l_1(x,y)\geq 0$ and $g_1(x,y)\leq 0$, as the equation $l_1(x,y)=0$ describes the line connecting $\alpha^{-1}(\tilde{\gamma}_1(0))$ and $\alpha^{-1}(\tilde{\gamma}_1(1))$ from the definition \eqref{eq-Symp-line1}.

Putting everything together, we obtain the description of the Schmidt number region $\mathbb{S}_k$ given by the system of inequalities in
Eq.~\eqref{eq-Symp-SN1o}, which completes the proof of part~(2).

The remaining proofs of Theorem \ref{thm-SympSch} (3) and (4), i.e., the casees $\frac{2}{d}<k<d$, is analogous to that of (2), so we only provide the sketch of the proof. Indeed, we can proceed with another smooth curve $\gamma_2:(-\eps,1+\eps)\to \Real^2$ satisfying
\begin{itemize}
    \item $f_{2k-d}^{(d,k)}(\gamma_2(t))\equiv 0$ and the image of $\gamma_2\bigm|_{[0,1]}$ lies entirely in the third quadrant;

    \item $\gamma_2(0)=(-\frac{1}{kd-1},0)$ and $\gamma_2(1) = (-\frac{2}{d^2-d-2},-\frac{d}{d^2-d-2})$,
\end{itemize}
and choose $g_2(x,y):=\tilde{f}_{2k-d}^{(d,k)}(x,y)$. Since
    $${\rm ext}(\mathbb{P}_k)=\begin{cases}
        \{(1,0), \left(0,\frac{1}{d+1}\right), \left(-\frac{1}{kd-1},0\right),\left(\frac{-1}{kd-k-1},\frac{-k}{kd-k-1}\right)\}\cup \{\gamma_2(t)\}_{t\in [0,1]} & \text{in the case (3),}\\
        \{(1,0), \left(0,\frac{1}{d+1}\right), \left(-\frac{1}{kd-1},0\right),\left(\frac{-1}{kd-k-1},\frac{-k}{kd-k-1}\right)\}\\
        \qquad\qquad\qquad\qquad\qquad\qquad\qquad\quad \cup \{\gamma_1(t)\}_{t\in [0,1]} \cup \{\gamma_2(t)\}_{t\in [0,1]}& \text{in the case (4),}
    \end{cases}$$
    the same proof, combined with Lemma \ref{lem-SN-line}, works with a minor modification.
\end{proof}

\bibliography{references}

\newcommand{\etalchar}[1]{$^{#1}$}
\def\cprime{$'$}
\begin{thebibliography}{ADMH{\etalchar{+}}24}

\bibitem[ADMH{\etalchar{+}}24]{aubrun2024completely}
Guillaume Aubrun, Kenneth~R Davidson, Alexander M{\"u}ller-Hermes, Vern~I Paulsen, and Mizanur Rahaman.
\newblock Completely bounded norms of k k-positive maps.
\newblock {\em Journal of the London Mathematical Society}, 109(6):e12936, 2024.

\bibitem[BCH{\etalchar{+}}02]{BCH+02}
Dagmar Bru{\ss}, J~Ignacio Cirac, Pawel Horodecki, Florian Hulpke, Barbara Kraus, Maciej Lewenstein, and Anna Sanpera.
\newblock Reflections upon separability and distillability.
\newblock {\em Journal of Modern Optics}, 49(8):1399--1418, 2002.

\bibitem[BCHW15]{BCHW15}
Stefan B{\"a}uml, Matthias Christandl, Karol Horodecki, and Andreas Winter.
\newblock Limitations on quantum key repeaters.
\newblock {\em Nature communications}, 6(1):6908, 2015.

\bibitem[BD11]{BD11}
Francesco Buscemi and Nilanjana Datta.
\newblock Entanglement cost in practical scenarios.
\newblock {\em Phys. Rev. Lett.}, 106:130503, Mar 2011.

\bibitem[BK12]{TextBK}
Egbert Brieskorn and Horst Kn{\"o}rrer.
\newblock {\em Plane Algebraic Curves: Translated by John Stillwell}.
\newblock Springer Science \& Business Media, 2012.

\bibitem[Bra37]{Bra37}
Richard Brauer.
\newblock On algebras which are connected with the semisimple continuous groups.
\newblock {\em Annals of Mathematics}, 38(4):857--872, 1937.

\bibitem[Bre06]{Bre06}
Heinz-Peter Breuer.
\newblock Optimal entanglement criterion for mixed quantum states.
\newblock {\em Physical review letters}, 97(8):080501, 2006.

\bibitem[Car20]{Car20}
Daniel Cariello.
\newblock Inequalities for the {S}chmidt number of bipartite states.
\newblock {\em Lett. Math. Phys.}, 110(4):827--833, 2020.

\bibitem[CCF25]{CCF25}
Qian Chen, Benoît Collins, and Omar Fawzi.
\newblock Symmetry reduction for testing k-block-positivity via extendibility.
\newblock {\em Journal of Physics A: Mathematical and Theoretical}, 58(48):485302, nov 2025.

\bibitem[CF17]{CF17}
Matthias Christandl and Roberto Ferrara.
\newblock Private states, quantum data hiding, and the swapping of perfect secrecy.
\newblock {\em Physical review letters}, 119(22):220506, 2017.

\bibitem[Cho75]{Cho75a}
Man~Duen Choi.
\newblock Completely positive linear maps on complex matrices.
\newblock {\em Linear Algebra Appl.}, 10:285--290, 1975.

\bibitem[Cho82]{Cho82}
Man~Duen Choi.
\newblock Positive linear-maps.
\newblock In {\em Proceedings of Symposia in Pure Mathematics}, volume~38, pages 583--590. AMER MATHEMATICAL SOC 201 CHARLES ST, PROVIDENCE, RI 02940-2213, 1982.

\bibitem[Chr12]{PPTsq}
M.~Christandl.
\newblock P{PT} square conjecture.
\newblock {\em Banff International Research Station Workshop: \emph{Operator Structures in Quantum Information Theory}}, 2012.

\bibitem[CK06]{CK06}
Dariusz Chru{\'s}ci{\'n}ski and Andrzej Kossakowski.
\newblock Class of positive partial transposition states.
\newblock {\em Physical Review A—Atomic, Molecular, and Optical Physics}, 74(2):022308, 2006.

\bibitem[CK08]{CK08}
Dariusz Chru{\'s}ci{\'n}ski and Andrzej Kossakowski.
\newblock A class of positive atomic maps.
\newblock {\em Journal of Physics A: Mathematical and Theoretical}, 41(21):215201, 2008.

\bibitem[CK09]{CK09}
Dariusz Chru\'{s}ci\'{n}ski and Andrzej Kossakowski.
\newblock Spectral conditions for positive maps.
\newblock {\em Comm. Math. Phys.}, 290(3):1051--1064, 2009.

\bibitem[CMHW19]{CMHW19}
Matthias Christandl, Alexander M\"{u}ller-Hermes, and Michael~M. Wolf.
\newblock When do composed maps become entanglement breaking?
\newblock {\em Ann. Henri Poincar\'{e}}, 20(7):2295--2322, 2019.

\bibitem[Cox03]{Cox03}
H.S.M. Coxeter.
\newblock {\em Projective Geometry}.
\newblock Springer, 10 2003.

\bibitem[CYT17]{CYT17}
Lin Chen, Yu~Yang, and Wai-Shing Tang.
\newblock Schmidt number of bipartite and multipartite states under local projections.
\newblock {\em Quantum Inf. Process.}, 16(3):Paper No. 75, 27, 2017.

\bibitem[CYT19]{Chen2019}
Lin Chen, Yu~Yang, and Wai-Shing Tang.
\newblock Positive-partial-transpose square conjecture for $n=3$.
\newblock {\em Phys. Rev. A}, 99:012337, Jan 2019.

\bibitem[CYZ18]{CYZ18}
Beno\^{i}t Collins, Zhi Yin, and Ping Zhong.
\newblock The {PPT} square conjecture holds generically for some classes of independent states.
\newblock {\em J. Phys. A}, 51(42):425301, 19, 2018.

\bibitem[ES13]{ES13}
Christopher Eltschka and Jens Siewert.
\newblock Negativity as an estimator of entanglement dimension.
\newblock {\em Physical Review Letters}, 111(10):100503, 2013.

\bibitem[EW01]{EW01}
T.~Eggeling and R.~F. Werner.
\newblock Separability properties of tripartite states with ${U}\ensuremath{\bigotimes}{U}\ensuremath{\bigotimes}{U}$ symmetry.
\newblock {\em Phys. Rev. A}, 63:042111, Mar 2001.

\bibitem[Fol16]{Fol16}
Gerald~B. Folland.
\newblock {\em A course in abstract harmonic analysis}.
\newblock Textbooks in Mathematics. CRC Press, Boca Raton, FL, second edition, 2016.

\bibitem[GHP10]{GHP10}
Andrzej Grudka, Micha\l Horodecki, and \L~ukasz Pankowski.
\newblock Constructive counterexamples to the additivity of the minimum output {R}\'{e}nyi entropy of quantum channels for all {$p>2$}.
\newblock {\em J. Phys. A}, 43(42):425304, 7, 2010.

\bibitem[GKS21]{GKS21}
Mark Girard, Seung-Hyeok Kye, and Erling Størmer.
\newblock Convex cones in mapping spaces between matrix algebras.
\newblock {\em Linear Algebra and its Applications}, 608:248--269, 2021.

\bibitem[GNP25]{GNP25+}
Aabhas Gulati, Ion Nechita, and Sang-Jun Park.
\newblock Positive maps and extendibility hierarchies from copositive matrices.
\newblock {\em arXiv preprint arXiv:2509.15201}, 2025.

\bibitem[GNS25]{GNS25}
Aabhas Gulati, Ion Nechita, and Satvik Singh.
\newblock Entanglement in cyclic sign invariant quantum states.
\newblock {\em Journal of Mathematical Physics}, 66(12):122202, 12 2025.

\bibitem[Gur03]{Gur03}
Leonid Gurvits.
\newblock Classical deterministic complexity of {E}dmond's problem and quantum entanglement.
\newblock In {\em Proceedings of the {T}hirty-{F}ifth {A}nnual {ACM} {S}ymposium on {T}heory of {C}omputing}, pages 10--19. ACM, New York, 2003.

\bibitem[GW09]{TextGoodman}
Roe Goodman and Nolan~R Wallach.
\newblock {\em Symmetry, representations, and invariants}, volume 255.
\newblock Springer, 2009.

\bibitem[Hal06]{Hal06}
William Hall.
\newblock A new criterion for indecomposability of positive maps.
\newblock {\em Journal of Physics A: Mathematical and General}, 39(45):14119, 2006.

\bibitem[HH99]{HH99}
Micha\l{} Horodecki and Pawe\l{} Horodecki.
\newblock Reduction criterion of separability and limits for a class of distillation protocols.
\newblock {\em Phys. Rev. A}, 59:4206--4216, Jun 1999.

\bibitem[HHH96]{HHH96}
Micha\l Horodecki, Pawe\l Horodecki, and Ryszard Horodecki.
\newblock Separability of mixed states: necessary and sufficient conditions.
\newblock {\em Phys. Lett. A}, 223(1-2):1--8, 1996.

\bibitem[HHH98]{HHH98}
Micha{\l} Horodecki, Pawe{\l} Horodecki, and Ryszard Horodecki.
\newblock Mixed-state entanglement and distillation: Is there a “bound” entanglement in nature?
\newblock {\em Physical Review Letters}, 80(24):5239, 1998.

\bibitem[HHHH09]{HHHH09}
Ryszard Horodecki, Pawe\l Horodecki, Micha\l Horodecki, and Karol Horodecki.
\newblock Quantum entanglement.
\newblock {\em Rev. Modern Phys.}, 81(2):865--942, 2009.

\bibitem[HJ12]{TextHJ12}
Roger~A Horn and Charles~R Johnson.
\newblock {\em Matrix analysis}.
\newblock Cambridge university press, 2012.

\bibitem[HK25a]{HK25a}
Kyung~Hoon Han and Seung-Hyeok Kye.
\newblock Global locations of schmidt-number witnesses.
\newblock {\em Physical Review A}, 112(3):032426, 2025.

\bibitem[HK25b]{HK25b}
Kyung~Hoon Han and Seung-Hyeok Kye.
\newblock Supporting hyperplanes for schmidt numbers and schmidt number witnesses.
\newblock {\em Open Systems \& Information Dynamics}, 32(02):2550008, 2025.

\bibitem[HLLMH18]{HLLMH18}
Marcus Huber, Ludovico Lami, C\'ecilia Lancien, and Alexander M\"uller-Hermes.
\newblock High-dimensional entanglement in states with positive partial transposition.
\newblock {\em Phys. Rev. Lett.}, 121:200503, Nov 2018.

\bibitem[Hol93]{holevo1993note}
Alexander~S Holevo.
\newblock A note on covariant dynamical semigroups.
\newblock {\em Reports on mathematical physics}, 32(2):211--216, 1993.

\bibitem[Hor97]{Hor97}
Pawe\l Horodecki.
\newblock Separability criterion and inseparable mixed states with positive partial transposition.
\newblock {\em Phys. Lett. A}, 232(5):333--339, 1997.

\bibitem[HRF20]{hanson2020eventually}
Eric~P Hanson, Cambyse Rouz{\'e}, and Daniel~Stilck Fran{\c{c}}a.
\newblock Eventually entanglement breaking markovian dynamics: Structure and characteristic times.
\newblock {\em Ann. Henri Poincar{\'e}}, 21:1517--1571, 2020.

\bibitem[HR{\.Z}22]{HRZ22}
Pawe{\l} Horodecki, {\L}ukasz Rudnicki, and Karol {\.Z}yczkowski.
\newblock Five open problems in quantum information theory.
\newblock {\em PRX Quantum}, 3(1):010101, 2022.

\bibitem[Jk72]{Jam72}
A.~Jamio\l~kowski.
\newblock Linear transformations which preserve trace and positive semidefiniteness of operators.
\newblock {\em Rep. Mathematical Phys.}, 3(4):275--278, 1972.

\bibitem[KG24]{KG24}
Robin Krebs and Mariami Gachechiladze.
\newblock High schmidt number concentration in quantum bound entangled states.
\newblock {\em Physical Review Letters}, 132(22):220203, 2024.

\bibitem[KG25]{KG25+}
Robin Krebs and Mariami Gachechiladze.
\newblock Scaling bound entanglement through local extensions.
\newblock {\em arXiv preprint arXiv:2509.07086}, 2025.

\bibitem[KMP17]{Kennedy2017}
Matthew Kennedy, Nicholas Manor, and Vern Paulsen.
\newblock Composition of {PPT} maps.
\newblock {\em Quantum Information and Computation}, 18, 10 2017.

\bibitem[Kye23]{Kye23}
Seung-Hyeok Kye.
\newblock Compositions and tensor products of linear maps between matrix algebras.
\newblock {\em Linear Algebra Appl.}, 658:283--309, 2023.

\bibitem[LG15]{Lami2015entanglebreak}
L.~Lami and V.~Giovannetti.
\newblock Entanglement–breaking indices.
\newblock {\em Journal of Mathematical Physics}, 56(9):092201, Sep 2015.

\bibitem[LHF25]{LHF25}
Nicky Kai~Hong Li, Marcus Huber, and Nicolai Friis.
\newblock High-dimensional entanglement witnessed by correlations in arbitrary bases.
\newblock {\em npj Quantum Information}, 11(1):50, 2025.

\bibitem[LKCH00]{LKCH00}
Maciej Lewenstein, Barabara Kraus, J~Ignacio Cirac, and Pawel Horodecki.
\newblock Optimization of entanglement witnesses.
\newblock {\em Physical Review A}, 62(5):052310, 2000.

\bibitem[LY22]{LY22}
Hun~Hee Lee and Sang-Gyun Youn.
\newblock Quantum channels with quantum group symmetry.
\newblock {\em Communications in Mathematical Physics}, 389(3):1303--1329, 2022.

\bibitem[MGM24]{MGM24+}
Bivas Mallick, Nirman Ganguly, and AS~Majumdar.
\newblock On the characterization of schmidt number breaking and annihilating channels.
\newblock {\em arXiv preprint arXiv:2411.19315}, 2024.

\bibitem[MOM25]{MOM25}
Tomasz M{\l}ynik, Hiroyuki Osaka, and Marcin Marciniak.
\newblock Characterization of k-positive maps.
\newblock {\em Communications in Mathematical Physics}, 406(3):62, 2025.

\bibitem[MSD17]{MSD17}
Marek Mozrzymas, Micha{\l} Studzi{\'n}ski, and Nilanjana Datta.
\newblock Structure of irreducibly covariant quantum channels for finite groups.
\newblock {\em Journal of Mathematical Physics}, 58(5), 2017.

\bibitem[NC00]{NiCh}
Michael~A. Nielsen and Isaac~L. Chuang.
\newblock {\em Quantum computation and quantum information}.
\newblock Cambridge University Press, Cambridge, 2000.

\bibitem[NP25]{NP25}
Ion Nechita and Sang-Jun Park.
\newblock Random covariant quantum channels.
\newblock In {\em Annales Henri Poincar{\'e}}, pages 1--61. Springer, 2025.

\bibitem[Pau02]{paulsen2002completely}
Vern Paulsen.
\newblock {\em Completely bounded maps and operator algebras}, volume~78.
\newblock Cambridge University Press, 2002.

\bibitem[Per96]{Per96}
Asher Peres.
\newblock Separability criterion for density matrices.
\newblock {\em Phys. Rev. Lett.}, 77(8):1413--1415, 1996.

\bibitem[PJPY24]{PJPY24}
Sang-Jun Park, Yeong-Gwang Jung, Jeongeun Park, and Sang-Gyun Youn.
\newblock A universal framework for entanglement detection under group symmetry.
\newblock {\em Journal of Physics A: Mathematical and Theoretical}, 57(32):325304, 2024.

\bibitem[PV19]{PV19}
K\'aroly~F. P\'al and Tam\'as V\'ertesi.
\newblock Class of genuinely high-dimensionally-entangled states with a positive partial transpose.
\newblock {\em Phys. Rev. A}, 100:012310, Jul 2019.

\bibitem[PY24]{PY24}
Sang-Jun Park and Sang-Gyun Youn.
\newblock k-positivity and schmidt number under orthogonal group symmetries.
\newblock {\em Quantum Information Processing}, 23(5):162, 2024.

\bibitem[RBO20]{BO20}
BV~Rajarama~Bhat and Hiroyuki Osaka.
\newblock A factorization property of positive maps on c*-algebras.
\newblock {\em International Journal of Quantum Information}, 18(05):2050019, 2020.

\bibitem[RJP18]{Rahaman2018}
Mizanur Rahaman, Samuel Jaques, and Vern Paulsen.
\newblock Eventually entanglement breaking maps.
\newblock {\em Journal of Mathematical Physics}, 59, 01 2018.

\bibitem[SBL01]{SBL01}
Anna Sanpera, Dagmar Bru{\ss}, and Maciej Lewenstein.
\newblock Schmidt-number witnesses and bound entanglement.
\newblock {\em Phys. Rev. A}, 63:050301, Apr 2001.

\bibitem[Sim95]{Simon95}
Barry Simon.
\newblock {\em Representations of Finite and Compact Groups (Graduate Studies in Mathematics ; V. 10)}.
\newblock American Mathematical Society, 12 1995.

\bibitem[Sko11]{Sko11}
\L~ukasz Skowronek.
\newblock Cones with a mapping cone symmetry in the finite-dimensional case.
\newblock {\em Linear Algebra Appl.}, 435(2):361--370, 2011.

\bibitem[SN21]{SN21}
Satvik Singh and Ion Nechita.
\newblock Diagonal unitary and orthogonal symmetries in quantum theory.
\newblock {\em {Quantum}}, 5:519, August 2021.

\bibitem[SN22]{singh2022ppt}
Satvik Singh and Ion Nechita.
\newblock The {PPT}$^2$ conjecture holds for all {C}hoi-type maps.
\newblock {\em Annales Henri Poincar{\'e}}, 23(9):3311--3329, 2022.

\bibitem[SP18]{SP18}
Enrico Sindici and Marco Piani.
\newblock Simple class of bound entangled states based on the properties of the antisymmetric subspace.
\newblock {\em Phys. Rev. A}, 97:032319, Mar 2018.

\bibitem[Sr82]{Sto82}
Erling St\o~rmer.
\newblock Decomposable positive maps on {$C^{\ast} $}-algebras.
\newblock {\em Proc. Amer. Math. Soc.}, 86(3):402--404, 1982.

\bibitem[Sr86]{Sto86}
Erling St\o~rmer.
\newblock Extension of positive maps into {$B({\mathcal{H}})$}.
\newblock {\em J. Funct. Anal.}, 66(2):235--254, 1986.

\bibitem[SSrZ09]{SSZ09}
\L~ukasz Skowronek, Erling St\o~rmer, and Karol \.{Z}yczkowski.
\newblock Cones of positive maps and their duality relations.
\newblock {\em J. Math. Phys.}, 50(6):062106, 18, 2009.

\bibitem[St{\o}63]{stormer1963positive}
Erling St{\o}rmer.
\newblock Positive linear maps of operator algebras.
\newblock {\em Acta Mathematica}, 110(1):233--278, 1963.

\bibitem[St{\o}13]{TextSto}
Erling St{\o}rmer.
\newblock {\em Positive linear maps of operator algebras}.
\newblock Springer, 2013.

\bibitem[SWZ11]{SWZ11}
Stanis\l aw~J. Szarek, Elisabeth Werner, and Karol \.{Z}yczkowski.
\newblock How often is a random quantum state {$k$}-entangled?
\newblock {\em J. Phys. A}, 44(4):045303, 15, 2011.

\bibitem[Ter01]{Ter01}
Barbara~M Terhal.
\newblock A family of indecomposable positive linear maps based on entangled quantum states.
\newblock {\em Linear Algebra and its Applications}, 323(1-3):61--73, 2001.

\bibitem[TG09]{TG09}
G{\'e}za T{\'o}th and Otfried G{\"u}hne.
\newblock Entanglement and permutational symmetry.
\newblock {\em Physical review letters}, 102(17):170503, 2009.

\bibitem[TH00]{TH00}
Barbara~M. Terhal and Pawe\l{} Horodecki.
\newblock Schmidt number for density matrices.
\newblock {\em Phys. Rev. A (3)}, 61(4):040301, 4, 2000.

\bibitem[Tom85]{Tom85}
Jun Tomiyama.
\newblock On the geometry of positive maps in matrix algebras. {II}.
\newblock {\em Linear Algebra Appl.}, 69:169--177, 1985.

\bibitem[vEKD25]{EKD25}
Frederik vom Ende, Sumeet Khatri, and Sergey Denisov.
\newblock k-positive maps: New characterizations and a generation method.
\newblock {\em Open Systems \& Information Dynamics}, 32(04):2550015, 2025.

\bibitem[VW01]{VW01}
K.~G.~H. Vollbrecht and R.~F. Werner.
\newblock Entanglement measures under symmetry.
\newblock {\em Phys. Rev. A}, 64:062307, Nov 2001.

\bibitem[Wat18]{TextWat18}
John Watrous.
\newblock {\em The theory of quantum information}.
\newblock Cambridge university press, 2018.

\bibitem[Wen88]{Wen88}
Hans Wenzl.
\newblock On the structure of brauer's centralizer algebras.
\newblock {\em Annals of Mathematics}, 128(1):173--193, 1988.

\bibitem[Wer89]{Wer89}
Reinhard~F. Werner.
\newblock Quantum states with einstein-podolsky-rosen correlations admitting a hidden-variable model.
\newblock {\em Phys. Rev. A}, 40:4277--4281, Oct 1989.

\bibitem[Wey46]{TextWeyl}
Hermann Weyl.
\newblock {\em The classical groups: their invariants and representations}, volume~1.
\newblock Princeton university press, 1946.

\bibitem[Wor76]{Wor76}
S.~L. Woronowicz.
\newblock Positive maps of low dimensional matrix algebras.
\newblock {\em Rep. Math. Phys.}, 10(2):165--183, 1976.

\bibitem[YLT16]{YLT16}
Yu~Yang, Denny~H. Leung, and Wai-Shing Tang.
\newblock All 2-positive linear maps from {$M_3(\Bbb{C})$} to {$M_3(\Bbb{C})$} are decomposable.
\newblock {\em Linear Algebra Appl.}, 503:233--247, 2016.

\bibitem[Yu16]{Yu16}
Nengkun Yu.
\newblock Separability of a mixture of dicke states.
\newblock {\em Physical Review A}, 94(6):060101, 2016.

\end{thebibliography}
\bibliographystyle{alpha}
\bigskip
\hrule
\bigskip

\end{document}